\documentclass[11pt,english, hidelinks]{article}
\usepackage{geometry}
\usepackage{setspace}
\usepackage{authblk}
\usepackage[utf8]{inputenc}
\usepackage{color,soul,yhmath}
\usepackage[dvipsnames]{xcolor}
\usepackage{algpseudocode}
\usepackage{bm}
\usepackage{graphicx}
\graphicspath{ {.} }
\usepackage{hyperref}
\usepackage{natbib}
\usepackage{bbm}
\usepackage{amsmath}
\usepackage{amsthm}
\usepackage{amssymb}
\usepackage[makeroom]{cancel}

\usepackage{subcaption}

\RequirePackage{cleveref}
\RequirePackage{enumitem}
\RequirePackage{algorithm}

\usepackage{rotating}

\usepackage{array,booktabs}
\newcommand{\ra}[1]{\renewcommand{\arraystretch}{#1}}

\DeclareMathOperator*{\argmin}{arg\,min}
\hypersetup{
    colorlinks=true,
    linkcolor=blue,
    citecolor=blue,
    filecolor=magenta,
    urlcolor=blue
}

\newtheorem{theorem}{Theorem}
\newtheorem{assumption}{Assumption}

\newtheorem{claim}[theorem]{Claim}

\newtheorem{lemma}{Lemma}

\newtheorem{remark}{Remark}

\newtheorem{assumptionalt}{Assumption}[assumption]
\newenvironment{assumptionp}[1]{
  \renewcommand\theassumptionalt{#1}
  \assumptionalt
}{\endassumptionalt}

\numberwithin{equation}{section}
\def\pmk{p_{\text{-}k}}

\def\mat{\textnormal{\textsc{mat}}}
\def\map{\textnormal{\textsc{map}}}
\def\reshape{\textnormal{\textsc{reshape}}}
\def\tr{\textnormal{\text{tr}}}
\def\wt#1{\widetilde{#1}}
\def\wh#1{\widehat{#1}}

\def\E{\mathbb{E}}

\def\P{\mathbb{P}}
\def\R{\mathbb{R}}

\def\cD{\mathcal{D}} 
\def\cE{\mathcal{E}} 
\def\cF{\mathcal{F}}
\def\cG{\mathcal{G}}  \def\cI{\mathcal{I}} 
\def\cII{\mathcal{II}} 
\def\cIII{\mathcal{III}} 
\def\cIV{\mathcal{IV}} 
 
\def\cL{\mathcal{L}}
\def\cM{\mathcal{M}} 
\def\cN{\mathcal{N}} 
\def\cO{\mathcal{O}}

\def\cS{\mathcal{S}}

\def\cX{\mathcal{X}} 
\def\cY{\mathcal{Y}}

\def\zero{\mathbf{0}} 
\def\bi{\boldsymbol{i}} 
\def\bI{\mathbf{I}} 
\def\be{\mathbf{e}} 
\def\bw{\mathbf{w}}
\def\bh{\mathbf{h}}

\def\bx{\mathbf{x}} 

\def\bA{\mathbf{A}}
\def\bB{\mathbf{B}}
\def\bD{\mathbf{D}}
\def\bX{\mathbf{X}} 
\def\bY{\mathbf{Y}} 
\def\bU{\mathbf{U}} 
 
\def\bC{\mathbf{C}} 
\def\bL{\mathbf{L}} 
 
\def\bH{\mathbf{H}}
\def\bM{\mathbf{M}} 
\def\bR{\mathbf{R}} 
\def\bP{\mathbf{P}} 
\def\bQ{\mathbf{Q}}

\def\bv{\boldsymbol{v}} 
\def\bg{\mathbf{g}}

\def\balpha{\boldsymbol{\alpha}}
\def\bbeta{\boldsymbol{\beta}}

   \def\bSigma{\boldsymbol{\Sigma}} \def\bOmega{\boldsymbol{\Omega}}

\def\trans{\intercal}

\def\cov{\mathbb{C}\mathrm{ov}} 
\def\var{\mathbb{V}\mathrm{ar}}

\def\vec{\mathrm{vec}} 
 
\def\col{\mathrm{col}}

\def\rank{\mathrm{rank}}

\def\diag{\mathrm{diag}}

\allowdisplaybreaks[1]

\begin{document}
\setlength{\parindent}{18pt}
    \onehalfspacing
	\begin{titlepage}
		
		\title{Identification and Estimation of Multi-order Tensor Factor Models}
		        \author{Zetai Cen\thanks{School of Mathematics, University of Bristol. Corresponding author. Email: \url{zetai.cen@bristol.ac.uk}.}}
		
		
		\date{}
		
		\maketitle

\begin{abstract}
We propose a novel framework in high-dimensional factor models to simultaneously analyze multiple tensor time series, each with potentially different tensor orders and dimensionality. The connection between different tensor time series is through their global factors that are correlated to each other. A salient feature of our model is that when all tensor time series have the same order, it can be regarded as an extension of multilevel factor models from vectors to general tensors. 
{Under an explicit full-rank condition on the cross-thread covariance and an orthogonality condition between the global and local loading spaces, we separate the global and local components.}
Parameter estimation is thoroughly discussed, including a consistent factor number estimator. With strong correlation between global factors and noise allowed, we derive the rates of convergence of our estimators, which can be superior to those of existing methods for multilevel factor models. We also develop estimators that are more computationally efficient, with rates of convergence spelt out. Extensive experiments are performed under various settings, corroborating with the pronounced theoretical results. As a real application example, we analyze a set of taxi data to study the traffic flow between Times Square and its neighboring areas.
\end{abstract}
		
		\bigskip
		\bigskip

		\noindent
		{\sl Keywords and phrases:}
        Dimension reduction; global factors; multilevel factor models; tensor order; tensor time series.

        \noindent
        {\sl MSC subject classification:}
        62H25, 62M10

\noindent

	\end{titlepage}
	
	\setcounter{page}{2}

\maketitle

\tableofcontents

\newpage
\section{Introduction}
\label{sec: introduction}

With the rapid development of technology over the past few decades, high-dimensional time series are becoming more ubiquitous. This constantly demands methodologies to facilitate interpretation and forecasting. Among many other methods, factor models are found very efficient in dimension reduction and helping researchers to understand the dependence structure in data sets. Credited to these advantages of factor models, they have been applied to various fields, including but not limited to psychology \citep{McCraeJohn1992}, biology \citep{Hirzeletal2002, Hochreiteretal2006}, economics and finance \citep{ChamberlainRothschild1983, FamaFrench1993, StockWatson2002}.

More recently, researchers have opened up to time series formed by matrices or, more generally, multidimensional arrays or tensors. Not only are such representations often the most natural by how the data is collected, they also provide useful insight that flattening the data might overlook. Thus, the literature has seen increasing work in matrix- or generally tensor-valued factor models. For matrix-valued time series, \cite{Wangetal2019} proposed matrix factor models by extending the framework in \cite{Lametal2011}. \cite{yu2022projected} also studied a matrix factor model and proposed a projection estimator with faster rate of convergence than the PCA-type (principal component analysis) estimator. \cite{ChenFan2023} proposed $\alpha$-PCA by generalizing the method for vector factor models by \cite{Bai2003}. Addressing the interplay between rows and columns, \cite{LamCen2025_MEFM} proposed a main effect factor model, and \cite{Yuanetal2023} studied a two-way dynamic factor model. Methodologies are also ample for tensor factor models since the early work by \cite{Chenetal2022}, see \text{e.g.} \cite{barigozzi2022statistical}, \cite{Hanetal2022}, and \cite{Hanetal2024}.

Despite the vast amount of approaches in tensor factor modeling discussed above, all of them consider a single time series with tensor observation of the same order and dimensions.
Formally, for a single series of order-$K$ tensors $\{\cX_t\}$, a Tucker-type tensor factor model decomposes each $\cX_t$ as $\cX_t = \cF_t \times_{k=1}^K \bA_k + \cE_t$, where the latent core factor $\cF_t$ drives the dynamics of the observation through the factor loading matrices $\bA_k, \dots, \bA_K$.
This prevents us from conducting a comprehensive study on data sets containing multiple tensor time series, or to investigate the dependence structure among several tensor time series. In fact, this can be limited even when facing two vector time series that are closely connected in nature for panel data, yet of different dimensions. 

To tackle this problem, for vector time series, \cite{Wang2008} proposed a multilevel factor model, whose asymptotic theory and model selection are later studied by \cite{Choietal2018} and \cite{Choietal2023}, respectively. 
The most common scenario considers two levels of factors, often termed as global and local/group factors, where the setup is known as a group factor model \citep[e.g.][]{Andreouetal2019, Han2021}.
We also note that such a framework is similar to hierarchical factor models in \cite{Moenchetal2013}. 
In particular, consider $M$ series of observed vectors. For each group $m=1,\dots,M$, the vector observation $\bx_t^{(m)}\in \R^{p_m}$ is represented by $\bx_t^{(m)} = \bA_m \mathbf{g}_t + \bB_m \mathbf{f}_t^{(m)} +\mathbf{e}_t^{(m)}$. The low-rank vectors $\mathbf{g}_t$ and $\mathbf{f}_t^{(m)}$ are termed as the global and local factors respectively, with $\bA_m$ and $\bB_m$ being their corresponding loading matrices. By stacking all groups of vectors at the same timestamp, one may write
{\small
\[
\begin{pmatrix}
    \bx_t^{(1)} \\ \bx_t^{(2)} \\ \vdots \\ \bx_t^{(M)}
\end{pmatrix} = \begin{pmatrix}
    \bA_1 & \bB_1 & \zero & \ldots & \zero \\
    \bA_2 & \zero & \bB_2 & \ldots & \zero \\
    \vdots & \vdots & \vdots & \ddots & \vdots \\
    \bA_M & \zero & \zero & \ldots & \bB_M
\end{pmatrix} \begin{pmatrix}
    \mathbf{g}_t \\ \mathbf{f}_t^{(1)} \\ \vdots \\ \mathbf{f}_t^{(M)}
\end{pmatrix} + \begin{pmatrix}
    \mathbf{e}_t^{(1)} \\ \mathbf{e}_t^{(2)} \\ \vdots \\ \mathbf{e}_t^{(M)}
\end{pmatrix} ,
\]
}
which can be viewed as a vector factor model with specific structures. 
In addition to the interpretation of factor structures at each group, one benefit of the techniques developed for these works is that, the model estimation can be more efficient by accounting for the multilevel or group structure, compared to estimating a very sparse factor loading matrix if all observed panels with unaligned dimensions are directly stacked together, as shown above. However, re-organizing data sets becomes infeasible for higher-order tensor data (including matrix data which are order-2 tensors), urgently requiring appropriate methods to model the data sets.

Given $M$ groups in total, for each series of matrix observations with dimension $p_{m,1}\times p_{m,2}$, $m=1,\dots, M$, \cite{Zhangetal2025} considered multilevel factor models by decently extending the work in \cite{Wang2008} and \cite{Choietal2018} to matrix time series. However, their estimators of local factor loading matrices did not take advantage of the matrix format as \cite{yu2022projected} did, leaving the possibility to improve the estimation performance. Although their augmented iterative algorithm tried to amend the slow rates of convergence, their initial estimator and thereby the iterative algorithm are computationally inefficient. Moreover, their Condition~6 required that all row dimensions $p_{m,1}$'s and column dimensions $p_{m,2}$'s are respectively of the same order, which might be unrealistic to many datasets empirically.
See also \cite{Luetal2025} who studied global and row/column-specific local factors for series of matrix data.

As tensor data becomes more prevalent nowadays, it is crucial to study multiple time series.
A realistic scenario is that the tensor-valued time series may have different tensor orders and dimensions. Consequently, there generally does not exist a common latent tensor factor, nor can the observations be embedded into a common tensor space without imposing substantial additional structure. This feature fundamentally distinguishes our problem from existing multilevel factor models and from homogeneous tensor factor models.
Our solution is to in turn consider global factors that are connected only due to their correlation. Note that even when dealing with time series the same tensor orders, this setup is more flexible than that in multilevel factor models.

In the context of multilevel factor modeling,
one difficulty lies in specifying the number of global and local factors, especially when tensor orders can be different. 
\cite{Choietal2018} proposed a consistent estimator based on several information criteria, whereas \cite{Zhangetal2025} only provided a practical approach based on eigenvalue ratio. The former method remains unclear when higher-order tensor time series are present, and for the latter, we will show in our numerical experiments that the vanilla eigenvalue-ratio estimator tends to overestimate the factor numbers, resulting in a part of idiosyncratic noise included in the estimated factors. Although this is not too problematic in practice, the effect of overestimating the factor numbers can be non-negligible \citep{BarigozziCho2020}.

To address all the above concern, we propose a multi-order tensor factor model to study a joint collection of tensor-valued time series, with at least the following contributions. As far as we are aware, our proposed model is the first work identifying and studying multiple tensor time series, each with potentially different orders and dimensions. 
Table~\ref{tab: compare_literature} summarizes the various setups considered in the existing literature.
Secondly, we adapt both techniques in \cite{Bai2003} and \cite{Lametal2011}, and establish theoretical rates of convergence that are often superior to the results in, for example, \cite{Zhangetal2025}, even under very general dependence conditions. We also propose a consistent eigenvalue-ratio estimator for the global factor numbers, with a detailed discussion therein. 
Last but not least, a by-product of our model is a multilevel factor model for general tensor time series, with minor dimension requirements.
{Our method can be viewed as a structured representation-learning tool for collections of heterogeneous tensors. 
In finance, for example, one may treat different data sources as separate threads, such as returns, realized volatility surfaces, and order-book snapshots. 
The global factors would then capture co-movement shared across these sources, with the local factors absorbing source-specific dynamics.
Hence we hope the proposed framework sheds new light on statistical learning and financial applications involving multiple heterogeneous data sources.}

\begin{table}[t!]\centering
\caption{Model setups in comparison with existing works on multilevel/group factor models.}
\label{tab: compare_literature}
\ra{1.2}
\setlength{\tabcolsep}{9pt}
\begin{tabular}{@{}llllll@{}}
\hline 
Model &  Tensor order & Global factor structure \\
\hline 
\cite{Andreouetal2019} & 1 & Identical \\
\cite{Choietal2018} & 1 & Identical \\
\cite{Choietal2023} & 1 & Identical \\
\cite{Han2021} & 1 & Identical \\
\cite{Luetal2025} & 2 & Identical \\
\cite{Wang2008} & 1 & Identical \\
\cite{Zhangetal2025} & 2 & Identical \\
This paper & Heterogeneous & Correlated \\
\hline 
\end{tabular}
\end{table}

The rest of this paper is organized as follows. Section~\ref{sec: mm_tfm} introduces the data profile and formalizes the multi-order tensor factor model. In detail, Section~\ref{subsec: model} concerns the identification issue, followed by parameter estimation in Sections~\ref{subsec: est_loading} and \ref{subsec: est_fac}. Assumptions are listed and explained in Section~\ref{subsec: assumption}, while theoretical guarantee of the estimated parameters is included in Section~\ref{subsec: theorem}. Section~\ref{sec: further_discussion} concerns various topics of interests. 
Section~\ref{subsec: simulation} presents the numerical performance of the proposed methods, and finally, in Section~\ref{subsec: real_data}, a set of taxi traffic data is analyzed using our model. 
Section~\ref{sec: conclusion} concludes.
All the proofs of theoretical results, additional numerical results, and details on the real data, are relegated to the supplement.

Some notational remarks.
Unless specified otherwise, we denote vectors, matrices and tensors by bold lower-case, bold capital, and calligraphic letters, i.e., $\bx$, $\bX$ and $\cX$, respectively. We also use $x_i, X_{ij}, \bX_{i\cdot}, \bX_{\cdot i}$ to denote, respectively, the $i$th element of a vector $\bx$, the $(i,j)$th element of $\bX$, the $i$th row vector (as a column vector) of $\bX$, and the $i$th column vector of $\bX$. We use $\otimes$ to represent the Kronecker product. By convention, the total Kronecker product for an index set is computed in descending index. We use $a\asymp b$ to denote $a=\cO(b)$ and $b=\cO(a)$, or their probabilistic versions.
Given a positive integer $m$, define $[m]=\{1,2,\dots,m\}$. The $i$th largest eigenvalue and singular value of a matrix $\bX$ are denoted by $\lambda_i(\bX)$ and $\sigma_i(\bX)$, respectively. For an order-$K$ tensor $\cX = (X_{i_1,\ldots,i_K}) \in \mathbb{R}^{p_1\times \cdots\times p_K}$, we denote by $\mat_k(\cX) \in \R^{p_k\times \pmk}$ its \textit{mode-$k$ unfolding/matricization}. We denote by $\cX \times_k \bA$ the \textit{mode-$k$ product} of a tensor $\cX$ with a matrix $A$, defined by $\mat_k(\cX \times_k \bA) = \bA\, \mat_k(\cX)$.

\section{Multi-order Tensor Factor Models}\label{sec: mm_tfm}

\subsection{Model setup}\label{subsec: model}

Under some categorization $m\in[M]$ with fixed $M$, suppose that we observe a tensor time series $\{\cX_t^{(m)}\}$, $t\in[T]$, such that each $\cX_t^{(m)}$ is a mean-zero order-$K_m$ tensor of dimension $p_{m,1}\times \dots \times p_{m,K_m}$. Without loss of generality, let $K_1\leq \dots\leq K_M$. By joining $\{ \cX_t^{(1)}, \dots,\cX_t^{(M)} \}_{t\in[T]}$ as a collection (JXC hereafter), we assume JXC follows a \textit{multi-order tensor factor model} such that for each $m\in[M]$, the $m$th \textit{thread} $\{\cX_t^{(m)}\}$ admits the representation for each $t\in[T]$:
\begin{equation}
\label{eqn: model}
\cX_t^{(m)} = \cG_t^{(m)} \times_{k=1}^{K_m} \bA_{m,k} + \cF_t^{(m)} \times_{\ell=1}^{K_m} \bB_{m,\ell} + \cE_t^{(m)} ,
\end{equation}
where $\cG_t^{(m)}$ is the global core factor with dimension $r_{m,1}\times \dots\times r_{m,K_m}$, $\cF_t^{(m)}$ is the local core factor with dimension $u_{m,1}\times \dots\times u_{m,K_m}$, $\bA_{m,k}$ and $\bB_{m,\ell}$ ($k,\ell \in[K_m]$) are respectively the deterministic global and local factor loading matrices, and $\cE_t^{(m)}$ is the noise. For a non-trivial discussion, we assume $M>1$ throughout this paper, and otherwise \eqref{eqn: model} boils down to the classical Tucker-decomposition tensor factor model. Note that only JXC is observed and the entire right-hand side of \eqref{eqn: model} is latent.

Model~\eqref{eqn: model} turns out to be a very general framework. A useful special case is when $K\equiv K_1=\dots =K_M$ and $\cG_t\equiv \cG_t^{(1)} =\dots=\cG_t^{(M)}$, which is commonly known as a multilevel/group factor model if $K=1$ \citep{Choietal2018, Huetal2025} or its matrix-valued extension for $K=2$ \citep{Zhangetal2025}. An immediate importance of \eqref{eqn: model} is thus a generalization of multilevel factor models to tensor time series. It might be worth pointing out that even when the tensor orders are the same across $m\in[M]$, traditional factor models cannot directly characterize the factor structure in JXC, unless some forms of concatenation are carried out and thus inevitably inflating the number of parameters and overlooks the structure along each mode.

Note that in most, if not all, multilevel factor models, the same global factor contributes to different groups, whereas the interplay among $\cG_t^{(m)}$ ($m\in[M]$) in our formulation \eqref{eqn: model} remains unclear. First, it should be straightforward that $\cG_t^{(m)} \neq \cG_t^{(n)}$ if $K_m\neq K_n$. More importantly, we choose not to specify, e.g., how $\cG_t^{(1)}$ affects $\cG_t^{(2)}$, and we see this lack of formulation as an advantage in view of the generality of \eqref{eqn: model}. 

We refer to $\cX_{G,t}^{(m)}:= \cG_t^{(m)} \times_{k=1}^{K_m} \bA_{m,k}$ and $\cX_{F,t}^{(m)}:= \cF_t^{(m)} \times_{\ell=1}^{K_m} \bB_{m,\ell}$ as the global and local components, respectively. Then with \eqref{eqn: model}, one can model multiple tensor time series with different orders at the same time, with the global component featuring the dependence of observed tensors across threads, which is new to the literature. To separate the global and local components, we impose the following identification conditions.

\begin{assumption}[Identification]\label{ass: identification}
For all $t\in[T]$, $m\in[M]$, we assume that: 
\begin{enumerate}[itemsep=0pt, label = (\alph*), left = 0pt]
    \item Any element in $\big\{ \cG_t^{(m)}, \cF_t^{(m)} \big\}$ has zero mean;
    \item Any element in $\big\{ \cG_t^{(m)}, \cF_t^{(m)} \big\}$ is uncorrelated to those in $\cF_t^{(n)}$ for all $n\neq m$.
\end{enumerate}
\end{assumption}

Assumption~\ref{ass: identification} is in the same spirit of Assumption~1 in \cite{Choietal2018} and Condition~5 in \cite{Zhangetal2025}. Part~(a) resolves the indeterminacy in mean,
whereas part~(b) indicates that the local factor does not carry any dependence structure among the JXC.
These assumptions are mild in the sense that the global factor $\cG_t^{(m)}$ can be contemporaneously correlated to $\cF_t^{(m)}$. In comparison, such a correlation is prohibited in Assumption~1 in \cite{Choietal2018} for instance, due to their restricted $\cG_t^{(1)} =\dots= \cG_t^{(M)}$ which is required to be uncorrelated to the local factors.

Next, for every $m\in[M]$, $k\in[K_m]$, $n\neq m$, and $\bi=(i_1,\ldots,i_{K_n}) \in [p_{n,1}] \times \cdots \times[p_{n,K_n}]$,
define the population filtered unfolding
\begin{equation}
\label{eqn: def_Omega_pop}
\bOmega_{k,n,\bi}^{(m)} := \frac{1}{T} \sum_{t=1}^T \E\Big\{ \mat_k(\cX_t^{(m)}) \cdot (\cX_t^{(n)})_{\bi} \Big\} .
\end{equation}
and the population cross-thread covariance matrix as the non-negative definite matrix
\begin{equation}
\label{eqn: def_Sigma_pop}
\bSigma_{m,k} := \sum_{\substack{n=1\\ n\neq m}}^M \sum_{\bi \in [p_{n,1}] \times \cdots \times[p_{n,K_n}]} \bOmega_{k,n,\bi}^{(m)} \bOmega_{k,n,\bi}^{(m) \trans} .
\end{equation}

With Assumption~\ref{ass: identification} and the above notations, we formalize the identification of global and local components in the theorem below.

\begin{theorem}
\label{prop: identification}
Let Assumption~\ref{ass: identification} hold. Further suppose for every $m\in[M]$ and $k\in[K_m]$:
(a)~the rank of $\bSigma_{m,k}$ is $r_{m,k}$;
and (b)~the column space $\col(\bB_{m,k}) \subseteq \col(\bA_{m,k})_{\perp}$, where $\col(\bA_{m,k})_{\perp}$ denotes the orthogonal complement of $\col(\bA_{m,k})$.

Consider the oracle scenario with $\cE_t^{(m)}=0$ in \eqref{eqn: model} for simplicity. 
Then for any $t\in[T]$ and two admissible sets of parameters $(\cX_{G,t}^{(m)}, \cX_{F,t}^{(m)})$ and $(\wt{\cX}_{G,t}^{(m)}, \wt{\cX}_{F,t}^{(m)})$ satisfying the above conditions, if
$\cX_{G,t}^{(m)}+\cX_{F,t}^{(m)} = \wt{\cX}_{G,t}^{(m)}+\wt{\cX}_{F,t}^{(m)}$, 
then $\cX_{G,t}^{(m)} = \wt{\cX}_{G,t}^{(m)}$ and $\cX_{F,t}^{(m)} = \wt{\cX}_{F,t}^{(m)}$.
\end{theorem}

Theorem~\ref{prop: identification} serves to separate the global and local components, which is sufficient for the discussion in this paper, since both components admit Tucker decompositions and hence identification of parameters within has been well studied. 

For this reason, Theorem~\ref{prop: identification} is stated in the oracle scenario,
and we emphasize that the theorem would remain true given, for instance, Assumption~\ref{ass: dependence_moment}(a).
The formal separation of the common component from the idiosyncratic noise is carried out in parameter estimation, where Assumptions~\ref{ass: noise}--\ref{ass: dependence_thread} (in Section~\ref{subsec: assumption}) ensure that the noise and cross-dependence terms are asymptotically negligible relative to the eigenvalues of the common components.

Markedly, Condition~(a) in Theorem~\ref{prop: identification} requires the population cross-thread covariance matrix $\bSigma_{m,k}$ to have full rank $r_{m,k}$. Under Assumption~\ref{ass: identification}(b), the model implies
$\bSigma_{m,k} = \bA_{m,k} \bC_{m,k} \bA_{m,k}^\trans$ 
for some non-negative definite matrix $\bC_{m,k} \in \R^{r_{m,k} \times r_{m,k}}$.
Hence condition~(a) implies that $\bA_{m,k}$ has full column rank and thus $\col(\bSigma_{m,k}) = \col(\bA_{m,k})$,
so the global loading space is exactly the leading eigenspace of $\bSigma_{m,k}$. 
Condition~(b) in the theorem requires the local loading space to be orthogonal to the global loading space, so that the global component is the orthogonal projection of the common component onto $\col(\bA_{m,k})$, while the local component is the residual in the orthogonal complement.
These two conditions are the population analogues of the conditions imposed later for estimation. 
In particular, Assumptions~\ref{ass: loadings}(b) and \ref{ass: core_factor}(c) ensure that condition~(a) holds,
while Assumption~\ref{ass: loadings}(b) gives condition~(b) exactly.
Thus provided Assumption~\ref{ass: identification}, the later structural assumptions in Section~\ref{subsec: assumption} are sufficient for the identification of the global and local components.

{We emphasize that the identification statement is a sufficient condition. 
In particular, Condition~(a) requires the full-rank contemporaneous cross-thread covariance separately for each pair $(m,n)$, $n\neq m$, whereas a strictly weaker aggregated full-rank condition over $n\neq m$ would also be sufficient for constructing $\bSigma_{m,k}$. 
We adopt the pairwise formulation for convenience; see also the discussion at the end of Remark~\ref{remark: connect_TIPUP_TOPUP}.}

\begin{remark}[A thread-wise view under classical factor models]
\label{remark: model_read}
Fix $m\in[M]$ and consider the dynamic of the $m$th thread only. Model~\eqref{eqn: model} assumes a sparse tensor factor representation in the sense that we can rewrite the model as
\begin{equation}
\label{eqn: Xt_rewrite_tfm}
\cX_t^{(m)} = \diag(\cG_t^{(m)}, \cF_t^{(m)}) \times_{k=1}^{K_m} (\bA_{m,k}, \bB_{m,k}) + \cE_t^{(m)} ,
\end{equation}
where $\diag(\cG_t^{(m)}, \cF_t^{(m)})$ denotes the block tensor consisting $\cG_t^{(m)}$ and $\cF_t^{(m)}$ on its diagonal. 
A serious issue by modeling each observed tensor time series as above is the neglected dependence between threads and interpretation lost.
To recover interpretation, a possible alternative is a two-stage procedure in which one first estimates the joint factor model in \eqref{eqn: Xt_rewrite_tfm}, and then separates the global and local parts by canonical correlation analysis between the estimated factors of different threads. However, this procedure would not simplify the estimation, not improve the asymptotic rate of the global loading estimator.
In fact, as more factors are introduced (due to omitting the block structure in \eqref{eqn: Xt_rewrite_tfm}), the estimation performance could be undermined in practice. 
See the supplement for a detailed discussion.
\end{remark}

\begin{remark}[A collective view under vectorization]
With tensors with heterogeneous orders in JXC, a plausible method to model all series is to concatenate the vectorized data. In detail, for each $m\in[M]$, \eqref{eqn: model} implies $\vec(\cX_t^{(m)}) = (\otimes_{k=1}^{K_m} \bA_{m,k}, \otimes_{k=1}^{K_m} \bB_{m,k}) \big\{ \vec(\cG_t^{(m)})^\trans, \vec(\cF_t^{(m)})^\trans \big\}^\trans + \vec(\cE_t^{(m)})$. Thus we can write
{\small
\begin{align*}
    \begin{pmatrix}
    \vec(\cX_t^{(1)}) \\ \vec(\cX_t^{(2)}) \\ \vdots \\ \vec(\cX_t^{(M)})
    \end{pmatrix} 
    = \begin{pmatrix}
    \otimes_{k=1}^{K_1} \bA_{1,k} & \otimes_{k=1}^{K_1} \bB_{1,k}  & \ldots & \zero \\
    \zero & \zero &  \ldots & \zero \\
    \vdots & \vdots &  \ddots & \vdots \\
    \zero & \zero &  \ldots & \otimes_{k=1}^{K_M} \bB_{M,k}
    \end{pmatrix} \begin{pmatrix}
    \vec(\cG_t^{(1)}) \\ \vec(\cF_t^{(1)}) \\ \vdots \\ \vec(\cG_t^{(M)}) \\ \vec(\cF_t^{(M)})
    \end{pmatrix} + \begin{pmatrix}
    \vec(\cE_t^{(1)}) \\ \vec(\cE_t^{(2)}) \\ \vdots \\ \vec(\cE_t^{(M)})
    \end{pmatrix} ,
\end{align*}
}
or, if $\vec(\cG_t^{(1)}) =\dots =\vec(\cG_t^{(M)})$ under certain contexts (such as tensor multilevel factor models),
{\small
\begin{align*}
    \begin{pmatrix}
    \vec(\cX_t^{(1)}) \\ \vec(\cX_t^{(2)}) \\ \vdots \\ \vec(\cX_t^{(M)})
    \end{pmatrix} 
    = \begin{pmatrix}
    \otimes_{k=1}^{K_1} \bA_{1,k} & \otimes_{k=1}^{K_1} \bB_{1,k}  & \ldots & \zero \\
    \otimes_{k=1}^{K_2} \bA_{2,k} & \zero & \ldots & \zero \\
    \vdots & \vdots &  \ddots & \vdots \\
    \otimes_{k=1}^{K_M} \bA_{M,k} & \zero &  \ldots & \otimes_{k=1}^{K_M} \bB_{M,k}
    \end{pmatrix} \begin{pmatrix}
    \vec(\cG_t^{(1)}) \\ \vec(\cF_t^{(1)}) \\ \vdots \\ \vec(\cF_t^{(M)})
    \end{pmatrix} + \begin{pmatrix}
    \vec(\cE_t^{(1)}) \\ \vec(\cE_t^{(2)}) \\ \vdots \\ \vec(\cE_t^{(M)})
    \end{pmatrix} .
\end{align*}
}
However, both representations above significantly inflate the number of parameters in the loading matrix, and recovering the non-zero blocks within is also difficult in its own sake.
\end{remark}

\begin{remark}[Partial factor structures]
\label{remark: absent_component}
In practice, the absence of the global or local components (or both) in the representation \eqref{eqn: model} is compatible with Assumption~\ref{ass: identification} and should be allowed, which is also more realistic and reasonable from the perspective of data modeling.
We warn that even though some threads might lack global components, there should at least two threads of time series with global components which can thus be identified. 
To account for the scenario where there is at least a global or local component in the $m$th thread, we discuss a mathematically rigorous procedure in Section~\ref{subsec: est_num_fac} by adapting our procedure in estimating global and local factor numbers.
\cite{Zhangetal2025} also discussed in their Remark~5 the possibility of groups of time series whose dynamics are only driven by either the global factor, the local factor, or the noise. However, we note that for higher-order tensor times series with $K_m>1$, an important multi-way structure stands in between a ``fully-loaded'' component and its absence. This is a natural extension of factor structures from vector to higher-order tensors, which has gained increasing attention; see similar works in \text{e.g.} \cite{Heetal2023} and \cite{LamCen2025_MEFM}. Considering the global component for instance, for any $m\in[M]$, there could be a set of mode indices $\cM_m\subseteq [K_m]$ such that
$\cX_{G,t}^{(m)} = \cG_t^{(m)} \times_{k\in \cM_{m}} \bA_{m,k}$,
where the mode-$k$ dimension of $\cG_t^{(m)}$ is $r_{m,k}$ if $k\in \cM_{m}$ and $p_{m,k}$ otherwise. When $K_m=2$, such a representation covers the boundary cases in \cite{Heetal2023}. While our method based on mock eigenvalues in Section~\ref{subsec: est_num_fac} can be adapted to study these variations together with the case where the observed time series lacks factor structure, \text{i.e.} model \eqref{eqn: model} only consists of noise, this is beyond the scope of this paper and we do not further pursue this line of analysis.
{The identification conditions are also compatible with the partial-connectivity configuration in which some threads are correlated only with a subset of the remaining threads. Additional simulation results in the supplement verify that the proposed estimator remains stable in such cases, even though Assumption~\ref{ass: core_factor}(c) is stated for all pairs $n\neq m$.}
\end{remark}

\subsection{Estimation of the global and local loading matrices}
\label{subsec: est_loading}

In this subsection, we discuss the estimation of factor loading matrices in model \eqref{eqn: model}. First, for any $m\in[M]$, $k\in[K_m]$, define $\bA_{m,\text{-}k} := \otimes_{j\in[K_m] \setminus\{k\}} \bA_{m,j}$, and similarly $\bB_{m,\text{-}k}$. By convention, set $\bA_{m,\text{-}1}= \bB_{m,\text{-}1}= 1$ if $K_m=1$. Throughout this subsection, we treat the factor numbers $r_{m,k}$ and $u_{m,k}$ as given; the details of how to estimate them are deferred to Section~\ref{subsec: est_num_fac}.

To estimate the global factor loading matrices, consider the mode-$k$ unfolding of $\cX_t^{(m)}$ for each $t\in[T]$:
\begin{equation}
\label{eqn: unfold_Xt}
\mat_k(\cX_t^{(m)}) = \bA_{m,k} \mat_k(\cG_t^{(m)}) \bA_{m,\text{-}k}^\trans + \bB_{m,k} \mat_k(\cF_t^{(m)}) \bB_{m,\text{-}k}^\trans + \mat_k(\cE_t^{(m)}) .
\end{equation}

For any $n\neq m$ and
$\bi=(i_1,\ldots,i_{K_n}) \in [p_{n,1}] \times \cdots \times[p_{n,K_n}]$,
recall the population filtered unfolding $\bOmega_{k,n,\bi}^{(m)}$ and the population cross-thread covariance matrix $\bSigma_{m,k}$ defined in \eqref{eqn: def_Omega_pop} and \eqref{eqn: def_Sigma_pop}, respectively.
Using \eqref{eqn: unfold_Xt}, Assumption~\ref{ass: identification}(b), and Assumption~\ref{ass: dependence_moment}(a) (in Section~\ref{subsec: assumption}), we have
\begin{equation}
\label{eqn: Omega_pop_decomp}
    \bOmega_{k,n,\bi}^{(m)} =
    \bA_{m,k} \E\Bigg\{ \frac{1}{T} \sum_{t=1}^T \mat_k(\cG_t^{(m)}) \bA_{m,\text{-}k}^\trans \Big( \cG_t^{(n)} \times_{h=1}^{K_n} \bA_{n,h} \Big)_{\bi} \Bigg\} ,
\end{equation}
Then enlightened by Equation~(4) in \cite{Lametal2011}, we expect the $r_{m,k}$ leading eigenvectors of $\bSigma_{m,k}$ span the same linear space as $\bA_{m,k}$ does.

Hence for any $m\in[M]$, $k\in[K_m]$, we may estimate $\bA_{m,k}$ by $\wh{\bA}_{m,k}$ which is defined as $\sqrt{p_{m,k}}$ times the $r_{m,k}$ leading eigenvectors of
\begin{equation}
\label{eqn: def_hat_Sigma_mk}
\begin{split}
    \wh\bSigma_{m,k} &:= \sum_{\substack{n=1\\ n\neq m}}^M \sum_{\bi \in [p_{n,1}] \times \cdots \times[p_{n,K_n}]} \wh\bOmega_{k,n,\bi}^{(m)} \wh\bOmega_{k,n,\bi}^{(m) \trans} ,
\end{split}
\end{equation}
where
\[
\wh\bOmega_{k,n,\bi}^{(m)} := \frac{1}{T} \sum_{t=1}^T \mat_k(\cX_t^{(m)}) \cdot (\cX_t^{(n)})_{\bi} .
\]

To estimate the local loading matrices, first consider $K_m >1$.
Let $\wh{\bA}_{m,k,\perp}$ be
all but the $r_{m,k}$ leading eigenvectors of $\wh\bSigma_{m,k}$, \text{viz.} $\wh{\bA}_{m,k,\perp}$ is the orthogonal complement of $\wh{\bA}_{m,k}$. 
Then from \eqref{eqn: unfold_Xt},
\begin{align*}
    &\;\quad
    \mat_k(\cX_t^{(m)}) \wh{\bA}_{m,\text{-}k,\perp} - \bA_{m,k} \mat_k(\cG_t^{(m)}) \bA_{m,\text{-}k}^\trans \wh{\bA}_{m,\text{-}k,\perp} \\
    &= \bB_{m,k} \mat_k(\cF_t^{(m)}) \bB_{m,\text{-}k}^\trans \wh{\bA}_{m,\text{-}k,\perp} + \mat_k(\cE_t^{(m)}) \wh{\bA}_{m,\text{-}k,\perp} ,
\end{align*}
where $\wh{\bA}_{m,\text{-}k,\perp} := \otimes_{j\in[K_m] \setminus\{k\}} \wh{\bA}_{m,j,\perp}$. Intuitively, $\wh{\bA}_{m,\text{-}k,\perp}$ is close to the orthogonal complement of $\bA_{m,\text{-}k}$, and hence $\bA_{m,\text{-}k}^\trans \wh{\bA}_{m,\text{-}k,\perp}$ after appropriate scaling should be close to zero. This motivates a PCA-type estimator for $\bB_{m,k}$, denoted by $\wh{\bB}_{m,k}$ which is defined as $\sqrt{p_{m,k}}$ times the $u_{m,k}$ leading eigenvectors of
\begin{equation}
\label{eqn: def_hat_Sigma_B_mk}
\begin{split}
    \wh\bSigma_{B,m,k} &:= \frac{1}{Tp_m} \sum_{t=1}^T \mat_k(\cX_t^{(m)}) \wh{\bA}_{m,\text{-}k,\perp} \wh{\bA}_{m,\text{-}k,\perp}^\trans \mat_k(\cX_t^{(m)})^\trans ,
\end{split}
\end{equation}
where $p_m:= \prod_{k=1}^{K_m} p_{m,k}$. Now suppose $K_m=1$ and the above approach to obtain $\wh{\bB}_{m,1}$ is thus infeasible. Fortunately, \eqref{eqn: unfold_Xt} can be rewritten as
\[
\mat_1(\cX_t^{(m)}) = (\bA_{m,1}, \bB_{m,1}) \big\{ \mat_1(\cG_t^{(m)})^\trans, \mat_1(\cF_t^{(m)})^\trans \big\}^\trans + \mat_k(\cE_t^{(m)}) .
\]
Hence by leveraging again the orthogonal complement of the global loading matrix, we may compute $\wh{\bB}_{m,1}$ as $\sqrt{p_{m,1}}$ times the $u_{m,1}$ leading eigenvectors of
\begin{equation}
\label{eqn: def_tilde_Sigma_B_m1}
\begin{split}
    \wt\bSigma_{B,m,1} &:= \frac{1}{Tp_{m,1}} \sum_{t=1}^T \wh{\bA}_{m,1,\perp} \wh{\bA}_{m,1,\perp}^\trans \mat_1(\cX_t^{(m)}) \mat_1(\cX_t^{(m)})^\trans \wh{\bA}_{m,1,\perp} \wh{\bA}_{m,1,\perp}^\trans .
\end{split}
\end{equation}

\subsection{Estimation of the core factors and components}
\label{subsec: est_fac}

With $\wh{\bA}_{m,k}$, $\wh{\bA}_{m,k,\perp}$, and $\wh{\bB}_{m,k}$ ($m\in[M], k\in[K_m]$) obtained in Section~\ref{subsec: est_loading}, we have for any $t\in[T]$, $m\in[M]$,
\begin{align*}
    &\;\quad
    (\otimes_{k=1}^K \wh{\bA}_{m,k,\perp})^\trans \vec(\cX_t^{m}) - (\otimes_{k=1}^K \wh{\bA}_{m,k,\perp})^\trans (\otimes_{k=1}^K \bA_{m,k}) \vec(\cG_t^{(m)}) \\
    &= (\otimes_{k=1}^K \wh{\bA}_{m,k,\perp})^\trans (\otimes_{k=1}^K \bB_{m,k}) \vec(\cF_t^{(m)}) + (\otimes_{k=1}^K \wh{\bA}_{m,k,\perp})^\trans \vec(\cE_t^{m}) ,
\end{align*}
where we expect $(\otimes_{k=1}^K \wh{\bA}_{m,k,\perp})^\trans (\otimes_{k=1}^K \bA_{m,k})$ after some scaling to be close to zero, as discussed in Section~\ref{subsec: est_loading}. Then with the notation $\wh{\bC}_m := (\otimes_{k=1}^K \wh{\bA}_{m,k,\perp})^\trans (\otimes_{k=1}^K \wh{\bB}_{m,k})$, the vectorized local factor can be estimated in a least squared manner such that
\begin{equation}
\label{eqn: def_hat_Ft}
\begin{split}
    \vec(\wh{\cF}_t^{(m)}) &:= (\wh{\bC}_m^\trans \wh{\bC}_m)^{-1} \wh{\bC}_m^\trans (\otimes_{k=1}^K \wh{\bA}_{m,k,\perp})^\trans \vec(\cX_t^{m}) .
\end{split}
\end{equation}
For any $t\in[T]$, $m\in[M]$, the global factor can also be formulated as in a linear regression form:
\begin{align*}
    &\;\quad
    \vec(\cX_t^{m}) - (\otimes_{k=1}^K \bB_{m,k}) \vec(\cF_t^{(m)}) = (\otimes_{k=1}^K \bA_{m,k}) \vec(\cG_t^{(m)}) + \vec(\cE_t^{m}) .
\end{align*}
Then similar to \eqref{eqn: def_hat_Ft}, we may plug in the previous estimators and arrive at
\begin{equation}
\label{eqn: def_hat_Gt}
\begin{split}
    \vec(\wh{\cG}_t^{(m)}) &:= \Big\{ (\otimes_{k=1}^K \wh{\bA}_{m,k})^\trans (\otimes_{k=1}^K \wh{\bA}_{m,k}) \Big\}^{-1} \\
    &\;\quad
    \cdot (\otimes_{k=1}^K \wh{\bA}_{m,k})^\trans \Big\{ \vec(\cX_t^{m}) - (\otimes_{k=1}^K \wh{\bB}_{m,k}) \vec(\wh\cF_t^{(m)}) \Big\} .
\end{split}
\end{equation}

Finally, the global and local common components are respectively estimated by
\[
\wh\cX_{G,t}^{(m)}:= \wh\cG_t^{(m)} \times_{k=1}^{K_m} \wh{\bA}_{m,k}, \quad
\wh\cX_{F,t}^{(m)}:= \wh\cF_t^{(m)} \times_{\ell=1}^{K_m} \wh{\bB}_{m,\ell} .
\]

\begin{remark}[Simplification of $\wh\bSigma_{m,k}$ under homogeneous orders]
\label{remark: hat_Sigma_rewrite}
Suppose all tensor orders are the same, i.e., $K_1=\dots=K_M$. For any $k\in[K_1]$, define $p_{n,\text{-}k} :=p_n/p_{n,k}$ and the matrix
\[
\wh\bOmega_{k,n,i,j}^{(m)} := \frac{1}{T} \sum_{t=1}^T \mat_k(\cX_t^{(m)})_{\cdot j} \mat_k(\cX_t^{(n)})_{\cdot i}^\trans .
\]
Then we may rewrite $\wh\bSigma_{m,k}$ defined in \eqref{eqn: def_hat_Sigma_mk} as
\begin{align*}
    \wh\bSigma_{m,k} &= \frac{1}{T^2} \sum_{\substack{n=1\\ n\neq m}}^M \sum_{t,s=1}^T \mat_k(\cX_t^{(m)}) \mat_k(\cX_s^{(m)})^\trans \sum_{\bi \in [p_{n,1}] \times \cdots \times[p_{n,K_n}]} (\cX_t^{(n)})_{\bi} (\cX_s^{(n)})_{\bi} \\
    &=
    \frac{1}{T^2} \sum_{\substack{n=1\\ n\neq m}}^M \sum_{t,s=1}^T \Big\{ \vec(\cX_s^{(n)})^\trans \vec(\cX_t^{(n)}) \Big\} \otimes \Big\{ \mat_k(\cX_t^{(m)}) \mat_k(\cX_s^{(m)})^\trans \Big\} \\
    &=
    \frac{1}{T^2} \sum_{\substack{n=1\\ n\neq m}}^M \sum_{t,s=1}^T \tr\Big\{ \sum_{i=1}^{p_{n,\text{-}k}} \mat_k(\cX_s^{(n)})_{\cdot i} \mat_k(\cX_t^{(n)})_{\cdot i}^\trans \Big\} \Big\{ \mat_k(\cX_t^{(m)}) \mat_k(\cX_s^{(m)})^\trans \Big\} \\
    &=
    \frac{1}{T^2} \sum_{\substack{n=1\\ n\neq m}}^M \sum_{i=1}^{p_{n,\text{-}k}} \sum_{t,s=1}^T \Big\{ \mat_k(\cX_t^{(n)})_{\cdot i}^\trans  \mat_k(\cX_s^{(n)})_{\cdot i} \Big\} \mat_k(\cX_t^{(m)}) \mat_k(\cX_s^{(m)})^\trans \\
    &=
    \frac{1}{T^2} \sum_{\substack{n=1\\ n\neq m}}^M \sum_{i,j=1}^{p_{n,\text{-}k}} \sum_{t,s=1}^T \mat_k(\cX_t^{(m)})_{\cdot j} \mat_k(\cX_t^{(n)})_{\cdot i}^\trans \mat_k(\cX_s^{(n)})_{\cdot i} \mat_k(\cX_s^{(m)})_{\cdot j}^\trans \\
    &=
    \sum_{\substack{n=1\\ n\neq m}}^M \sum_{i,j=1}^{p_{n,\text{-}k}} \wh\bOmega_{k,n,i,j}^{(m)} \wh\bOmega_{k,n,i,j}^{(m) \trans} ,
\end{align*}
which provides an equivalent way to compute $\wh\bSigma_{m,k}$. When $K_1=\dots=K_M=2$, the above $\wh\bOmega_{k,n,i,j}^{(m)}$ coincides with the notation used in \cite{Zhangetal2025} to estimate the global factor loading matrix. This indicates that under the special case of multilevel matrix factor models, our estimator $\wh\bA_{m,k}$ is the same as that by \cite{Zhangetal2025}.
\end{remark}

\begin{remark}[Connection with TIPUP and TOPUP]
\label{remark: connect_TIPUP_TOPUP}
Our estimation procedure is closely related to the TIPUP and TOPUP methods developed for tensor factor models \citep{Chenetal2022, Hanetal2024}. 
Simply put, for a single tensor-valued time series $\cX_t=\cF_t \times_{k=1}^{K} \bA_k+ \cE_t$, $t\in[T]$,
TIPUP estimates the loading $\bA_k$ through a matrix constructed from mode-$k$ inner-products 
between lagged observations, 
whereas TOPUP uses the outer products of the mode-$k$ unfoldings. 
In particular, TIPUP first contracts the mode-$k$ fibers before squaring (viz.\ a matrix right-multiplied by its transpose),
and this signed aggregation can lead to signal cancellation when the lagged auto-covariances sum to a rank-deficient matrix. 
TOPUP avoids this by keeping the lagged auto-covariances of all mode-$k$ fiber pairs before summation, 
which amplifies the signal at the cost of accumulated noise (cf.\ discussion in \citeauthor{Chenetal2022}, \citeyear{Chenetal2022}). 
Also markedly, the pre-averaging procedure in \cite{ChenLam2024} is an intermediate aggregation strategy.

In view of these, our construction can be seen as a cross-thread analogue of the TIPUP inner-product matrices,
in that we can write
\[
\wh\bSigma_{m,k} =
\sum_{\substack{n=1\\ n\neq m}}^M \sum_{\bi \in [p_{n,1}] \times \cdots \times[p_{n,K_n}]}
\left( \frac{1}{T} \sum_{t=1}^T \mat_k(\cX_t^{(m)}) (\cX_t^{(n)})_{\bi} \right)
\left( \frac{1}{T} \sum_{t=1}^T \mat_k(\cX_t^{(m)}) (\cX_t^{(n)})_{\bi} \right)^\trans .
\]
whose differences with TIPUP are mainly twofold.
First, no spatial contraction occurs since the cross-thread signal aggregated before squaring are scalars $(\cX_t^{(n)})_{\bi}$, rather than a linear combination of some tensor fibers; the summation over cross-sectional units and threads is performed after squaring.
Hence there is no TIPUP-versus-TOPUP choice to be made in our framework.
Second, our construction uses contemporaneous cross-thread dependence rather than lagged dependence.
While lagged information could be helpful, it would require additional dynamic identification conditions and considerably complicate the separation of global and local components in our context. We leave such extensions to future work.
\end{remark}

\section{Assumptions and Theoretical Results}\label{sec: theorem}

\subsection{Main assumptions}\label{subsec: assumption}

For each $m\in[M]$, define $p_{m,\text{-}k} :=p_m/p_{m,k}$, $r_m:= \prod_{k=1}^{K_m} r_{m,k}$, and $r_{m,\text{-}k} :=r_m/r_{m,k}$. In the following, we present assumptions on the factor structure, noise, and dependence structure.

\begin{assumption}[Loading matrices]\label{ass: loadings}
For any $m\in[M]$, $k\in[K_m]$,
both the global and local factor numbers $r_{m,k}$ and $u_{m,k}$ are fixed, and we assume:
\begin{enumerate}[itemsep=0pt, label = (\alph*), left = 0pt]
    \item There exists some constant $c$ such that $\|\bA_{m,k}\|_{\max}, \|\bB_{m,k}\|_{\max} \leq c <\infty$.
    \item $p_{m,k}^{-1} (\bA_{m,k}, \bB_{m,k})^\trans (\bA_{m,k}, \bB_{m,k}) = \bI_{r_{m,k}+u_{m,k}}$
    for sufficiently large $p_{m,k}$.
\end{enumerate}
\end{assumption}

\begin{assumption}[Core factor]\label{ass: core_factor}
For all $m\in[M]$, we assume that:
\begin{enumerate}[itemsep=0pt, label = (\alph*), left = 0pt]
    \item Any element in $\big\{ \cG_t^{(m)}, \cF_t^{(m)} \big\}$ has bounded fourth moment.
    \item For any $k\in[K_m]$, as $T\to\infty$, we have
    \[
    \frac{1}{T} \sum_{t=1}^T \mat_k(\cF_t^{(m)}) \mat_k(\cF_t^{(m)})^\trans \xrightarrow{P} \bSigma_{F,m,k} ,
    \]
    where $\bSigma_{F,m,k}$ is a positive definite matrix with finite eigenvalues.
    {\item For any $m\in [M]$, 
    $n\neq m$,
    the mean of the process $\big\{ \vec(\cG_t^{(m)}) \otimes \vec(\cG_t^{(n)})^\trans \big\}$ is time-invariant, and
    for any $k\in[K_m]$, $j\in[r_{m,k}]$, we have
    \[
    \sigma_{j}\Big( \E\big\{ \mat_k(\cG_t^{(m)}) \otimes \vec(\cG_t^{(n)})^\trans \big\} \Big) \asymp 1 .
    \]}
\end{enumerate}
\end{assumption}

\begin{assumption}[Idiosyncratic noise]\label{ass: noise}
For all $m\in[M]$, $t\in[T]$, we assume that:
\begin{enumerate}[itemsep=0pt, label = (\alph*), left = 0pt]
    \item Any element in $\cE_t^{(m)}$ has zero mean and bounded fourth moment.
    \item For all $k\in[K_m]$, $t\in[T]$, $i\in[p_{m,k}]$, $j\in[p_{m,\text{-}k}]$, there exists some constant $c$ such that
    \begin{align*}
        & \sum_{s=1}^T \sum_{l=1}^{p_{m,k}} \sum_{h=1}^{p_{m,\text{-}k}} \Big| \E\Big\{ \mat_k(\cE_t^{(m)})_{ij} \mat_k(\cE_s^{(m)})_{lh} \Big\} \Big| \leq c .
    \end{align*}
    \item For all $k\in[K_m]$, $i,l_1\in[p_{m,k}]$, $j,h_1\in[p_{m,\text{-}k}]$, there exists some constant $c$ such that
    {\small
    \begin{align*}
        & \sum_{s=1}^T \sum_{l_2=1}^{p_{m,k}} \sum_{h_2=1}^{p_{m,\text{-}k}} \Big| \cov\Big\{ \mat_k(\cE_t^{(m)})_{ij} \mat_k(\cE_t^{(m)})_{l_1 j} , \mat_k(\cE_s^{(m)})_{ih_2} \mat_k(\cE_s^{(m)})_{l_2 h_2} \Big\} \Big| \leq c, \\
        & \sum_{s=1}^T \sum_{l_2=1}^{p_{m,k}} \sum_{h_2=1}^{p_{m,\text{-}k}} \Big| \cov\Big\{ \mat_k(\cE_t^{(m)})_{ij} \mat_k(\cE_t^{(m)})_{i h_1} , \mat_k(\cE_s^{(m)})_{l_2 j} \mat_k(\cE_s^{(m)})_{l_2 h_2} \Big\} \Big| \leq c, \\
        & \sum_{s=1}^T \sum_{l_2=1}^{p_{m,k}} \sum_{h_2=1}^{p_{m,\text{-}k}} \Big| \cov\Big\{ \mat_k(\cE_t^{(m)})_{ij} \mat_k(\cE_t^{(m)})_{l_1 h_1} , \mat_k(\cE_s^{(m)})_{ij} \mat_k(\cE_s^{(m)})_{l_2 h_2} \Big\} \Big| \leq c, \\
        & \sum_{s=1}^T \sum_{l_2=1}^{p_{m,k}} \sum_{h_2=1}^{p_{m,\text{-}k}} \Big| \cov\Big\{ \mat_k(\cE_t^{(m)})_{l_1 j} \mat_k(\cE_t^{(m)})_{i h_1} , \mat_k(\cE_s^{(m)})_{l_2 j} \mat_k(\cE_s^{(m)})_{i h_2} \Big\} \Big| \leq c .
    \end{align*}
    }
\end{enumerate}
\end{assumption}

\begin{assumption}[Dependence and moments across threads]\label{ass: dependence_moment}
For all $m\in[M]$: 
\begin{enumerate}[itemsep=0pt, label = (\alph*), left = 0pt]
    \item Any element in 
    $\{\cG_t^{(m)}, \cF_t^{(m)}, \cE_t^{(m)}\}$
    is uncorrelated to those in $\cE_t^{(n)}$ for all $t\in[T]$, $n\neq m$.
    \item For all $n\neq m$, there exists some constant $c<\infty$ such that
    \begin{align*}
    & \Bigg\| \var\Bigg\{ \frac{1}{\sqrt{T}} \sum_{t=1}^T \vec(\cX_t^{(m)}) \otimes \vec(\cE_t^{(n)}) \Bigg\} \Bigg\|_{\max} \leq c , \\
    & \Bigg\| \var\Bigg\{ \frac{1}{\sqrt{T}} \sum_{t=1}^T \vec(\cG_t^{(m)}) \otimes \vec(\cF_t^{(n)}) \Bigg\} \Bigg\|_{\max} \leq c , \\
    & \Bigg\| \var\Bigg\{ \frac{1}{\sqrt{T}} \sum_{t=1}^T \vec(\cF_t^{(m)}) \otimes \vec(\cF_t^{(n)}) \Bigg\} \Bigg\|_{\max} \leq c , \\
    & \Bigg\| \var\Bigg\{ \frac{1}{\sqrt{T}} \sum_{t=1}^T \vec(\cG_t^{(m)}) \otimes \vec(\cG_t^{(n)}) \Bigg\} \Bigg\|_{\max} \leq c .
    \end{align*}
\end{enumerate}
\end{assumption}

\begin{assumption}[Dependence and moments within threads]\label{ass: dependence_thread}
For all $m\in[M]$, $k\in[K_m]$, we assume that for any deterministic vectors $\bv\in\R^{p_{m,k}}$, $\bw\in\R^{p_{m,\text{-}k}}$, $\bg\in\R^{r_{m,k}}$, and $\bh\in\R^{r_{m,\text{-}k}}$ with $\|\bv\|, \|\bw\|, \|\bg\|, \|\bh\| =1$, there exists some constant $c$ such that: 
\begin{enumerate}[itemsep=0pt, label = (\alph*), left = 0pt]
    \item $\E\Big\| T^{-1/2} \sum_{t=1}^T \mat_k(\cF_t^{(m)}) \big\{ \bv^\trans \mat_k(\cE_t^{(m)}) \bw \big\} \Big\|_F^2 \leq c$.
    \item $\E\Big\| T^{-1/2} \sum_{t=1}^T \mat_k(\cF_t^{(m)}) \big\{ \bg^\trans \mat_k(\cG_t^{(m)}) \bh \big\} \Big\|_F^2 \leq c$.
    \item $\E\Big\| T^{-1} \sum_{t=1}^T \mat_k(\cG_t^{(m)}) \big\{ \bv^\trans \mat_k(\cE_t^{(m)}) \bw \big\} \Big\|_F^2 \leq c$.
\end{enumerate}

\end{assumption}

{Assumption~\ref{ass: loadings} characterizes the factor loading matrices.
While part~(b) is an identification convention in that it normalizes the joint loading matrix $(\bA_{m,k},\bB_{m,k})$ to be exactly half-orthogonal,
it also plays a substantive role in establishing our theoretical results.
Particularly, it also indicates that the magnitude of each column of the loading matrices grows proportionally to the dimension $p_{m,k}$ and hence the factors are pervasive \citep[e.g.][Section~2.3]{Lametal2011}. 
We mark that an asymptotically half-orthogonal setup is possible
(see the discussion in the supplement)
and would reduce to Assumption~\ref{ass: loadings}(b) by QR decomposition
where any upper triangular matrix can be absorbed into the core factor.
our canonical normalization is imposed for neatness and should be viewed as an identification convention rather than a substantive restriction, which is also seen in e.g.\ \cite{Wangetal2019}, \cite{Zhangetal2025}, and \cite{ZhangCenCho2026}.
Here, our rate of convergence together with the the exact orthogonality $\bA_{m,k}^\trans \bB_{m,k} = \zero$ let the projected global component be annihilated at first order,
thus allowing local loading estimator to have faster rate of convergence than the global loading estimator (see Theorem~\ref{thm: local_loading}). 
Weak factors as in \cite{Wangetal2019} (\text{cf.} Condition~4 within) are possible, and we decide to consider strong factors only for improving the reading experience.}

Furthermore, in standard multilevel factor models with a common global factor \citep{Choietal2023}, one may weaken the orthogonality condition between global and local loadings in Assumption~\ref{ass: loadings}(b). In our more general setting where the global factors are thread-specific and may have different tensor orders, weakening this condition would require a different approach to separate and estimate the local factors.
It is worth emphasizing that the case with identical global factors across threads (hence at least $K_1=\dots=K_M$) may permit weaker assumptions and alternative estimators, but this setting is outside the scope of this paper.

Assumption~\ref{ass: core_factor}(a) constrains the moments of core factors, while part~(b) describes the local factors and is very common in the factor model literature. 
Part~(c) complements Assumption~\ref{ass: identification} by pinpointing the behavior of global factors, which is not particularly strong; see the discussion in Section~\ref{subsec: assump_fac_satisfied} on how part~(c) can be satisfied.
{Assumption~\ref{ass: core_factor}(c) can also be viewed as a temporal non-cancellation condition.
The time average in $\bOmega_{k,n,\bi}^{(m)}$ is the only signed aggregation before squaring. 
If the covariance between $\cG_t^{(m)}$ and $\cG_t^{(n)}$ were to change sign or cancel on average over time, then $\bSigma_{m,k}$ could be rank-deficient even though the two global factors were contemporaneously correlated.
Assumption~\ref{ass: core_factor}(c) rules this out, and the time-invariant mean therein, 
together with Assumption~\ref{ass: identification}(b) and \ref{ass: dependence_moment}(a),
allow us to simplify $\bOmega_{k,n,\bi}^{(m)} = \E\big\{ \mat_k(\cX_t^{(m)}) \cdot (\cX_t^{(n)})_{\bi} \big\}$ in \eqref{eqn: def_Omega_pop}.}
Lastly, note that Assumption~\ref{ass: core_factor}(c) is not the weakest possible in terms of dependency in the global factors, since it requires $\cG_t^{(m)}$ to be correlated to all $\cG_t^{(n)}$ for $n\neq m$. It is possible to allow $\cG_t^{(m)}$ to be independent of global factors in some threads, as long as there is at least one $\cG_t^{(n)}$ with $n\neq m$ that it is correlated to. This relaxation could come at the cost of very restricted dimensionality assumptions, and hence is not pursued here.

Assumption~\ref{ass: noise} includes a set of regularity conditions for the idiosyncratic noise. Parts~(b)--(c) are 
comparable to 
Assumption~3 in \cite{barigozzi2022statistical} and Assumption~D in \cite{yu2022projected} for matrix factor models. We point out that requiring elements to have bounded fourth moment in Assumption~\ref{ass: noise} is weaker than many literature restricting eighth moment to be bounded \citep[e.g.][]{Choietal2018, ChenFan2023, Huetal2025}.

Lastly, Assumptions~\ref{ass: dependence_moment} and \ref{ass: dependence_thread} depict general moment conditions between and within threads, respectively. Assumption~\ref{ass: dependence_moment}(a) strengthens Assumption~\ref{ass: identification} such that the correlation among different threads is all pooled inside global factors, whereas Assumption~\ref{ass: dependence_moment}(b) is comparable to, for example, Assumption~4 in \cite{barigozzi2022statistical}.
As an additional remark, it is possible to generalize Assumption~\ref{ass: dependence_moment}(a) to allow for weak correlation between noise in different threads---it is the factor structure in the global component that is of interest here.
Assumption~\ref{ass: dependence_thread} significantly relaxes the dependence structure in similar model framework such as multilevel factor models. Overall, it implies the global component, local component, and noise can be pairwise correlated. A more salient feature according to part~(c) is that the correlation between global factor and noise can be quite strong, while still ensuring satisfactory theoretical results.

\subsection{Theoretical results}\label{subsec: theorem}

To appreciate the performance of parameter estimators, we first show the rate of convergence for the global loading matrix estimator as follows.

\begin{theorem}[Asymptotic consistency of global loading estimators]
\label{thm: global_loading}
Under Assumptions~\ref{ass: identification}--\ref{ass: dependence_moment}, for any $m\in[M]$, $k\in[K_m]$, there exists an exactly orthogonal matrix $\bH_{A,m,k} \in \R^{r_{m,k}\times r_{m,k}}$, i.e., $\bH_{A,m,k}^\trans \bH_{A,m,k} = \bI_{r_{m,k}}$, 
such that
\begin{equation}
\label{eqn: thm_global_loading_rotation}
\frac{1}{p_{m,k}}
\Big\| \wh\bA_{m,k} - \bA_{m,k} \bH_{A,m,k} \Big\|_F^2
= \cO_P\Big( \frac{1}{T} \Big) .
\end{equation}
\end{theorem}

Theorem~\ref{thm: global_loading} is an important step for the results hereafter. 
Note that the alignment matrix $\bH_{A,m,k}$ is exactly orthogonal, which is stronger than merely non-singular or asymptotically orthogonal rotation typically obtained in PCA-based factor analysis \citep[e.g.][see also Theorem~\ref{thm: local_loading}]{BaiNg2002, yu2022projected, BaiNg2023}. This follows from the canonical normalization in Assumption~\ref{ass: loadings}(b)
and the eigenvector alignment by using Lemma~3 of \cite{Lametal2011}.
In particular, the approach to obtain $\wh{\bA}_{m,k}$ leverages the condition that the global factor is uncorrelated to local factors and noise in different threads, which is akin to the approach in \cite{Lametal2011} and \cite{Wangetal2019} assuming the noise is white. Hence it is by no accident that our rate of $1/T$ coincides with \text{e.g.} Theorem~1 in \cite{Wangetal2019} under strong factors and the same normalization.

Note that the rate of convergence in Theorem~\ref{thm: global_loading} does not involve any cross-sectional dimension. First, in a vector factor model ($K_m=1$), such a $1/T$ rate is at least as good as the rate of convergence for the PCA-type estimator in \cite{BaiNg2002}, or even more superior when $p_{m,1}$ grows in a slower rate than $T$. For higher-order threads, i.e., $K_m>1$, it can be intuited by the difficulty in relating tensor times series with different orders, and the global factor is only identified through the correlation among threads.

In multilevel matrix factor models, \cite{Zhangetal2025} estimate the global loading matrix in a similar way and essentially the same result as our Theorem~\ref{thm: global_loading}. However, the same approach is also used in estimating the local loading matrix, resulting in the $1/T$ rate again. This overlooks the opportunity to improve the result by involving information from other modes (rows and columns in a matrix-valued data). Our PCA-type loading estimator based on \eqref{eqn: def_hat_Sigma_B_mk} resolves this. Some results are now in order, providing a theoretical guarantee for the remaining parameter estimators.

\begin{theorem}[Asymptotic consistency of local loading estimators]\label{thm: local_loading}
Let Assumptions~\ref{ass: identification}--\ref{ass: dependence_thread} hold. The following claims are true.
\begin{enumerate}[itemsep=0pt, label = (\alph*), left = 0pt]
    \item For any $m\in[M]$, $k\in[K_m]$ with $K_m> 1$, let $\wh{\bD}_{m,k}$ be the diagonal matrix consisting of the first $u_{m,k}$ eigenvalues of $\wh\bSigma_{B,m,k}$ defined in \eqref{eqn: def_hat_Sigma_B_mk}. Define
    \begin{align*}
        \wh{\bH}_{m,k} &= \frac{1}{Tp_m} \sum_{t=1}^T \mat_k(\cF_t^{(m)}) \bB_{m,\text{-}k}^\trans \wh{\bA}_{m,\text{-}k,\perp} \\
        &\;\quad
        \cdot \wh{\bA}_{m,\text{-}k,\perp}^\trans \bB_{m,\text{-}k} \mat_k(\cF_t^{(m)})^\trans \bB_{m,k}^\trans \wh{\bB}_{m,k} \wh{\bD}_{B,m,k}^{-1} .
    \end{align*}
    Then $\wh{\bH}_{m,k}$ is asymptotically orthogonal such that $\wh{\bH}_{m,k} \wh{\bH}_{m,k}^\trans = \bI + o_P(1)$. Moreover, the PCA-type estimator $\wh{\bB}_{m,k}$ satisfies
    \[
    \frac{1}{p_{m,k}} \Big\| \wh{\bB}_{m,k} - \bB_{m,k} \wh{\bH}_{m,k} \Big\|_F^2 = \cO_P\Big( \frac{1}{p_{m,k}^2} + \frac{1}{T p_{m,k}} + \frac{1}{T p_{m,\text{-}k}} \Big) .
    \]
    \item For any $m\in[M]$ with $K_m= 1$, let $\wt{\bD}_{m,1}$ be the diagonal matrix consisting of the first $u_{m,1}$ eigenvalues of $\wt\bSigma_{B,m,1}$ defined in \eqref{eqn: def_tilde_Sigma_B_m1}. Define
    \[
    \wt{\bH}_{m,1} = \frac{1}{Tp_{m,1}} \sum_{t=1}^T \mat_1(\cF_t^{(m)}) \mat_1(\cF_t^{(m)})^\trans \bB_{m,1}^\trans \wh{\bA}_{m,1,\perp} \wh{\bA}_{m,1,\perp}^\trans \bB_{m,1}^\trans \wh{\bB}_{m,1} \wt{\bD}_{B,m,1}^{-1} .
    \]
    Then $\wt{\bH}_{m,1}$ is asymptotically orthogonal, and the estimator $\wh{\bB}_{m,1}$ based on \eqref{eqn: def_tilde_Sigma_B_m1} satisfies
    \[
    \frac{1}{p_{m,1}} \Big\| \wh{\bB}_{m,1} - \bB_{m,1} \wt{\bH}_{m,1} \Big\|_F^2 = \cO_P\Big( \frac{1}{p_{m,1}^2} + \frac{1}{T} \Big) .
    \]
\end{enumerate}
\end{theorem}

\begin{theorem}[Asymptotic consistency of core factor estimators and component estimators]\label{thm: core_factor}
Let assumptions in Theorem~\ref{thm: local_loading} hold. Then for any $m\in[M]$, $t\in[T]$:
\begin{enumerate}[itemsep=0pt, label = (\alph*), left = 0pt]
    \item The local factor estimator defined in \eqref{eqn: def_hat_Ft} satisfies that
    \begin{align*}
        & \text{for } K_m> 1, \,
        \Big\| \vec(\wh{\cF}_t^{(m)}) - (\otimes_{k=1}^K \wh{\bH}_{m,k}^\trans) \vec(\cF_t^{(m)}) \Big\|^2 = \cO_P\Big( \sum_{k=1}^{K_m} \frac{1}{p_{m,k}^2} + \frac{1}{T} + \frac{1}{p_m} \Big) ; \\
        & \text{for } K_m =1, \,
        \Big\| \vec(\wh{\cF}_t^{(m)}) - \wt{\bH}_{m,1}^\trans \vec(\cF_t^{(m)}) \Big\|^2 = \cO_P\Big( \frac{1}{T} + \frac{1}{p_{m,1}} \Big) .
    \end{align*}
    \item With $\bH_{A,m,k}$ from Theorem~\ref{thm: global_loading},
    the global factor estimator in \eqref{eqn: def_hat_Gt} is consistent such that 
    \begin{align*}
    \Big\| \vec(\wh{\cG}_t^{(m)}) - \big( \otimes_{k=1}^{K_m} \bH_{A,m,k}^\trans \big) \vec(\cG_t^{(m)}) \Big\|^2
    = \cO_P\Big( \sum_{k=1}^{K_m} \frac{1}{p_{m,k}^2} + \frac{1}{T} + \frac{1}{p_m} \Big) .
    \end{align*}
    \item Both global and local component estimators are consistent such that
    \begin{align*}
        \frac{1}{p_m} \big\| \vec(\wh\cX_{G,t}^{(m)}) -\vec(\cX_{G,t}^{(m)}) \big\|^2 &= \cO_P\Big( \sum_{k=1}^{K_m} \frac{1}{p_{m,k}^2} + \frac{1}{T} + \frac{1}{p_m} \Big) , \\
        \frac{1}{p_m} \big\| \vec(\wh\cX_{F,t}^{(m)}) -\vec(\cX_{F,t}^{(m)}) \big\|^2 &= \cO_P\Big( \sum_{k=1}^{K_m} \frac{1}{p_{m,k}^2} + \frac{1}{T} + \frac{1}{p_m} \Big) .
    \end{align*}
\end{enumerate}
\end{theorem}

Theorem~\ref{thm: local_loading}(a) spells out the rate of convergence for the local loading estimator under $K_m>1$, and part~(b) under $K_m=1$. As mentioned previously, this result is generally better than our Theorem~\ref{thm: global_loading}. In fact, in a balanced-dimension scenario with $T\asymp p_{m,1}\asymp \dots \asymp p_{m,K_m}$, the rate in (a) becomes $1/T^2$ which is better than the result in Theorem~3 in \cite{Zhangetal2025}. Until $p_{m,k}$ is too small such that $p_{m,k}^2 \asymp T$, both rates in Theorem~\ref{thm: local_loading}(a) and (b) boil down to $1/T$.

Compared with the classical rate of $1/p_{m,k}^2 + 1/Tp_{m,\text{-}k}$ for PCA-type estimators \citep[e.g.][]{barigozzi2022statistical}, Theorem~\ref{thm: local_loading}(a) contains an additional $1/Tp_{m,k}$. This arises from the interaction between the projected noise and the local loading matrix, 
which partially inherits the error of the global loading estimator.
Such term would have no impact to the result when $p_{m,1}\asymp \dots \asymp p_{m,K_m}$ and $T$ is of order larger than or equal to $p_{m,1}$, which is often the case in practice.

Theorem~\ref{thm: core_factor} spells out the consistency of local/global factor/component estimators, and indicate the same rates of convergence. Although we have a superior result in Theorem~\ref{thm: local_loading}, it cannot be translated to Theorem~\ref{thm: core_factor} which suffers from the $1/T$ rate in Theorem~\ref{thm: global_loading}. This once again reflects the difficulty in parameter estimation under our model. To compare with multilevel factor models, for the vector case such that $K_m=1$ for all $m\in[M]$, our result is the same as Proposition~2 in \cite{Choietal2018}. Under the matrix case when $K_m=2$ and $p_{m,1}\asymp p_{m,2}$ for all $m\in[M]$, our derived asymptotic rates would be the same as Theorems~4 and 5 in \cite{Zhangetal2025}. In view of this, our presented results point out the behavior of parameter estimators in a multilevel factor model for tensor time series with general orders.
The core factor estimators are aligned by the same exactly orthogonal rotations that appear in the global loading estimator, so the component estimators remain rotation-invariant and require no alignment.

\begin{remark}[Asymptotic rate of the global loading estimator]
\label{remark: global_loading_rate}
We note that the rate of convergence $\cO_P(1/T)$ in Theorem~\ref{thm: global_loading} involves no cross-sectional dimensions, which differs from the loading estimation in a classical (higher-order tensor) factor model \citep{Wangetal2019, ChenFan2023}.
Also, projection-based estimators can often achieve improved rates through additional averaging over other modes \citep{yu2022projected, Hanetal2024, Babiietal2025}.

In our setting, the global loading space is identified only through contemporaneous cross-thread dependence. 
A rate involving e.g.\ $\cO_P\{ 1/ (p_{m,\text{-}k} T)\}$ would require additional within-thread averaging or a stronger pooling structure, such as identical global factors across data threads with the same tensor orders.
Under our heterogeneous and thread-specific global factor structure, the rate in Theorem~\ref{thm: global_loading} appears reasonable. 
An alternative two-stage estimation procedure is discussed in the supplement, which would not improve the rate either.
\end{remark}

\begin{remark}[Alternative estimation of the global factor]
\label{remark: alternative_global_fac}
One natural alternative to \eqref{eqn: def_hat_Gt} would estimate the global factor by projecting the vectorized observation using the orthogonal complements of the estimated local loading matrices, and then regressing on the estimated global loading matrices. 
That is, for each $m\in[M]$ and $k\in[K_m]$, if we define $\wh\bP_{B,m,\perp} := \otimes_{k=1}^{K_m} (\wh\bB_{m,k,\perp} \wh\bB_{m,k,\perp}^\trans)$,
then an alternative global factor estimator would be
\begin{align*}
    \vec(\wh\cG_t^{(m),\textsf{alt}}) &:=
    \left\{ (\otimes_{k=1}^{K_m} \wh\bA_{m,k})^{\trans} \wh\bP_{B,m,\perp} (\otimes_{k=1}^{K_m} \wh\bA_{m,k}) \right\}^{-1} 
    (\otimes_{k=1}^{K_m}  \wh\bA_{m,k})^\trans  \wh\bP_{B,m,\perp}  \vec(\cX_t^{(m)}) .
\end{align*}
Under Assumption~\ref{ass: loadings}(b), we have $\bP_{B,m,\perp} (\otimes_{k=1}^{K_m} \bA_{m,k}) = \otimes_{k=1}^{K_m} \bA_{m,k}$.
Hence we expect the alternative estimator $\vec(\wh\cG_t^{(m),\textsf{alt}})$ to be close to $\vec(\cG_t^{(m)})$.
However, the projection carries the rates of global loading estimator $\wh{\bA}_{m,k}$ which is of order $1/T$ from Theorem~\ref{thm: global_loading}, hence such alternative would hardly improve the results in Theorem~\ref{thm: core_factor}(b), while the extra projection step increases the computational cost.

More importantly, the estimated global factor absorbs a thread-specific rotation. Under Theorem~\ref{thm: global_loading}, the estimated global loading space is aligned through an orthogonal matrix $\bH_{A,m,k}$, 
which differs across threads due to its construction from $\wh\bSigma_{m,k}$. 
Hence even if one considers identical global factors such that $\cG_t^{(1)}=\cdots=\cG_t^{(M)}$,
the estimated factors are not directly comparable across threads.
This would require additional alignment e.g.\ by selecting one thread as a reference and rotating the estimated global factors from other threads towards it.
Such alignment issue is also seen in the multilevel matrix factor model \citep{Zhangetal2025}, where a common global factor is assumed but the theoretical results are established only for thread-specific ``global signal parts'', rather than for a single unified global factor. 
Estimating a common global factor jointly across threads with heterogeneous orders and dimensions is therefore a substantially more involved problem and hence not pursued in this paper, and we leave the detailed study in an ongoing work.
\end{remark}

\section{Further Discussion}
\label{sec: further_discussion}

\subsection{Estimation of number of factors}
\label{subsec: est_num_fac}

In model~\eqref{eqn: model}, the numbers of global and local factors are unknown, which need to be either pre-specified based on some prior information or appropriately estimated. For the latter, the representation in \eqref{eqn: Xt_rewrite_tfm} enlightens a feasible approach. In particular, by Remark~\ref{remark: model_read}, $\{\cX_t^{(m)}\}$ follows a classical tensor factor model with each mode-$k$ loading matrix $(\bA_{m,k}, \bB_{m,k})$ with column rank $r_{m,k} +u_{m,k}$; see also Assumption~\ref{ass: loadings}(b) in Section~\ref{subsec: assumption}. Hence any consistent factor number estimators can be employed, such as the correlation-thresholding estimator by \cite{Lam2021}, the information-criteria approach by \cite{Hanetal2022}, etc. Denote such estimated factor numbers by $\wh{s}_{m,k}$ which effectively estimates $r_{m,k} +u_{m,k}$. 

Akin to the eigenvalue-ratio estimator in \cite{Zhangetal2025} for multilevel matrix factor models, a plausible estimator on the global factor number can be obtained by minimizing the eigenvalue ratio $\lambda_{i+1}(\wh\bSigma_{m,k}) /\lambda_i (\wh\bSigma_{m,k})$ over $i\in[p_{m,k}]$. However, no theoretical guarantee is made, which is the similar case for the eigenvalue-ratio estimator briefly mentioned in \cite{Wangetal2019}. In fact, Example~1 in \cite{Zhangetal2024} demonstrate that using the direct ratio between consecutive eigenvalues could lead to inconsistent estimator. To circumvent this, suppose $r_{\max,m,k}$ is some upper bound of the factor numbers, we define
\begin{equation}
\label{eqn: est_num_fac_define}
\wh{r}_{m,k} := \argmin_{1\leq i\leq r_{\max,m,k}} \frac{\lambda_{i+1}(\wh\bSigma_{m,k}) + \xi_{m,k}}{ \lambda_i (\wh\bSigma_{m,k}) +\xi_{m,k}} , \quad \text{where} \quad
\xi_{m,k} \asymp \frac{p_m}{\sqrt{T}} \Bigg( \sum_{\substack{n=1\\ n\neq m}}^M p_n \Bigg) .
\end{equation}
In practice, we can take $r_{\max,m,k}=\lfloor p_{m,k}/2 \rfloor$ and $\xi_{m,k} = p_m(\sum_{n\in[M] \setminus \{m\}} p_n) /(5\sqrt{T})$;
{see also the discussion in parameter selection in the supplement.}
The form in \eqref{eqn: est_num_fac_define} is in the same spirit of the perturbed eigenvalue-ratio estimator in Theorem~6 in \cite{Pelger2019}. Intuitively, $\xi_{m,k}$ helps to stabilize the estimator. If one takes $\xi_{m,k}=0$, then we can only show $\P(\wh{r}_{m,k} \geq r_{m,k} ) \xrightarrow{P} 1$.
With our choice of $\xi_{m,k}$, we show in Theorem~\ref{thm: eigenratio} that $\wh{r}_{m,k}$ is consistent, in contrast to the practical estimator in \cite{Zhangetal2025}.

As discussed in Remark~\ref{remark: absent_component}, there might be no global components for some threads. To study such specification problem, for multilevel matrix factor models, \cite{Zhangetal2025} proposed to use the mock eigenvalues from \cite{ahn2013eigenvalue} to complement the factor number estimation, thereby providing an insight to the model specification. Such mock eigenvalues might not be consistent by the analogous reasoning for the plain eigenvalue-ratio estimator. In fact, the theoretical guarantee is only constructed for an arbitrary, meta form of such mock eigenvalues in Corollary~1 in \cite{ahn2013eigenvalue}.

In this paper, we put forward a refined form on the mock eigenvalue. To be specific, for each $m\in[M]$, $k\in[K_m]$, we define the perturbed mock eigenvalue as
\begin{equation}
\label{eqn: mock_evalue}
\lambda_{0}(\wh\bSigma_{m,k}) := c \Big( \xi_{m,k} \cdot \min\{p_{m,k}, T\}^{1/2} + \sum_{i=1}^{\min\{p_{m,k}, T\}} \lambda_{i}(\wh\bSigma_{m,k}) \Big) \big/ \log(\min\{p_{m,k}, T\}) ,
\end{equation}
where $c>0$ is some specified constant and $\xi_{m,k}$ is from \eqref{eqn: est_num_fac_define}. We use $c=5$ in our numerical studies. By the convention in \cite{Heetal2023}, we use $r_{m,k}=0$ for all $k\in[K_m]$ to denote the lack of global component in the $m$th thread. The following theorem presents the consistency of the global factor number estimator, regardless of the presence of global components.

\begin{theorem}\label{thm: eigenratio}
Let all assumptions in Theorem~\ref{thm: global_loading} hold. For any $m\in[M]$, $k\in[K_m]$, suppose there is global component such that $r_{m,k}\geq 1$. Then $\wh{r}_{m,k}$ defined in \eqref{eqn: est_num_fac_define} is consistent such that
\[
\P(\wh{r}_{m,k} = r_{m,k} ) \xrightarrow{P} 1.
\]
When global components can be absent with $r_{m,k}\geq 0$, let $\wh{r}_{m,k}$ be redefined with $i$ in \eqref{eqn: est_num_fac_define} starting from 0 and the perturbed mock eigenvalue from \eqref{eqn: mock_evalue}. Further assume all $p_{m,k}$'s are of the same order. Then $\wh{r}_{m,k}$ is also consistent as above.
\end{theorem}

It is worth pointing out that the dimensionality condition in Theorem~\ref{thm: eigenratio} can be relaxed at the expenses of more complicated, technical arguments. From the proof of the theorem, it actually holds if the log base in \eqref{eqn: mock_evalue} is replaced by any constant $a>1$. Note that for a fixed $m$, there could be some but not all $k\in[K_m]$ resulting in $\wh{r}_{m,k}=0$ in practice. This hints on a multi-way structure in the global component, \text{cf.} Remark~\ref{remark: absent_component}, and a detailed investigation is left to future research.
With both $\wh{s}_{m,k}$ and $\wh{r}_{m,k}$ computed, the local factor number can thus be estimated by $\wh{u}_{m,k} := \wh{s}_{m,k} -\wh{r}_{m,k}$, whose consistency is straightforward provided the consistent estimation discussed previously. Note that $\wh{u}_{m,k}=0$ would suggest the absence of local component. Following this, estimators for the local loading matrix and local factor can be trivially set as zero in Sections~\ref{subsec: est_loading} and \ref{subsec: est_fac}, and other procedures therein remain valid.

\subsection{Fast implementation}
\label{subsec: bootstrap_implement}

When threads in JXC have large dimensions and higher orders, obtaining \eqref{eqn: def_hat_Sigma_mk} can be computationally heavy. To circumvent this, recall from Theorem~\ref{thm: global_loading} that only $T$ is involved in the rate of convergence for the global loading matrix estimator. This hints on reformulating \eqref{eqn: def_hat_Sigma_mk} by using a random subset of multi-indices. Indeed, this can be formalized as follows. For each $m\in[M]$, $k\in[K_m]$, define
\begin{equation}
\label{eqn: def_hat_Sigma_mk_bootstrap}
\begin{split}
    \wh\bSigma_{m,k}^S &:= \sum_{\substack{n=1\\ n\neq m}}^M \sum_{\bi \in \cS_n} \wh\bOmega_{k,n,\bi}^{(m)} \wh\bOmega_{k,n,\bi}^{(m) \trans} ,
\end{split}
\end{equation}
where $\cS_n \subseteq [p_{n,1}] \times \cdots \times[p_{n,K_n}]$ is a random set drawn independently of JXC, and $\wh\bOmega_{k,n,\bi}^{(m)}$ is the same as in \eqref{eqn: def_hat_Sigma_mk}. Then the global loading matrix can be estimated by $\wh{\bA}_{m,k}^S$, defined as $\sqrt{p_{m,k}}$ times the $r_{m,k}$ leading eigenvectors of $\wh\bSigma_{m,k}^S$. More importantly, the performance of such estimator can be made mathematically rigorous, as shown below.

\begin{theorem}
\label{prop: global_loading_bootstrap}
Let all assumptions in Theorem~\ref{thm: global_loading} hold. Fix $m$, and for any $n\in[M] \setminus\{m\}$, let $\cS_n^\dagger \subseteq [p_n]$ denote the set of row indices of the matrix $\otimes_{k=1}^{K_n} \bA_{n,k}$ that corresponds to the rows of $\{\bA_{n,1}, \dots, \bA_{n,K_n}\}$ with indices in $\cS_n$ from \eqref{eqn: def_hat_Sigma_mk_bootstrap}. For any $j\in[r_n]$, assume that
\[
\sigma_{j}\big\{ (\otimes_{k=1}^{K_n} \bA_{n,k} )_{\cS_n^\dagger} \big\} \asymp |\cS_n^\dagger |^{1/2} ,
\]
where $(\otimes_{k=1}^{K_n} \bA_{n,k} )_{\cS_n^\dagger}$ represents the sub-matrix of $\otimes_{k=1}^{K_n} \bA_{n,k}$ with rows restricted on $\cS_n^\dagger$. 
Then for any $k\in[K_m]$, there exists an exactly orthogonal matrix $\bH_{A,m,k}^S\in \R^{r_{m,k}\times r_{m,k}}$ such that
\begin{align*}
    \frac{1}{p_{m,k}} \Big\| \wh{\bA}_{m,k}^S -\bA_{m,k} \bH_{A,m,k}^S \Big\|_F^2 = \cO_P\Big( \frac{1}{T} \Big) .
\end{align*}
\end{theorem}

The additional assumption in Theorem~\ref{prop: global_loading_bootstrap} can be regarded as a rank condition to ensure $\wh\bSigma_{m,k}^S$ is well behaved. By Assumption~\ref{ass: loadings}(b), it is automatically fulfilled if $\cS_n^\dagger =[p_n]$, i.e., $\wh\bSigma_{m,k}^S =\wh\bSigma_{m,k}$. Through our numerical studies in Section~\ref{subsec: simulation}, this additional assumption can be satisfied in general. It is also worth pointing out that $\wh\bSigma_{m,k}^S$ can replace $\wh\bSigma_{m,k}$ to compute the factor number estimator in \eqref{eqn: est_num_fac_define}, except that $\xi_{m,k}$ therein needs to fulfill
\[
\xi_{m,k} \asymp \frac{p_m}{\sqrt{T}} \Bigg( \sum_{\substack{n=1\\ n\neq m}}^M |\cS_n^\dagger | \Bigg) .
\]
Then the resulted estimator is consistent as in Theorem~\ref{thm: eigenratio}.

\subsection{*How Assumption~\ref{ass: core_factor}(c) can be implied}
\label{subsec: assump_fac_satisfied}

In what follows, we showcase how Assumption~\ref{ass: core_factor}(c) can be satisfied.
First of all, note that the matrix in the expectation within Assumption~\ref{ass: core_factor}(c) has dimension $r_{m,k} \times (r_n r_{m,\text{-}k})$, so the assumption can be treated as a full rank condition
on the global factors (under heterogeneous tensor orders across threads).
Inspired by this observation, consider the special case of multilevel vector factor model, i.e., $1= K_1=\dots=K_M$. 
As a natural setup, let the global factors be the same over threads, viz.\ $\cG_t^{(1)} = \dots =\cG_t^{(M)}$.
By direct simplification,
Assumption~\ref{ass: core_factor}(c) is equivalent to $\E\big\{ \cG_t^{(m)} (\cG_t^{(m)})^\trans \big\}$ being positive definite with eigenvalues bounded away from zero and infinity.

In general, Assumption~\ref{ass: core_factor}(c) can be replaced by simpler assumptions for any tensor orders. 
We present the result in the following claim. Note that $K_1=\dots=K_M$ is not even required in the second result in the claim, and the assumption that elements in the global factors are uncorrelated is not necessary (e.g., as shown for the vector case above) and can be implied by Assumption~(F1) in \cite{CenLam2025} for instance. The proof of the claim, together with more discussion on Assumption~\ref{ass: core_factor}(c), are deferred to the online supplement.

\begin{claim}
\label{claim: assump_fac_simpler}
Let Assumption~\ref{ass: loadings}(b) hold.
Under $K_1=\dots=K_M$ and $\cG_t^{(1)} = \dots =\cG_t^{(M)}$, if $\vec(\cG_t^{(M)})$ has uncorrelated elements with zero mean and variance $\otimes_{k\in[K_M]} \bSigma_{G,k}$ with bounded elements in each $\bSigma_{G,k}\in \R^{r_{M,k}\times r_{M,k}}$, then Assumption~\ref{ass: core_factor}(c) is satisfied.
When each global factor 
is weakly stationary and 
admits the form $\cG_t^{(m)} =\map(\cG_t, V_m)$ as in \eqref{eqn: model_map_Gt} (Section~\ref{subsec: simulation}), Assumption~\ref{ass: core_factor}(c) is also satisfied.
\end{claim}

\section{Numerical Analysis}
\label{sec: numerical}

\subsection{Simulation}
\label{subsec: simulation}

In the following, we demonstrate the numerical performance of the estimators described in Sections~\ref{subsec: est_loading} and \ref{subsec: est_fac} using Monte Carlo experiments. As this paper includes multilevel factor model as a special case, we will also make comparison with \cite{Zhangetal2025} when JXC only contains matrix time series.

To start with,
using the notation of tensor map in Appendix~\ref{sec: intro_map}, 
we consider a natural choice that $\cG_t^{(m)} = \map(\cG_t, V_m)$ for some order-$K_M$ tensor $\cG_t$ with dimension $r_1\times \dots\times r_{K_M}$, and some pre-specified channel $V_m$.
Essentially, $\cG_t^{(m)} = \map(\cG_t, V_m)$ is some re-organization of $\cG_t$.
This can be useful when the meaning of modes is related across threads with different orders. For example, one time series may contain the fMRI scans of brains over time, while another reports vital signs of patients such as heart rate, blood pressure, and respiratory rate. Then the sign series is potentially driven by the vectorized core factor for the fMRI series, representing the latent health profiles.
Without loss of generality, let $V_M=\big\{\{1\}, \dots, \{K_M\} \big\}$. We may read \eqref{eqn: model} as
\begin{equation}
\label{eqn: model_map_Gt}
\cX_t^{(m)} = \map(\cG_t, V_m) \times_{k=1}^{K_m} \bA_{m,k} + \cF_t^{(m)} \times_{\ell=1}^{K_m} \bB_{m,\ell} + \cE_t^{(m)} ,
\end{equation}
which fulfills Assumption~\ref{ass: core_factor}(c) by Claim~\ref{claim: assump_fac_simpler}.
{In the supplement, additional experiments using other types of global factor specifications show similar patterns as below.}

We generate each thread as in \eqref{eqn: model_map_Gt}.
Each entry of the global and local factors is independent autoregressive process AR(1), with coefficient $0.5$ and innovation following \text{i.i.d.} $\cN(0,1)$. The entries of $\cE_t^{(m)}$ is constructed similarly, except that the innovation is either $\cN(0,1)$ or $t_6$, which will be specified later. Each global loading matrix is generated as $\bA_{m,k} = \bU_{m,k}$ for pervasive factors and $\bA_{m,k} = p_{m,k}^{-0.1} \cdot \bU_{m,k}$ for weak factors 
{(which intentionally violates Assumption~\ref{ass: loadings}(b) to assess sensitivity to factor strength),} 
where $\bU_{m,k}\in \R^{p_{m,k}\times r_{m,k}}$ consists of \text{i.i.d.} $\cN(0,1)$ elements. Let $\bA_{m,k,\perp}$ be the half-orthogonal matrix representing the orthogonal complement of $\bA_{m,k}$. Then the local loading matrix is constructed as $p_{m,k}^{0.5}$ (or $p_{m,k}^{0.4}$ for weak factors) times a $p_{m,k}\times u_{m,k}$ sub-matrix of $\bA_{m,k,\perp}$. For simplicity, we only consider one factor for all modes and threads.

As the loading matrices are only identified up to an arbitrary rotation, we measure the distance between the column space spanned by the true loading matrix and the estimator. Formally, given any pair of matrices $\bX$ and $\bY$, the column space distance is computed as
{\small
\[ \cD(\bX, \bY) := \Big\| \bX(\bX^\trans \bX)^{-1} \bX^\trans -\bY(\bY^\trans \bY)^{-1} \bY^\trans \Big\| . \]}

Due to the rotational indeterminacy in the loading matrices, performance of the core factor estimator cannot be directly compared, but reflected in the estimated global and local components. We use the (relative) mean squared error to measure the performance of the component estimators:
\begin{align*}
    & \text{MSE}\big(\wh\cX_{G,t}^{(m)} \big) := \frac{ \sum_{t=1}^{T} \big\| \wh\cX_{G,t}^{(m)} - \cX_{G,t}^{(m)} \big\|_F^2}{ \sum_{t=1}^{T} \big\| \cX_{G,t}^{(m)} \big\|_F^2} , \quad
    \text{MSE}\big(\wh\cX_{F,t}^{(m)} \big) := \frac{\sum_{t=1}^{T} \big\| \wh\cX_{F,t}^{(m)} - \cX_{F,t}^{(m)} \big\|_F^2}{ \sum_{t=1}^{T} \big\| \cX_{F,t}^{(m)} \big\|_F^2} .
\end{align*}

We experiment a variety of settings, detailed as follows. Each setting is repeated 400 times, unless otherwise stated.

\vspace{12pt}
\noindent\textbf{Setting~A:}
\begin{enumerate}[itemsep=-1pt, parsep=0pt, topsep=0pt, label = (A.\arabic*), left=1pt]
    \item Number of threads $M=3$ with $(K_1,K_2,K_3)=(1,1,2)$. Fix the cross-sectional dimensions $(p_{1,k}, p_{2,h}, p_{3,l}) =(30, 30, 10)$, and the sample size is experimented through $T\in\{100, 400\}$. Factors are strong and the noise innovation is Gaussian.
    \item Same as (A.1), but $(K_1,K_2,K_3)=(1,2,3)$ and $(p_{1,k}, p_{2,h}, p_{3,l}) =(30, 15, 10)$.
    \item Same as (A.1), except that factors are weak.
    \item Same as (A.1), except that the noise innovation is $t$ distributed.
\end{enumerate}
\vspace{12pt}

\noindent\textbf{Setting~B:}
\begin{enumerate}[itemsep=-1pt, parsep=0pt, topsep=0pt, label = (B.\arabic*), left=1pt]
    \item Number of threads $M=2$ with $(K_1,K_2)=(1,2)$. Fix dimension $(p_{1,k}, p_{2,h}) =(30, 10)$, and the sample size is experimented through $T\in\{100, 400\}$. Factors are strong and the noise innovation is Gaussian.
    \item Same as (B.1), except that $(p_{1,k}, p_{2,h}) =(100, 20)$.
    \item Same as (B.1), except that $(K_1,K_2)=(1,3)$ and $(p_{1,k}, p_{2,h}) =(30, 10)$.
    \item Same as (B.1), except that $(K_1,K_2)=(2,3)$ and $(p_{1,k}, p_{2,h}) =(10, 10)$.
\end{enumerate}
\vspace{12pt}

Settings~A and B consider $M=3$ and $M=2$, respectively. In summary, (A.1) showcases the scenario when some threads can have the same order but different from another thread, while the JXC in (A.2) has threads with different orders. (A.3) and (A.4) respectively illustrate the sensitivity of (A.1) under weak factors and heavy-tailed noise. In Setting~B, we demonstrate scenarios with different combination of thread orders. (B.2) also investigates the effect of increasing cross-sectional dimension to the estimator performance.
For computational concern, we estimate the global loading matrix using $\wh{\bA}_{m,k}^S$ discussed in Section~\ref{subsec: bootstrap_implement}, with $\sum_{n\in[M] \setminus\{m\}}|\cS_n^\dagger| =50$ in Setting~A, and $\sum_{n\in[M] \setminus\{m\}}|\cS_n^\dagger| =30$ in Setting~B which is the maximum possible in some settings.

\begin{table}[htp!]\centering
\caption{Results of Settings~(A.1)--(A.4). Each cell shows the mean of measure over 400 runs for the corresponding setting.}
\label{tab: simulation_setting_A}
\small
\begin{tabular}{@{}crrrrrrrrrrrrrrrrrrrrrrrrr@{}}
\hline 
Setting A & \multicolumn{2}{c}{(A.1)} && \multicolumn{2}{c}{(A.2)} && \multicolumn{2}{c}{(A.3)} && \multicolumn{2}{c}{(A.4)}  \\
$T$
&  $100$ & $400$ && $100$ & $400$ && $100$ & $400$
&& $100$ & $400$ \\
\hline 
$\cD(\bA_{1,1}, \wh{\bA}_{1,1}^S)$ &  0.189 & 0.088 && 0.196 & 0.095  && 0.237 & 0.111 && 0.216 & 0.101  \\
$\cD(\bA_{2,1}, \wh{\bA}_{2,1}^S)$ &  0.194 & 0.092 && 0.039 & 0.019  && 0.243 & 0.116 && 0.217 & 0.105 \\
$\cD(\bA_{2,2}, \wh{\bA}_{2,2}^S)$ &  - & - && 0.045 & 0.020  && - & - && - & -  \\
$\cD(\bA_{3,1}, \wh{\bA}_{3,1}^S)$ &  0.045 & 0.022 && 0.018 & 0.008  && 0.073 & 0.035 && 0.055 & 0.027 \\
$\cD(\bA_{3,2}, \wh{\bA}_{3,2}^S)$ &  0.046 & 0.022 && 0.016 & 0.008  && 0.075 & 0.035 && 0.055 & 0.027 \\
$\cD(\bA_{3,3}, \wh{\bA}_{3,3}^S)$ &  - & - && 0.018 & 0.008  && - & - && - & -  \\
\\
$\cD(\bB_{1,1}, \wh{\bB}_{1,1})$ &  0.186 & 0.087 && 0.191 & 0.093  && 0.232 & 0.109 && 0.211 & 0.100  \\
$\cD(\bB_{2,1}, \wh{\bB}_{2,1})$ &  0.192 & 0.089 && 0.038 & 0.016  && 0.240 & 0.112 && 0.216 & 0.102 \\
$\cD(\bB_{2,2}, \wh{\bB}_{2,2})$ &  - & - && 0.033 & 0.016  && - & - && - & -  \\
$\cD(\bB_{3,1}, \wh{\bB}_{3,1})$ &  0.040 & 0.020 && 0.013 & 0.006  && 0.068 & 0.034 && 0.049 & 0.024 \\
$\cD(\bB_{3,2}, \wh{\bB}_{3,2})$ &  0.040 & 0.020 && 0.013 & 0.006  && 0.068 & 0.033 && 0.049 & 0.024 \\
$\cD(\bB_{3,3}, \wh{\bB}_{3,3})$ &  - & - && 0.013 & 0.006  && - & - && - & -   \\
\\
$\text{MSE}\big(\wh\cX_{F,t}^{(1)} \big)$ &  0.095 & 0.046 && 0.106 & 0.051  && 0.152 & 0.084 && 0.125 & 0.066  \\
$\text{MSE}\big(\wh\cX_{F,t}^{(2)} \big)$ &  0.098 & 0.047 && 0.149 & 0.005  && 0.159 & 0.085 && 0.130 & 0.067 \\
$\text{MSE}\big(\wh\cX_{F,t}^{(3)} \big)$ &  0.014 & 0.011 && 0.002 & 0.001  && 0.036 & 0.028 && 0.021 & 0.016 \\
$\text{MSE}\big(\wh\cX_{G,t}^{(1)} \big)$ &  0.111 & 0.049 && 0.122 & 0.054  && 0.173 & 0.088 && 0.143 & 0.069 \\
$\text{MSE}\big(\wh\cX_{G,t}^{(2)} \big)$ &  0.114 & 0.052 && 0.229 & 0.008  && 0.181 & 0.094 && 0.149 & 0.075 \\
$\text{MSE}\big(\wh\cX_{G,t}^{(3)} \big)$ &  0.020 & 0.016 && 0.006 & 0.002  && 0.052 & 0.040 && 0.030 & 0.024  \\
\hline 
\end{tabular}
\end{table}

\begin{table}[htp!]\centering
\caption{Results of Settings~(B.1)--(B.4). Each cell shows the mean of measure over 400 runs for the corresponding setting.}
\label{tab: simulation_setting_B}
\small
\begin{tabular}{@{}crrrrrrrrrrrrrrrrrrrrrrrrr@{}}
\hline 
Setting B & \multicolumn{2}{c}{(B.1)} && \multicolumn{2}{c}{(B.2)} && \multicolumn{2}{c}{(B.3)} && \multicolumn{2}{c}{(B.4)}  \\
$T$
&  $100$ & $400$ && $100$ & $400$ && $100$ & $400$
&& $100$ & $400$ \\
\hline 
$\cD(\bA_{1,1}, \wh{\bA}_{1,1}^S)$ &  0.203 & 0.097 && 0.202 & 0.094  && 0.212 & 0.103 && 0.057 & 0.027  \\
$\cD(\bA_{1,2}, \wh{\bA}_{1,2}^S)$ &  - & - && - & -  && - & - && 0.052 & 0.026  \\
$\cD(\bA_{2,1}, \wh{\bA}_{2,1}^S)$ &  0.046 & 0.023 && 0.032 & 0.015  && 0.022 & 0.008 && 0.017 & 0.008 \\
$\cD(\bA_{2,2}, \wh{\bA}_{2,2}^S)$ &  0.049 & 0.023 && 0.032 & 0.015  && 0.022 & 0.008 && 0.017 & 0.009  \\
$\cD(\bA_{2,3}, \wh{\bA}_{2,3}^S)$ &  - & - && - & -  && 0.021 & 0.008 && 0.024 & 0.009  \\
\\
$\cD(\bB_{1,1}, \wh{\bB}_{1,1})$ &  0.196 & 0.093 && 0.193 & 0.091  && 0.201 & 0.097 && 0.042 & 0.020  \\
$\cD(\bB_{1,2}, \wh{\bB}_{1,2})$ &  - & - && - & -  && - & - && 0.040 & 0.019  \\
$\cD(\bB_{2,1}, \wh{\bB}_{2,1})$ &  0.040 & 0.020 && 0.028 & 0.014  && 0.017 & 0.006 && 0.015 & 0.006 \\
$\cD(\bB_{2,2}, \wh{\bB}_{2,2})$ &  0.039 & 0.020 && 0.028 & 0.014  && 0.017 & 0.006 && 0.014 & 0.006  \\
$\cD(\bB_{2,3}, \wh{\bB}_{2,3})$ &  - & - && - & -  && 0.017 & 0.006 && 0.013 & 0.006  \\
\\
$\text{MSE}\big(\wh\cX_{F,t}^{(1)} \big)$ &  0.106 & 0.052 && 0.084 & 0.024  && 0.120 & 0.055 && 0.016 & 0.011  \\
$\text{MSE}\big(\wh\cX_{F,t}^{(2)} \big)$ &  0.014 & 0.011 && 0.004 & 0.003  && 0.007 & 0.001 && 0.004 & 0.001 \\
$\text{MSE}\big(\wh\cX_{G,t}^{(1)} \big)$ &  0.127 & 0.059 && 0.099 & 0.026  && 0.140 & 0.061 && 0.028 & 0.018 \\
$\text{MSE}\big(\wh\cX_{G,t}^{(2)} \big)$ &  0.021 & 0.018 && 0.006 & 0.003  && 0.125 & 0.002 && 0.009 & 0.002 \\
\hline 
\end{tabular}
\end{table}

Results for Settings~A and B are reported in Tables~\ref{tab: simulation_setting_A} and \ref{tab: simulation_setting_B}, respectively. First of all, from $T=100$ to $T=400$, the improvement of global loading matrix estimators in all settings corroborates with the rate of convergence presented in Theorem~\ref{thm: global_loading} and Theorem~\ref{prop: global_loading_bootstrap}. In Settings~(A.3) and (A.4), all estimation errors are slight inflated, compared to those in Setting~(A.1), but still remains satisfactory. This shows our parameter estimators are quite robust to different factor strengths and heavy-tailed distributed noise. Secondly, note that errors of loading matrices are similar when they are in threads with the same order and dimensionality, \text{e.g.} $\cD(\bA_{1,1}, \wh{\bA}_{1,1}^S)$ and $\cD(\bA_{2,1}, \wh{\bA}_{2,1}^S)$ in (A.1) compared to $\cD(\bA_{1,1}, \wh{\bA}_{1,1}^S)$ in (A.2). It is consistent with the fact that in Section~\ref{subsec: theorem}, theoretical results for parameters in one thread involve no dimensionality from other threads. The interpretation in Table~\ref{tab: simulation_setting_B} is analogous to this, and performance is generally better off with larger dimensionality from the results in Setting~(B.2).
{For example, $\cD(\bA_{1,1}, \wh\bA_{1,1}^S)$ remains nearly unchanged from Setting~(B.1) to Setting~(B.2) even though the cross-sectional dimensions increase substantially, which is consistent with the rate spelled out in Theorem~\ref{thm: global_loading} and Theorem~\ref{prop: global_loading_bootstrap}.}

As mentioned previously, our proposed framework includes the multilevel matrix factor model by \cite{Zhangetal2025} as a special case. From Remark~\ref{remark: hat_Sigma_rewrite}, our global loading estimator described in Section~\ref{subsec: est_loading} is the same as theirs and hence would give the same result. In contrast, our local loading estimator is generally superior to theirs, according to Theorem~\ref{thm: local_loading}. We verify this by numerical results shown below. For all settings on this, we fix $M=2$ with $K_1=K_2=2$, and set $T=40$. All factors are pervasive and the noise has Gaussian innovation. We consider two settings representing different noise dynamics, specified below.

\vspace{12pt}
\noindent\textbf{Setting~C:}
\begin{enumerate}[itemsep=-1pt, parsep=0pt, topsep=0pt, label = (C.\arabic*), left=1pt]
    \item Cross-sectional dimensions are the same such that $p_{1,1}=p_{1,2}=p_{2,1}=p_{2,2}$, and the noise is white such that its AR coefficient is zero. We experiment through $p_{1,1}\in \{10, 40\}$.
    \item Same as (C.1), except that the AR coefficient of noise is 0.5.
\end{enumerate}
\vspace{12pt}

\begin{table}[t!]\centering
\caption {Results of local loading estimators under Settings~(C.1)--(C.2), using methods proposed in this paper and in \cite{Zhangetal2025}. Each cell shows the mean of measure over 400 runs for the corresponding setting.}
\label{tab: simulation_setting_C}
\small
\begin{tabular}{@{}crrrrrrrrrrrrrrrrrrrrrrrrr@{}}
\hline 
& \multicolumn{5}{c}{Our method} && \multicolumn{5}{c}{Zhang et al.~(2025)} \\
\hline 
Setting C & \multicolumn{2}{c}{(C.1)} && \multicolumn{2}{c}{(C.2)} && \multicolumn{2}{c}{(C.1)} && \multicolumn{2}{c}{(C.2)}  \\
$p_{1,1}$
&  $10$ & $40$ && $10$ & $40$ && $10$ & $40$
&& $10$ & $40$ \\
\hline 
$\cD(\bB_{1,1}, \wh{\bB}_{1,1})$ &  0.076 & 0.022 && 0.100 & 0.037  && 0.108 & 0.042 && 0.159 & 0.071  \\
$\cD(\bB_{1,2}, \wh{\bB}_{1,2})$ &  0.083 & 0.025 && 0.106 & 0.042  && 0.110 & 0.043 && 0.160 & 0.070  \\
$\cD(\bB_{2,1}, \wh{\bB}_{2,1})$ &  0.075 & 0.027 && 0.140 & 0.047  && 0.111 & 0.045 && 0.178 & 0.070  \\
$\cD(\bB_{2,2}, \wh{\bB}_{2,2})$ &  0.086 & 0.029 && 0.113 & 0.034  && 0.113 & 0.044 && 0.178 & 0.071  \\
\\
Run Time (s) &  0.073 & 0.091 && 0.108 & 0.089  && 22.347 & 466.275 && 33.267 & 461.613  \\
\hline 
\end{tabular}
\end{table}

Table~\ref{tab: simulation_setting_C} presents the results for Setting~C, which clearly shows the better performance of our method both in estimation accuracy and computational cost. Note that although both methods have lower estimation errors when $p_{1,1}$ increases from 10 to 40, our method enjoys faster decay on the errors, benefiting from the explicit cross-sectional dimensions involved in the convergence rates derived in Theorem~\ref{thm: local_loading}. Moreover, it is not surprising that the estimator in \cite{Zhangetal2025} suffers more in Setting~(C.2), since their method, albeit with theoretical guarantees established under weak serial correlation in the idiosyncratic noise, essentially relied on the noise being white as in \cite{Wangetal2019}.

Lastly, we demonstrate the performance of our global factor number estimator defined in \eqref{eqn: est_num_fac_define}, and the version by including the mock eigenvalue in \eqref{eqn: mock_evalue}. We adopt the choice of $r_{\max,m,k}$ and $\xi_{m,k}$ below \eqref{eqn: est_num_fac_define}, and consider settings described as follows.

\vspace{12pt}
\noindent\textbf{Setting~D:}
\begin{itemize}
    \item [(D.1--2)] Same as (B.1)--(B.2), except that the each mode in the second thread has two global factors and therefore the first thread has four global factors. Besides the standard setup with strong factors and Gaussian noise, we also experiment weak factors and $t_6$ noise. The global factor number estimator is from \eqref{eqn: est_num_fac_define}.
    \item [(D.3)] Number of threads $M=3$ with each $K_m=2$ and $p_{m,k}=10$. The first thread has no global component, i.e., $r_{1,1}=0$. Factors are strong and the noise innovation is Gaussian. The sample size is experimented through $T\in\{400, 800\}$. The global factor number estimator is \eqref{eqn: est_num_fac_define} with $i$ starting from 0.
    \item [(D.4)] Same as (D.3), except that factors are weak.
\end{itemize}
\vspace{12pt}

\begin{table}[htp!]\centering
\caption{Performance of global factor number estimator in Settings~(D.1)--(D.2). Each cell is formatted as $a \,(b \mid c)$, where $a$, $b$, and $c$ respectively represent the correct, underestimation, and overestimation proportion (in $100\%$) over 400 runs for the corresponding setting. We use ``Strong'', ``Weak'', ``$\cN$'', and ``$t_6$'' to respectively denote settings with strong factors, weak factors, Gaussian noise, and heavy-tailed noise. In each setting, the first row displays results with constant $1/5$ used in $\xi_{m,k}$ and the second displays those using $\xi_{m,k}=0$.}
\label{tab: simulation_setting_D}
\small
\begin{tabular}{@{}crrrrrrrrrrrrrrrrrrrrrrrrr@{}}
\hline 
Setting D & \multicolumn{2}{c}{(D.1)} && \multicolumn{2}{c}{(D.2)}  \\
$T$
&  $100$ & $400$ && $100$ & $400$ \\
\hline 
\multicolumn{1}{c}{Strong, $\cN$}  \\
$r_{1,1}$ &  60.8 (32.5 $\mid$ 06.7)  & 87.5 (08.5 $\mid$ 04.0) && 77.3 (15.7 $\mid$ 07.0) & 93.8 (02.3 $\mid$ 04.0)  \\
          &  56.0 (11.0 $\mid$ 33.0)  & 77.8 (01.2 $\mid$ 21.0) && 00.0 (00.0 $\mid$ 100.0) & 00.0 (00.0 $\mid$ 100.0) \\
$r_{2,1}$ &  73.5 (04.0 $\mid$ 22.5) & 93.0 (01.0 $\mid$ 06.0) && 63.0 (00.5 $\mid$ 36.5) & 90.8 (00.0 $\mid$ 09.2) \\
          &  71.5 (04.0 $\mid$ 24.5) & 92.5 (00.8 $\mid$ 06.7) && 60.3 (00.2 $\mid$ 39.5) & 90.5 (00.0 $\mid$ 09.5) \\
$r_{2,2}$ &  70.8 (04.8 $\mid$ 24.5) & 96.5 (00.0 $\mid$ 03.5) && 63.3 (00.2 $\mid$ 36.5) & 91.3 (00.0 $\mid$ 08.7) \\
          &  69.5 (04.8 $\mid$ 25.7)  & 96.0 (00.0 $\mid$ 04.0) && 59.5 (00.8 $\mid$ 39.7) & 90.8 (00.0 $\mid$ 09.2) \\
\\
\multicolumn{1}{c}{Strong, $t_6$} \\
$r_{1,1}$ &  51.3 (39.8 $\mid$ 09.0) & 85.3 (09.5 $\mid$ 05.2) && 71.5 (21.0 $\mid$ 07.5) & 92.8 (02.5 $\mid$ 04.7) \\
          &  48.8 (20.0 $\mid$ 31.2)  & 83.0 (02.3 $\mid$ 14.7) && 00.0 (00.0 $\mid$ 100.0) & 00.0 (00.0 $\mid$ 100.0) \\
$r_{2,1}$ &  66.5 (10.0 $\mid$ 23.5) & 96.3 (01.0 $\mid$ 02.7) && 70.8 (01.2 $\mid$ 28.0) & 95.3 (00.0 $\mid$ 04.7) \\
          &  65.3 (09.8 $\mid$ 25.0) & 95.8 (01.0 $\mid$ 03.2) && 67.3 (01.2 $\mid$ 31.5) & 94.5 (00.3 $\mid$ 05.2) \\
$r_{2,2}$ &  70.0 (08.5 $\mid$ 21.5) & 98.0 (00.8 $\mid$ 01.2) && 69.5 (00.8 $\mid$ 29.7) & 94.8 (00.0 $\mid$ 05.2)  \\
          &  68.8 (08.0 $\mid$ 23.2) & 97.5 (00.8 $\mid$ 01.7) && 68.8 (00.2 $\mid$ 31.0) & 96.3 (00.0 $\mid$ 03.7) \\
\\
\multicolumn{1}{c}{Weak, $\cN$} \\
$r_{1,1}$ &  29.3 (68.5 $\mid$ 02.2) & 69.8 (29.2 $\mid$ 01.0) && 34.8 (64.2 $\mid$ 01.0) & 75.5 (22.0 $\mid$ 02.5) \\
          &  43.5 (34.8 $\mid$ 21.7) & 78.8 (05.5 $\mid$ 15.7) && 00.0 (00.0 $\mid$ 100.0) & 00.0 (00.0 $\mid$ 100.0) \\
$r_{2,1}$ &  72.5 (16.3 $\mid$ 11.2) & 97.5 (02.5 $\mid$ 00.0) && 80.8 (04.7 $\mid$ 14.5) & 99.0 (00.0 $\mid$ 01.0) \\
          &  71.8 (15.5 $\mid$ 12.7) & 97.0 (02.5 $\mid$ 00.5) && 78.0 (04.2 $\mid$ 17.8) & 98.8 (00.0 $\mid$ 01.2) \\
$r_{2,2}$ &  74.8 (13.0 $\mid$ 12.2) & 98.3 (01.5 $\mid$ 00.2) && 80.8 (03.5 $\mid$ 15.7) & 99.5 (00.0 $\mid$ 00.5)  \\
          &  73.8 (12.0 $\mid$ 14.2) & 98.3 (01.5 $\mid$ 00.2) && 82.0 (03.0 $\mid$ 15.0) & 98.8 (00.0 $\mid$ 01.2) \\
\hline 
\end{tabular}
\end{table}

\begin{table}[htp!]\centering
\caption{Performance of global factor number estimator in Settings~(D.3)--(D.4). Refer to Table~\ref{tab: simulation_setting_D} for the meaning of cells.}
\label{tab: simulation_setting_D_mock}
\small
\begin{tabular}{@{}crrrrrrrrrrrrrrrrrrrrrrrrr@{}}
\hline 
Setting D & \multicolumn{2}{c}{(D.3)} && \multicolumn{2}{c}{(D.4)}  \\
$T$
&  $400$ & $800$ && $400$ & $800$ \\
\hline 
$r_{1,1}$ &  86.8 (00.0 $\mid$ 13.2)  & 88.2 (00.0 $\mid$ 11.8) && 100.0 (00.0 $\mid$ 00.0) & 99.2 (00.0 $\mid$ 00.8)  \\
          &  75.5 (00.0 $\mid$ 24.5)  & 80.5 (00.0 $\mid$ 19.5) && 98.5 (00.0 $\mid$ 01.5) & 99.0 (00.0 $\mid$ 01.0) \\
$r_{1,2}$ &  86.0 (00.0 $\mid$ 14.0)  & 89.2 (00.0 $\mid$ 10.8) && 100.0 (00.0 $\mid$ 00.0) & 100.0 (00.0 $\mid$ 00.0)  \\
          &  76.5 (00.0 $\mid$ 23.5)  & 81.2 (00.0 $\mid$ 18.8) && 99.5 (00.0 $\mid$ 00.5) & 98.2 (00.0 $\mid$ 01.8) \\
\\
$r_{2,1}$ &  88.0 (06.5 $\mid$ 05.5)  & 95.0 (02.8 $\mid$ 02.2) && 73.5 (26.2 $\mid$ 00.3) & 86.0 (14.0 $\mid$ 00.0)  \\
          &  86.5 (05.0 $\mid$ 08.5)  & 93.5 (02.0 $\mid$ 04.5) && 74.8 (25.0 $\mid$ 00.2) & 89.5 (10.2 $\mid$ 00.3) \\
$r_{2,2}$ &  89.5 (05.5 $\mid$ 05.0)  & 94.8 (02.7 $\mid$ 02.5) && 74.2 (25.8 $\mid$ 00.0) & 84.5 (15.5 $\mid$ 00.0)  \\
          &  85.8 (04.2 $\mid$ 10.0)  & 91.8 (02.0 $\mid$ 06.2) && 75.5 (23.2 $\mid$ 01.3) & 89.5 (09.5 $\mid$ 01.0) \\
\\
$r_{3,1}$ &  88.0 (05.8 $\mid$ 06.2)  & 95.8 (03.5 $\mid$ 00.7) && 73.0 (26.8 $\mid$ 00.2) & 87.2 (12.8 $\mid$ 00.0)  \\
          &  84.0 (03.5 $\mid$ 12.5)  & 93.2 (02.0 $\mid$ 04.8) && 76.0 (22.0 $\mid$ 02.0) & 89.0 (10.5 $\mid$ 00.5) \\
$r_{3,2}$ &  86.8 (08.2 $\mid$ 05.0)  & 91.8 (03.7 $\mid$ 04.5) && 73.0 (26.5 $\mid$ 00.5) & 85.0 (14.8 $\mid$ 00.2)  \\
          &  85.0 (05.2 $\mid$ 09.8)  & 92.0 (02.0 $\mid$ 06.0) && 76.2 (23.3 $\mid$ 00.5) & 89.2 (10.8 $\mid$ 00.0) \\
\hline 
\end{tabular}
\end{table}

Results of our estimator and those using $\xi_{m,k}=0$ are all included in Table~\ref{tab: simulation_setting_D} for Settings~(D.1)--(D.2), and Table~\ref{tab: simulation_setting_D_mock} for Settings~(D.3)--(D.4). In Table~\ref{tab: simulation_setting_D}, the performance of our estimator remains satisfactory under $t_6$ noise, comparing with the results under Gaussian noise. It is also clear that our proposed estimator has higher correct proportion than the naive eigenvalue-ratio estimator (using $\xi_{m,k}=0$) in most setting. Exceptions include estimating $r_{1,1}$ for Setting~(D.1) with weak factors, where our estimator suffers more. This can be understood since the number of factors is more underestimated when factors are weak, while the naive estimator tends to overestimate the factor numbers and hence remedy this when cross-sectional dimensions are small. Notwithstanding, the naive estimator significantly suffers and always overestimates $r_{1,1}$ in Setting~(D.2). In comparison, our proposed estimator has improved accuracy when dimensionality increases, numerically reflecting the consistency result in Theorem~\ref{thm: eigenratio}.

Consider Settings~(D.3)--(D.4), and we first justify the use of $M=3$. With no global component in the first thread, at least $M=3$ is needed for the global components in other threads to be identified, as pointed out in Remark~\ref{remark: absent_component}. For Setting~(D.3), Table~\ref{tab: simulation_setting_D_mock} shows that our proposed estimator always outperforms the estimator using $\xi_{m,k}=0$. Although the latter has generally better performance in Setting~(D.4) by the same reasoning in Table~\ref{tab: simulation_setting_D_mock}, our estimator remains advantageous in estimating $r_{1,1}$ and $r_{1,2}$, namely, to tell if global components exist.

\subsection{Real data analysis}
\label{subsec: real_data}

We use our proposed multi-order tensor factor model to analyze New York traffic data, by studying a JXC constructed from a dataset containing individual taxi rides operated by Yellow Taxi within Manhattan Island of New York City. 
To preclude the influence of Covid-19, we focus on trip records during the pre-Covid period from January 1, 2013 to December 31, 2019. For an overview of the dataset, it contains information on the pick-up and drop-off time and locations which are coded according to 69 predefined taxi zones on Manhattan Island. In particular, each day is divided into 24 hourly periods to represent the pick-up and drop-off time, with the first hourly period from 0 \text{a.m.} to 1 a.m, so that on day $t$ we have $\cY_t\in \R^{69\times 69\times 24}$, where $y_{i_1,i_2,i_3,t}$ is the number of trips from zone $i_1$ to zone $i_2$, with pick-up time within the $i_3$th hourly period.

\paragraph{Data construction.}
{We consider the business-day series which is 1,763 days long. 
Each entry of the tensor observations is centered by subtracting its sample mean across the series, where the subtraction is performed entry-wise.
We do not apply additional adjustment such as de-trending. Low-frequency common variation due to, e.g.\ seasonal periodicity is expected to be captured by the estimated factors.}
The taxi data examples in both \cite{Chenetal2022} and \cite{CenLam2025} suggest that Times Square zone is one of the center zones leading the traffic dynamic in its local areas. To investigate this in detail, for each $t\in[1763]$, we form two threads: 
(i) $\cX_t^{(1)} \in\R^{24}$ denoting the hourly trip data within Times Square zone; 
and (ii) $\cX_t^{(2)} \in \R^{18\times 18\times 24}$ as a sub-tensor of $\cY_t$ by selecting 18 neighboring zones around Times Square. The map of selected taxi zones on Manhattan Island is shown in the first plot of Figure~\ref{Fig: taxi_biz_pickup_loading}.
The number of factors is estimated according to Section~\ref{subsec: est_num_fac}, where we adopt the iTIPUP estimator by \cite{Chenetal2022} to obtain $\wh{s}_{1,1} = 2$, $\wh{s}_{2,1} = 2$, $\wh{s}_{2,2} = 3$, and $\wh{s}_{2,3} = 2$, with the resulted global and local factor numbers 
$(\wh{r}_{1,1}, \wh{u}_{1,1}) =(1,1)$,
$(\wh{r}_{2,1}, \wh{u}_{2,1}) =(1,1)$,
$(\wh{r}_{2,2}, \wh{u}_{2,2}) =(1,2)$,
and $(\wh{r}_{2,3}, \wh{u}_{2,3}) =(1,1)$.

\begin{table}[t!]
\caption{Estimated global and local factor loading matrices on the hour factor. Values are scaled by 30 and in bold if their magnitudes are larger than 10.}
\label{tab: taxi_biz_hour_loading}
\footnotesize
\setlength{\tabcolsep}{3.6pt}
\begin{center}
\begin{tabular}{l||ccccccccccccccccccccccccccc}
\hline $0_\text{am}$ & & 2 & & 4 & & 6 & & 8 & & 10 & & $12_\text{pm}$ & & 2 & & 4 & & 6 & & 8 & & 10 & & $12_\text{am}$ \\
\hline $\wh{\bA}_{1,1}$ & 5 & 3 & 2 & 2 & 1 & 1 & 2 & 5 & 6 & 6 & 6 & 6 & 7 & 7 & 7 & 6 & 6 & 6 & \textbf{10} & \textbf{12} & 7 & 6 & 6 & 6 \\
$\wh{\bB}_{1,1}$ & -6 & -5 & -3 & -2 & -1 & -1 & 3 & 5 & 2 & 0 & -1 & 1 & 3 & 3 & 1 & 1 & 0 & 5 & \textbf{11} & 3 & \textbf{-25} & -2 & 1 & -2 \\
\\
$\wh{\bA}_{2,3}$ & 3 & 2 & 1 & 1 & 1 & 1 & 3 & 6 & 8 & 9 & 8 & 8 & 8 & 8 & 8 & 7 & 6 & 6 & 7 & 8 & 7 & 6 & 6 & 5 \\
$\wh{\bB}_{2,3}$ & 2 & 1 & 1 & 1 & 0 & 1 & 2 & 5 & 8 & 8 & 8 & 4 & 8 & 8 & 8 & 9 & 8 & 8 & 9 & 9 & 7 & 5 & 4 & 3 \\
\hline
\end{tabular}
\end{center}
\end{table}

\paragraph{Estimation and discussion.}
For the first thread, we use Table~\ref{tab: taxi_biz_hour_loading} to illustrate the estimated loading matrices on the hour factor. One key feature is the sheer difference in the global and local loading estimators for $\{\cX_t^{(1)}\}$: although they resonate at the 7--8 \text{p.m.} period, the local loading estimator contributes to the taxi traffic in Times Square zone very mildly or even passively. This suggests that the flow in Times Square can be largely explained by traffic in and out of Times Square, which is reasonable since Times Square can be seen as a core for tourism or offices, \text{cf.} Table~4 in \cite{Chenetal2022}. Such a centric position of Times Square is also reflected by the peaking magnitudes in the table.

\begin{figure}
\begin{center}
\includegraphics[width=0.88\columnwidth]{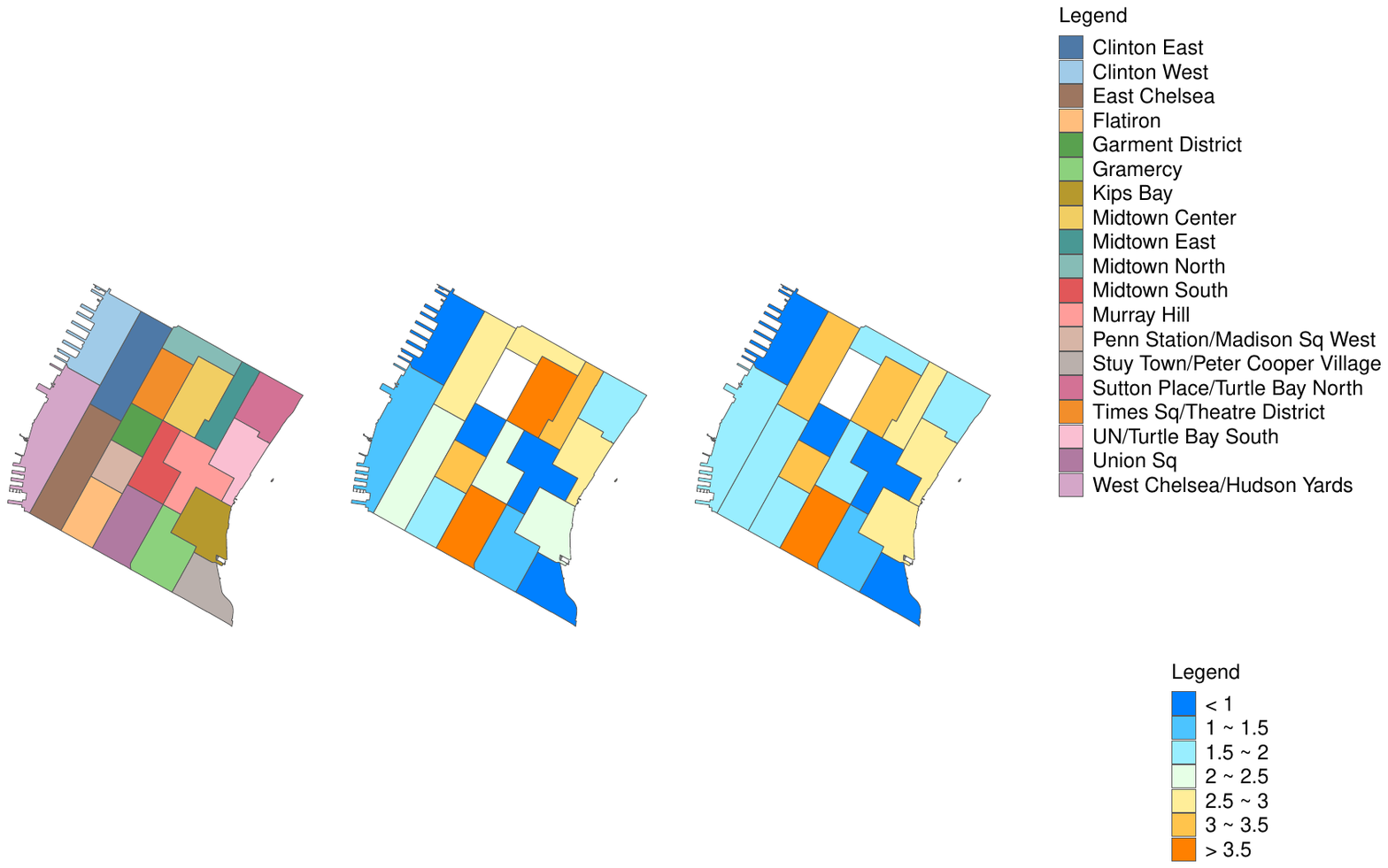}
\caption{Map of zones (left), together with the estimated global (middle) and local (right) loadings on the pickup factor. Loading values are scaled by 10, and Times Square is in white.}
\label{Fig: taxi_biz_pickup_loading}
\end{center}
\begin{center}
\centerline{\includegraphics[width=0.8\columnwidth]{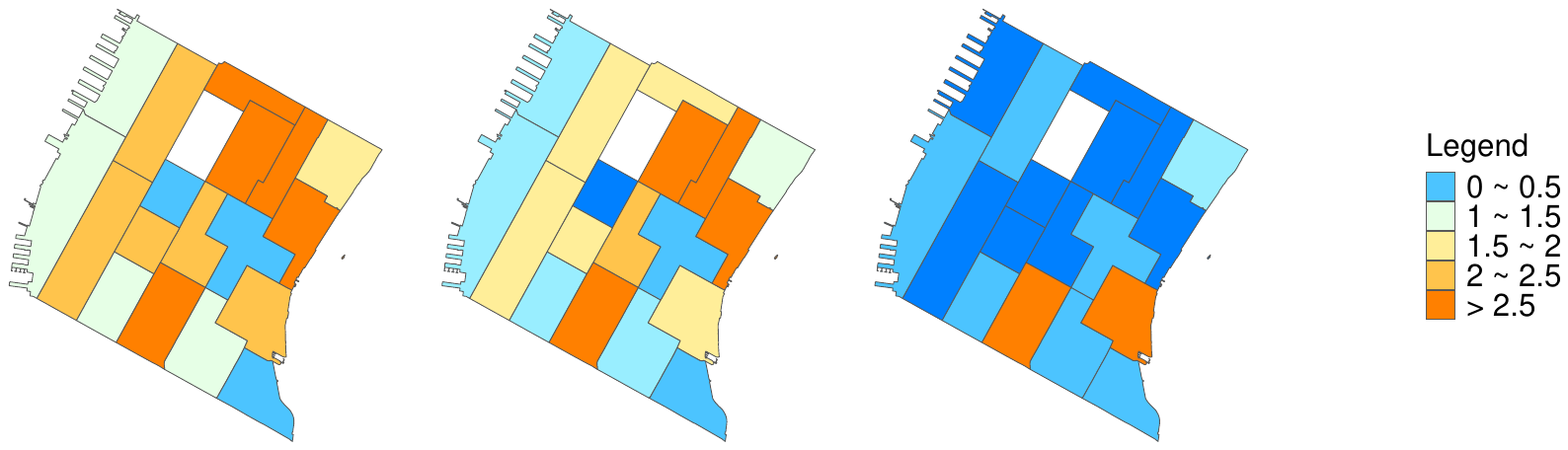}}
\caption{Estimated global (left) and local (middle and right) loadings on the drop-off factor. Loading values are scaled by 10, and Times Square is in white.}
\label{Fig: taxi_biz_dropoff_loading}
\end{center}
\end{figure}

Next, consider the second thread which represents the local areas around Times Square zone. Similarly, we plot the estimated loading matrices on the pickup and drop-off factors in Figures~\ref{Fig: taxi_biz_pickup_loading} and \ref{Fig: taxi_biz_dropoff_loading}, respectively. The first feature is the different patterns between the pickup and drop-off factors. In detail, the pickup factors are generally stronger and centering at the Midtown areas and Union Square, whereas the drop-off factors have more intricate patterns, which is consistent with more drop-off factors, in our selected zones, found in Table~4 in \cite{Chenetal2022}.

Another finding lies at the right panel of Figure~\ref{Fig: taxi_biz_dropoff_loading}, denoting the second local loading matrix estimated on the drop-off factor. Besides pinpointing again the importance of Union Square which is a transportation hub with many shops and restaurants, the active Kips Bay zone can be clearly identified, indicating the busy traffic within. 
{The Kips Bay zone is almost silent in the estimated global pickup loading,
which indicates that Kips Bay's traffic has a weaker contemporaneous co-movement with the overall Times Square factor structure.
It is also noteworthy that the Kips Bay zone has been omitted in both \cite{Chenetal2022} and \cite{CenLam2025},
whereas our analysis based on the proposed model identifies it through the local drop-off factor.}

\paragraph{Model assessment.}
We assess if the proposed model extracts cross-thread dependence beyond what tensor factor models capture.
For any fitted model on each thread $m\in \{1,2\}$, let $\{\wh\cE_t^{(m)}\}$ denote the residual,
and define the sample cross-covariance matrix between the residuals of threads~1 and 2 by
$\wh\bSigma^{\textsf{res}} := T^{-1} \sum_{t=1}^T \vec(\wh\cE_t^{(1)}) \vec(\wh\cE_t^{(2)})^\trans$.
Then define two measures of cross-thread residual dependence as
\begin{align*}
    \psi_1 &:= 
    \frac{ \big\| \wh\bSigma^{\textsf{res}} \big\|_F^2 }{ \big( T^{-1} \sum_{t=1}^T \|\vec(\wh\cE_t^{(1)})\|^2 \big)
    \big( T^{-1} \sum_{t=1}^T \|\vec(\wh\cE_t^{(2)})\|^2 \big)} , \quad\quad
    \psi_2 :=
    \frac{1}{p_1 p_2} \sum_{i=1}^{p_1} \sum_{j=1}^{p_2} \wh\rho_{ij}^2 ,
\end{align*}
where $\wh\rho_{ij}$ denotes the sample correlation between the $i$th residual coordinate of thread~1 and the $j$th residual coordinate of thread~2.
The measure $\psi_1$ reflects aggregate linear association relative to the total residual scale, 
while $\psi_2$ averages squared pairwise residual correlations and is less influenced by the overall size of individual residuals.

We compare the proposed model under several choices.
The main specification, denoted as ``Our (TS)'', takes Times Square as thread~1 and its 18 neighboring zones as thread~2.
The specifications ``Our (KB)'' and ``Our (MC)'' take Kips Bay and Midtown Center, respectively, as thread~1,
with the remaining 18 zones as thread~2.
Based on the main specification, we further consider:
``TFM (TS)'' respectively models threads~1 and 2 as approximate factor models and a tensor factor model using PCA-type estimators \citep{BaiNg2002, ZhangCenCho2026};
``TFM'' models all 19 zones in one tensor factor model and then separates the estimated residuals into two threads (under the main specification);
analogously for ``CP (TS)'' and ``CP'' except that a CP-decomposition factor model is assumed and the double-projection estimator is applied on any tensor series \citep{Changetal2026}.

\begin{table}[t!]
\centering
\caption{Measures of cross-thread residual dependence for different models.
Values are multiplied by $10^3$, and the minimum in each row is in bold.}
\label{tab: validation}
\setlength{\tabcolsep}{6pt}
\renewcommand{\arraystretch}{1.2}
\begin{tabular}{@{}l|ccccccccccc@{}}
\hline
 & Our (TS) & Our (KB) & Our (MC) & TFM (TS) & TFM & CP (TS) & CP \\
\hline
$\psi_1$ &  \textbf{2.93} & 9.80 & 22.22 & 9.39 & 35.52 & 6.42 & 197.68 \\
$\psi_2$ & \textbf{5.14} & 17.36 & 28.03 & 7.68 & 14.71 & 7.99 & 13.71 \\
\hline
\end{tabular}
\end{table}

Table~\ref{tab: validation} reports the results. The proposed model with Times Square as thread~1 attains the smallest value of $\psi_1$ and $\psi_2$ among all specifications, which evidently suggests that the estimated global component absorbs a considerable part of the contemporaneous cross-thread association that remains in the residuals of the separate factor models.
Such improvement is particularly clear relative to the single tensor factor models fitted to all 19 zones, i.e., TFM and CP,
indicating that modeling all zones without a global-local-type decomposition leaves strong cross-zone dependence in the residuals.
Also, note that the choice of zone in thread~1 is crucial, in that all measures in both Our (KB) and Our (MC) are larger than TFM (TS),
which is actually consistent with the interpretation that the proposed factor model is the most informative there is a clear central role in the surrounding traffic network.

Lastly, these comparisons are descriptive rather than formal model-selection tests,
and this analysis on the taxi data is for demonstration of our method, and more sophisticated investigations are required to further understand the traffic pattern.
For instance, checking the stability of the estimated loadings across sub-periods requires a formal change-point analysis, which is beyond the scope of this paper and left for future study.

\section{Conclusion}\label{sec: conclusion}
This paper proposes a multi-order tensor factor model for the joint analysis of multiple tensor-valued time series with heterogeneous tensor orders and dimensions. Unlike existing multilevel or group factor models, our proposed framework does not require a common latent factor shared across all threads. Instead, dependence is captured through correlated global factors, allowing substantially greater flexibility. We established the uniqueness of the decomposition into global and local components and developed estimation procedures for all model parameters. Our theoretical analysis permits fairly general dependence structures and gives convergence rates that are competitive with, and in several cases improve upon, those available in the existing literature. The effectiveness of our method is corroborated by extensive simulations and a real data application.
We hope the proposed framework provides a useful foundation for studying collections of heterogeneous tensor time series.
One plausible direction is to adapt the global factor structure herein from Tucker decomposition to CP decomposition. This allows for exactly the same global factors even when time series have different tensor orders, which we leave for future research.

\section*{Acknowledgments}
The author thanks the editor, associate editor, and two anonymous referees for their constructive comments, which led to a significant improvement in this work.
The author gratefully acknowledges the support from the Engineering and Physical Sciences Research Council (\href{https://gtr.ukri.org/projects?ref=EP\%2FZ531327\%2F1}{EP/Z531327/1}), and the support from the Department of Statistics, London School of Economics.

\section*{Conflicts of Interest}
The authors declare no conflicts of interest.

\section*{Data Availability Statement}
The New York traffic dataset analyzed in Section~\ref{subsec: real_data} is available at:
{\url{https://www1.nyc.gov/site/tlc/about/tlc-trip-record-data.page}}.


\begin{appendix}

\section{Introduction of tensor map operator}
\label{sec: intro_map}

We extend the tensor reshape operator in \cite{CenLam2025_KronProd} to a tensor map operator. Recall first that given any order-$K$ tensor $\cX\in \R^{I_1\times \dots \times I_K}$ and a set of ordered, strictly increasing elements $v=\{a_1,\dots, a_\ell\} \subseteq[K]$, the reshaped tensor $\reshape(\cX,v)$ is the order-$(K-\ell+1)$ tensor resulting from combining all modes indexed by $v$ into the last mode. Let $V=\{v_1,\dots, v_L\}$ be a set of ordered vectors where each vector $v_l$ is the parameter representing mode indices in a reshape operation and all $v_l$'s, $l\in[L]$, form a partition of $[K]$. Then we denote the mapped tensor $\map(\cX, V)$ as the order-$L$ tensor such that the $l$th mode has dimension $\prod_{j\in v_l} I_j$ and is constructed by combining entries of $\cX$ with modes in $v_l$. This $\map(\cdot,\cdot)$ operation is more general than permuting the array elements in $\cX$ as it also allows for stacking elements from different modes into one, and more importantly, preserves the factor structure in the data, by iteratively applying Theorem~1 in \cite{CenLam2025_KronProd}.
We refer to $V$ as \textit{channel}. Without loss of generality, let the data tensor $\cX$ be pre-processed such that, for $a<b$, elements in $v_a$ are larger than those in $v_b$. Then the procedure of mapping a tensor can be summarized in Algorithm~\ref{alg: map}. Finally, by \cite{CenLam2025_KronProd}, the reshape operator can be inverted, meaning that one can recover $\cX$ from $\reshape(\cX, v)$ given the original dimensions of $\cX$, we may also recover $\cX$ from $\map(\cX, V)$.

\begin{algorithm}
\caption{Tensor map}\label{alg: map}
\begin{algorithmic}[1]
\State \textbf{Input:} Tensor $\cX$, channel $V=\{v_1,\dots, v_L\}$
\State \textbf{Initialize:} Set $\map(\cX, V) \gets \cX$
\For{$l \in [L]$}
\State Set $\map(\cX, V) \gets \reshape\big\{ \map(\cX, V), v_l \big\}$
\EndFor
\State \textbf{Output:} The mapped tensor $\map(\cX, V)$
\end{algorithmic}
\end{algorithm}

\end{appendix}

\bibliographystyle{apalike}
\bibliography{reference}

\clearpage


\setcounter{section}{0}
\setcounter{equation}{0}
\def\theequation{S\arabic{equation}}
\def\thesection{S\arabic{section}}
\def\thetheorem{S\arabic{theorem}}
\def\thelemma{S\arabic{lemma}}
\def\thefigure{S\arabic{figure}}
\def\thedefinition{S\arabic{definition}}
\def\theexample{S\arabic{example}}
\def\theremark{S\arabic{remark}}
\def\thetable{S\arabic{table}}
\def\thecorollary{S\arabic{corollary}}

\newtheorem{lem}{Lemma S. \hspace{-8pt}}
\renewcommand{\theequation}{S.\arabic{equation}}

\section{Additional numerical results}\label{sec: additional_num}

In Section~\ref{subsec: simulation}, the global factor for each thread $m$ is generated by $\cG_t^{(m)} = \map(\cG_t, V_m)$, which is a deterministic reorganization of a single underlying tensor $\cG_t$.
While it seems to be a particularly favorable design, we show that our method is insensitive to this by studying heterogeneous correlated global factors below.

Consider the configuration of Setting~(A.1) in Section~\ref{subsec: simulation},
where $M=3$, $(K_1,K_2,K_3) = (1,1,2)$, and $(p_{1,k}, p_{2,h}, p_{3,l}) =(30, 30, 10)$ (i.e., $(p_{1,1}, p_{2,1}, p_{3,1}, p_{3,2}) = (30, 30, 10,10)$),
with strong factors, and Gaussian noise. The local factors, local loading matrices, and idiosyncratic noise are generated exactly therein.
We next detail the specification of the global factors, which we vary across four designs.

To this end, for every $t$ and $m$,
let $h_{1,t}$, $h_{2,t}$, and the elements in $\eta_t^{(m)}$
be independent AR(1) processes with AR coefficient $0.5$ and standard Gaussian innovations.
We consider four designs as follows:
\begin{itemize}
    \item \textbf{Map}. 
    For all threads, set $\cG_t^{(1)} = \cG_t^{(2)} = \cG_t^{(3)} = h_{1,t}$.
    \item \textbf{Full}. 
    Set $\cG_t^{(m)} = h_{1,t} + \eta_t^{(m)}$ for all threads $m$.
    \item \textbf{Part-I}.
    Global factors are partially connected with groups $\{1,2\}$ and $\{1,3\}$, i.e.,
    \[
    \cG_t^{(1)} = h_{1,t} + h_{2,t} + \eta_t^{(1)}
    , \quad
    \cG_t^{(2)} = h_{1,t} + \eta_t^{(2)}
    , \quad
    \cG_t^{(3)} = h_{2,t} + \eta_t^{(3)} .
    \]
    \item \textbf{Part-II}. 
    Global factors are partially connected with groups $\{1,3\}$ and $\{2,3\}$, i.e.,
    \[
    \cG_t^{(1)} = h_{1,t} + \eta_t^{(1)}
    , \quad
    \cG_t^{(2)} = h_{2,t} + \eta_t^{(2)}
    , \quad
    \cG_t^{(3)} = h_{1,t} + h_{2,t} + \eta_t^{(3)} .
    \]
\end{itemize}

In \textbf{Map}, the global factor is a common scalar process for the two vector threads and the matrix thread,
which is the same setting as in Section~\ref{subsec: simulation}.
\textbf{Full} represents a fully connected scenario, generating heterogeneous global factors that are correlated but not deterministic transformations of one another.
Partial connectivity among threads is also considered in the last two designs
Specifically, in \textbf{Part-I}, thread~1 is correlated with both other threads, while threads~2 and~3 are not directly correlated;
whereas in \textbf{Part-II}, thread~3 is the common neighbor of threads~1 and 2.

We also assess the sensitivity of our fast implementation in Section~\ref{subsec: bootstrap_implement} to the choice of sub-sampling size.
In particular, 
we set $\sum_{n\in[M] \setminus\{m\}}|\cS_n^\dagger| = \lfloor  b_0 \cdot (\sum_{n\in[M] \setminus\{m\}} p_n) \rfloor $,
where $b_0$ varies from $\{0.4, 0.6, 0.8, 1\}$.
Note that $b_0=1$ corresponds to the full estimator without sub-sampling.

\begin{sidewaystable}[htp!]
\centering
\caption{Column space distances and computation time under alternative global factor designs with $T=100$. 
Each cell shows the mean of measure over 100 runs for the corresponding setting.}
\label{tab: sim_global_other_design_T100}
\setlength{\tabcolsep}{4pt}
\renewcommand{\arraystretch}{1.3}
\begin{tabular}{@{}c|cccc|cccc|cccc|cccccccccccccccccccccccccccc@{}}
\hline
Design & \multicolumn{4}{c|}{\textbf{Map}} & \multicolumn{4}{c|}{\textbf{Full}} & \multicolumn{4}{c|}{\textbf{Part-I}} & \multicolumn{4}{c}{\textbf{Part-II}} \\
\hline
$b_0$ & 1.0 & 0.8 & 0.6 & 0.4 & 1.0 & 0.8 & 0.6 & 0.4 & 1.0 & 0.8 & 0.6 & 0.4 & 1.0 & 0.8 & 0.6 & 0.4 \\
\hline
$\cD(\bA_{1,1}, \wh{\bA}_{1,1}^S)$
& 0.188 & 0.190 & 0.192 & 0.195
& 0.241 & 0.244 & 0.242 & 0.242
& 0.199 & 0.200 & 0.197 & 0.201
& 0.309 & 0.308 & 0.305 & 0.306 \\
$\cD(\bA_{2,1}, \wh{\bA}_{2,1}^S)$
& 0.175 & 0.176 & 0.178 & 0.179
& 0.221 & 0.222 & 0.228 & 0.225
& 0.317 & 0.324 & 0.335 & 0.315
& 0.290 & 0.292 & 0.293 & 0.296 \\
$\cD(\bA_{3,1}, \wh{\bA}_{3,1}^S)$
& 0.045 & 0.045 & 0.047 & 0.047
& 0.065 & 0.069 & 0.072 & 0.072
& 0.117 & 0.120 & 0.130 & 0.154
& 0.048 & 0.049 & 0.050 & 0.052 \\
$\cD(\bA_{3,2}, \wh{\bA}_{3,2}^S)$
& 0.045 & 0.045 & 0.046 & 0.048
& 0.067 & 0.070 & 0.070 & 0.064
& 0.115 & 0.118 & 0.140 & 0.160
& 0.048 & 0.049 & 0.051 & 0.053 \\
\hline
$\cD(\bB_{1,1}, \wh{\bB}_{1,1})$
& 0.186 & 0.188 & 0.189 & 0.190
& 0.218 & 0.221 & 0.218 & 0.215
& 0.195 & 0.196 & 0.193 & 0.194
& 0.270 & 0.268 & 0.263 & 0.261 \\
$\cD(\bB_{2,1}, \wh{\bB}_{2,1})$
& 0.175 & 0.175 & 0.177 & 0.178
& 0.203 & 0.203 & 0.209 & 0.205
& 0.287 & 0.294 & 0.303 & 0.285
& 0.251 & 0.254 & 0.254 & 0.255 \\
$\cD(\bB_{3,1}, \wh{\bB}_{3,1})$
& 0.037 & 0.037 & 0.037 & 0.037
& 0.050 & 0.050 & 0.050 & 0.040
& 0.072 & 0.078 & 0.095 & 0.110
& 0.039 & 0.039 & 0.039 & 0.040 \\
$\cD(\bB_{3,2}, \wh{\bB}_{3,2})$
& 0.038 & 0.038 & 0.038 & 0.038
& 0.049 & 0.056 & 0.058 & 0.058
& 0.072 & 0.078 & 0.088 & 0.100
& 0.039 & 0.039 & 0.039 & 0.039 \\
\hline
Run Time (s)
& 31.40 & 25.13 & 18.62 & 12.94
& 31.58 & 24.72 & 18.68 & 12.96
& 31.50 & 25.15 & 19.15 & 12.69
& 31.43 & 24.70 & 18.71 & 12.83 \\
\hline
\end{tabular}
\end{sidewaystable}

\begin{sidewaystable}[htp!]
\centering
\caption{Column space distances and computation time under alternative global factor designs with $T=400$. 
Each cell shows the mean of measure over 100 runs for the corresponding setting.}
\label{tab: sim_global_other_design_T400}
\setlength{\tabcolsep}{4pt}
\renewcommand{\arraystretch}{1.3}
\begin{tabular}{@{}c|cccc|cccc|cccc|cccc@{}}
\hline
Design & \multicolumn{4}{c|}{\textbf{Map}} & \multicolumn{4}{c|}{\textbf{Full}} & \multicolumn{4}{c|}{\textbf{Part-I}} & \multicolumn{4}{c}{\textbf{Part-II}} \\
\hline
$b_0$ & 1.0 & 0.8 & 0.6 & 0.4 & 1.0 & 0.8 & 0.6 & 0.4 & 1.0 & 0.8 & 0.6 & 0.4 & 1.0 & 0.8 & 0.6 & 0.4 \\
\hline
$\cD(\bA_{1,1}, \wh{\bA}_{1,1}^S)$
& 0.088 & 0.088 & 0.088 & 0.088
& 0.112 & 0.112 & 0.113 & 0.115
& 0.095 & 0.095 & 0.095 & 0.097
& 0.142 & 0.142 & 0.142 & 0.143 \\
$\cD(\bA_{2,1}, \wh{\bA}_{2,1}^S)$
& 0.092 & 0.092 & 0.092 & 0.091
& 0.104 & 0.104 & 0.105 & 0.105
& 0.166 & 0.169 & 0.169 & 0.179
& 0.153 & 0.153 & 0.154 & 0.154 \\
$\cD(\bA_{3,1}, \wh{\bA}_{3,1}^S)$
& 0.023 & 0.023 & 0.023 & 0.023
& 0.026 & 0.026 & 0.027 & 0.027
& 0.036 & 0.036 & 0.036 & 0.037
& 0.021 & 0.021 & 0.021 & 0.022 \\
$\cD(\bA_{3,2}, \wh{\bA}_{3,2}^S)$
& 0.022 & 0.022 & 0.022 & 0.023
& 0.026 & 0.026 & 0.027 & 0.028
& 0.036 & 0.036 & 0.037 & 0.037
& 0.021 & 0.021 & 0.022 & 0.022 \\
\hline
$\cD(\bB_{1,1}, \wh{\bB}_{1,1})$
& 0.087 & 0.087 & 0.087 & 0.087
& 0.098 & 0.099 & 0.100 & 0.101
& 0.090 & 0.090 & 0.090 & 0.091
& 0.110 & 0.110 & 0.110 & 0.111 \\
$\cD(\bB_{2,1}, \wh{\bB}_{2,1})$
& 0.092 & 0.092 & 0.092 & 0.091
& 0.092 & 0.092 & 0.092 & 0.091
& 0.139 & 0.140 & 0.140 & 0.149
& 0.123 & 0.123 & 0.124 & 0.124 \\
$\cD(\bB_{3,1}, \wh{\bB}_{3,1})$
& 0.020 & 0.020 & 0.020 & 0.020
& 0.020 & 0.020 & 0.020 & 0.020
& 0.019 & 0.019 & 0.019 & 0.019
& 0.019 & 0.019 & 0.019 & 0.019 \\
$\cD(\bB_{3,2}, \wh{\bB}_{3,2})$
& 0.020 & 0.020 & 0.020 & 0.020
& 0.020 & 0.020 & 0.020 & 0.020
& 0.020 & 0.020 & 0.020 & 0.020
& 0.020 & 0.020 & 0.020 & 0.020 \\
\hline
Run Time (s)
& 88.00 & 68.62 & 51.54 & 35.06
& 88.12 & 69.22 & 51.84 & 34.75
& 87.33 & 69.37 & 51.47 & 34.84
& 86.08 & 69.45 & 51.99 & 35.12 \\
\hline
\end{tabular}
\end{sidewaystable}

Tables~\ref{tab: sim_global_other_design_T100}  and \ref{tab: sim_global_other_design_T400} present the results for $T=100$ and $T=400$, respectively, over 100 Monte Carlo replications. 
First, the computation time decreases almost linearly as the sub-sampling fraction $b_0$ decreases, 
while the column space distances remain stable in all settings,
which numerically supports our fast implementation.
Secondly, \textbf{Full} leads to moderately larger global loading-space distances than the \textbf{Map}, which is expected as the thread-specific noise $\eta_t^{(m)}$ weakens the contemporaneous cross-thread correlation compared to the deterministic transformation.

Also notably, partial connectivity alters the estimation quality in a structured way. 
In \textbf{Part-I}, thread~1 is a common neighbor (of threads~2 and 3) and its global loading space is estimated more accurately than those in the other threads. 
In \textbf{Part-II}, the common-neighbor role is played by thread~3, whose global loading space remains well estimated, while the other two threads have larger errors. 
These are consistent with the interpretation that our global loading estimator pools information from all correlated threads, 
and that Assumption~\ref{ass: core_factor}(c) is a sufficient but not necessary condition for the global loading space to be estimable, as pointed out in the main text.

\clearpage

\section{Additional details and discussion}\label{sec: additional_detail}

\subsection{Practical guidance on tuning and computation}
\label{subsec: practical_tuning}

In what follows, we provide a practical guidance for the tuning parameters introduced in the estimation procedures.
Also, refer to the end of Section~\ref{subsec: two_stage_cca} for a discussion on the computational complexity of the proposed estimators.

\paragraph{Total factor numbers $s_{m,k} = r_{m,k}+u_{m,k}$.}
These can be taken from a consistent external estimator of the joint factor numbers, for instance the iTIPUP estimator of \cite{Chenetal2022}. Any other consistent estimator can be substituted, see also \cite{Lam2021, Hanetal2022}.

\paragraph{Perturbation constant $\xi_{m,k}$.}
The estimator in \eqref{eqn: est_num_fac_define} requires $\xi_{m,k}$, which we set to
\[
\xi_{m,k} =
\frac{p_m}{5\sqrt{T}}
\Bigg( \sum_{\substack{n=1\\ n\neq m}}^M p_n \Bigg).
\]
The choice of $5$ in the denominator is a mild stabilization constant. Its role is to shrink the ratio
\[
\frac{\lambda_{i+1}(\wh\bSigma_{m,k}) + \xi_{m,k}}{\lambda_i(\wh\bSigma_{m,k}) + \xi_{m,k}} 
\]
towards unity away from the spikes, without diluting the leading eigenvectors. 
Our simulation study with $\xi_{m,k}=0$ shows that such stabilization is important.
Without it, the ratio estimator tends to overestimate the true number of global factors. Any other choice of the constant in general may be used by practitioners.

\paragraph{Mock eigenvalue constant $c$.}
The mock eigenvalue in \eqref{eqn: mock_evalue} involves a constant $c$, which we set to $c=5$,
and any $c>0$ of moderate size gives the same asymptotic guarantee.

\paragraph{Upper bound $r_{\max,m,k}$.}
We take $r_{\max,m,k}=\lfloor p_{m,k} /2 \rfloor$. The procedure is insensitive to this choice as long as $r_{\max,m,k} \geq r_{m,k}$ and $r_{\max,m,k}$ does not grow too fast.

\subsection{Discussion on the core factor conditions}\label{subsec: discussion_core_fac_part_c}

We provide additional discussions on Assumption~\ref{ass: core_factor}(c) when the loading matrices are only asymptotically half-orthogonal\footnote{Note that our exact half-orthogonality setup in  Assumption~\ref{ass: loadings}(b) is by no means restrictive as discussed in the main text, since we may always take the QR decomposition, reformulate the loading matrix as e.g.\ $\bA_{m,k} = \bQ_{m,k} \bR_{m,k}$ with $\bQ_{m,k} \in \R^{p_{m,k} \times r_{m,k}}$ half-orthogonal, and absorb $\bR_{m,k}$ to the factors. The discussion here is to showcase the possibility of Assumption~\ref{ass: core_factor}(c) when loadings are identified as asymptotically half-orthogonal and the factors are free from any transformations such as $\bR_{m,k}$.}, i.e., if Assumption~\ref{ass: loadings}(b) is replaced by:
\begin{assumptionp}{\ref*{ass: loadings}(b)'}
\label{ass: loadings_part_b_asymp}
$\big\| p_{m,k}^{-1} (\bA_{m,k}, \bB_{m,k})^\trans (\bA_{m,k}, \bB_{m,k}) -\bI_{r_{m,k}+u_{m,k}} \big\| \leq c p_{m,k}^{-1/2}$ for large $p_{m,k}$.
\end{assumptionp}
To this end, we present the following version of Assumption~\ref{ass: core_factor}(c), which is the appropriate formulation when the loading matrices are asymptotically half-orthogonal as above.

\begin{assumptionp}{\ref*{ass: core_factor}(c)'}[Assumption~\ref{ass: core_factor}(c) under asymptotically half-orthogonal loading matrices]
\label{ass: core_factor_part_c_general}
For any $m\in [M]$, $n\neq m$,
the process $\big\{ \vec(\cG_t^{(m)}) \otimes \vec(\cG_t^{(n)})^\trans \big\}$ has constant mean, and
for any $k\in[K_m]$, $j\in[r_{m,k}]$, we have
\[
\sigma_{j}\bigg( \E\Big[ \mat_k(\cG_t^{(m)}) \otimes \Big\{ \vec(\cG_t^{(n)})^\trans \big( \otimes_{h=1}^{K_n} \bA_{n,h}
/\sqrt{p_{n,h}} \big)^\trans \Big\} \Big] \bigg) \asymp 1 .
\]
\end{assumptionp}

Assumption~\ref{ass: core_factor_part_c_general} makes the role of the loading matrices explicit.
Under the exact half-orthogonality imposed in Assumption~\ref{ass: loadings}(b),
we have $\wt\bA_n^\trans \wt\bA_n = \bI_{r_n}$ where $\wt\bA_n := \otimes_{h=1}^{K_n}(\bA_{n,h}/ \sqrt{p_{n,h}})$.
Hence the loading matrices cancel exactly in the population cross-thread covariance matrix, and Assumption~\ref{ass: core_factor_part_c_general} reduces to Assumption~\ref{ass: core_factor}(c). Lemma~\ref{lemma: eigenvalue_Sigma_aux} formalizes this equivalence.

The following claim provides sufficient conditions for Assumption~\ref{ass: core_factor_part_c_general} under the setting of asymptotically half-orthogonal loading matrices. It is the analogue of Claim~\ref{claim: assump_fac_simpler} in the main text. Its proof is deferred to Section~\ref{subsec: proof_claim}.

\begin{claim}[Claim~\ref{claim: assump_fac_simpler} under asymptotically half-orthogonal loading matrices]
\label{claim: assump_fac_simpler_asymp_ortho}
Let Assumption~\ref{ass: loadings_part_b_asymp} hold.
Under $K_1=\dots=K_M$ and $\cG_t^{(1)} = \dots =\cG_t^{(M)}$, if $\vec(\cG_t^{(M)})$ has uncorrelated elements with zero mean and variance $\otimes_{k\in[K_M]} \bSigma_{G,k}$ with bounded elements in each $\bSigma_{G,k}\in \R^{r_{M,k}\times r_{M,k}}$, then Assumption~\ref{ass: core_factor_part_c_general} is satisfied.
When each global factor
is weakly stationary and 
admits the form $\cG_t^{(m)} =\map(\cG_t, V_m)$ as in \eqref{eqn: model_map_Gt} (Section~\ref{subsec: simulation}), Assumption~\ref{ass: core_factor_part_c_general} is also satisfied.
\end{claim}

\subsection{An alternative two-stage estimation 
}
\label{subsec: two_stage_cca}

We elaborate on the two-stage procedure noted in Remark~\ref{remark: model_read} in the main text.
For simplicity, we let the factors be weakly stationary, and focus on the series of vectorized observations to show the main idea.

\paragraph{The two-stage procedure.}
For the first stage, we estimate the joint factor model.
For each $m\in[M]$, write
\[
\bh_t^{(m)} := \vec\big( \diag(\cG_t^{(m)},\cF_t^{(m)}) \big) 
\in \R^{s_m} ,
\quad
s_m := \prod_{k=1}^{K_m} (r_{m,k}+u_{m,k}) .
\]
Then by \eqref{eqn: Xt_rewrite_tfm} we have
\[
\vec(\cX_t^{(m)}) = \bL_m \bh_t^{(m)} + \vec(\cE_t^{(m)})
, \quad
\bL_m = \otimes_{k=1}^{K_m} (\bA_{m,k},\bB_{m,k}) .
\]
Suppose that a first-stage PCA-type estimator \citep[e.g.][]{BaiNg2002, BaiNg2023} of this joint factor model gives
\[
\wh\bh_t^{(m)} = \wh\bH_{h,m}^\trans \bh_t^{(m)} +  o_P(1) 
, \quad
\wh\bL_m = \bL_m \wh\bH_{h,m} + o_P(1) ,
\]
where $\wh\bH_{h,m} \in\mathbb R^{s_m\times s_m}$ is asymptotically orthogonal and $\wh{\bL}_m$ is the estimated joint loading matrix. The matrix $\wh\bH_{h,m}$ reflects the usual rotational indeterminacy of factor loading space.

In the second stage, we perform canonical correlation analysis to separate the global and local parts.
For any $n\neq m$, define the cross-covariance between the estimated joint factors as $\bC_{mn} := \E\big[ \bh_t^{(m)} (\bh_t^{(n)})^\trans \big]$.
Then it could be estimated by
\[
\wh\bC_{mn} := \frac{1}{T} \sum_{t=1}^T \wh\bh_t^{(m)} (\wh\bh_t^{(n)})^\trans ,
\]
which should be close to $\wh\bH_{h,m}^\trans \bC_{mn}
\wh\bH_{h,n}$ up to some rate of convergence such as $\cO_P(T^{-1/2})$.
Assumption~\ref{ass: identification}(b) implies that we may write (up to some alignment or re-ordering of elements in the joint factors):
\[
\bC_{mn} =
\begin{pmatrix}
\E\big[ \vec(\cG_t^{(m)}) \vec(\cG_t^{(n)})^\trans \big] & \zero \\
\zero & \zero
\end{pmatrix} .
\]
Hence all spiked left-singular vectors of $\bC_{mn}$ should lie in the space corresponding to the vectorized global factors $\vec(\cG_t^{(m)})$. 
Aggregating over all $n\neq m$ and under certain regularity condition similar to Assumption~\ref{ass: core_factor}(c), the leading left-singular vectors of $\wh\bC_{mn}$ should be close to the direction of the space corresponding to $\vec(\cG_t^{(m)})$. 
Then appropriate alignment on the columns of $\wh\bL_m$ should estimate $\col( \otimes_{k=1}^{K_m} \bA_{m,k})$.
Note that such alignment only identifies the vectorized global loading space. To recover each mode-$k$ global loading matrix $\bA_{m,k}$, we need to exploit the Kronecker product structure of $\bL_m$, at the cost of more complicated procedures.

As $s_m$ and $s_n$ are fixed, the optimal rate of $\wh\bC_{mn}$ is the parametric rate $\cO_P(T^{-1/2})$ without additional structure.
Such a rate would be carried to the final alignment, and the resulting global loading estimator would again satisfy the $\cO_P(T^{-1/2})$ rate established in Theorem~\ref{thm: global_loading}. This suggests that the two-stage procedure therefore does not improve the asymptotic rate.

\paragraph{Computational complexity.}
We next compare the computational costs of the two-stage procedure and the estimators presented in the main text.
For the PCA-type estimation in the first stage depicted previous, 
for each $m\in[M]$ and $k\in[K_m]$, the computational complexity to compute the mode-$k$ sample covariance matrix $T^{-1} \sum_{t=1}^T \mat_k(\cX_t^{(m)}) \mat_k(\cX_t^{(m)})^\trans$ is of order $T p_{m,k}^2p_{m,-k} = T p_m p_{m,k}$.
Over all modes, the total first-stage computational cost is
\begin{equation}
\label{eqn: comp_cost_two_stage}
\cO\Big( T p_m \sum_{k=1}^{K_m}
p_{m,k} \Big) .
\end{equation}
All subsequent steps (including the second-stage procedure) are of fixed rank and hence the computational complexity are of lower order.

For our method presented in the main text, for a given $m\in[M]$ and $k\in[K_m]$, the direct construction of $\wh\bSigma_{m,k}$ in \eqref{eqn: def_hat_Sigma_mk} requires, for each $n\neq m$ and each multi-index $\bi \in[p_{n,1}] \times \cdots \times[p_{n,K_n}]$, computing $\wh\bOmega_{k,n,\bi}^{(m)}$ which costs $\cO(Tp_m)$ for the time average, 
whereas computing $\wh\bOmega_{k,n,\bi}^{(m)} \wh\bOmega_{k,n,\bi}^{(m) \trans}$ costs $\cO( p_m p_{m,k})$. 
Thus, aggregating over all $\bi\in[p_n]$ and all $n\neq m$, the computational cost is
\begin{equation}
\label{eqn: comp_cost_main}
\cO\Big\{ p_m \sum_{n\neq m} p_n
\big( T + p_{m,k} \big) \Big\} ,
\end{equation}
which is usually (but not uniformly) more expensive than the two-stage procedure according to \eqref{eqn: comp_cost_two_stage}.
However, note that the first-stage procedure still requires consistent joint factor number estimation and joint loading estimation which very likely requires computing $\wh\bSigma_{m,k}$, and its second stage requires an additional mode-wise separation step if interpretation is desired. 
Moreover, the fast implementation in Section~\ref{subsec: bootstrap_implement} could help reduce the computational cost without extra steps, with $p_n$ in \eqref{eqn: comp_cost_main} replaced by $|\cS_n^\dagger|$. 
Hence we do not consider the two-stage procedure in this paper, and potential development is left to future research.

\subsection{New York traffic data: visualization}

Figure~\ref{Fig: manhattan} displays the full map of 69 taxi zones on Manhattan Island with the selected zones colored.

\begin{figure}[htp!]
\begin{center}
\centerline{\includegraphics[width=.9\columnwidth]{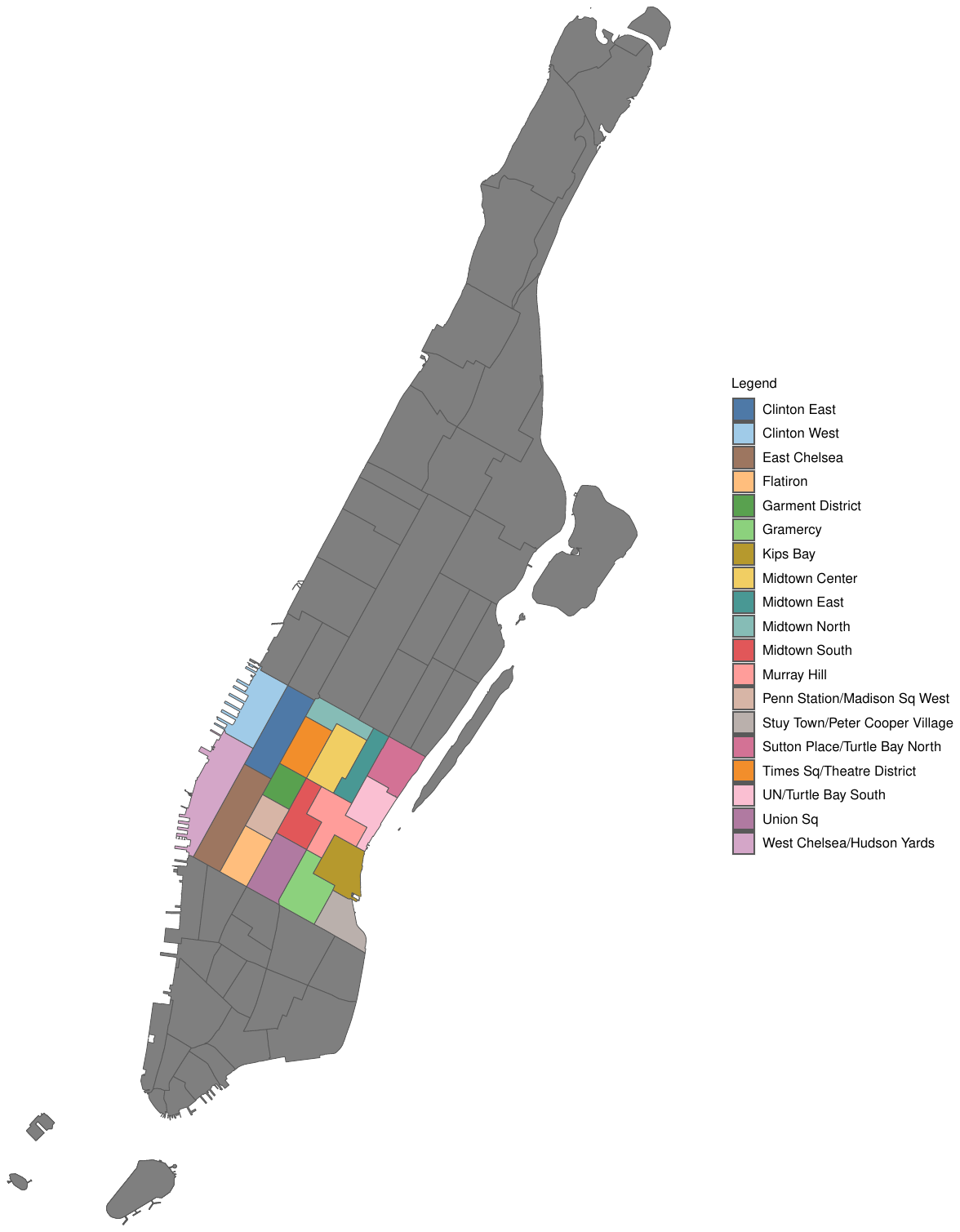}}
\caption{Taxi zones on Manhattan Island. Times Square zone is in orange, 18 neighboring zones are in different colors, and all remaining zones are shaded.}
\label{Fig: manhattan}
\end{center}
\end{figure}

\clearpage
\section{Proof of theorems and auxiliary results}

\subsection{Lemmas with their proofs}

\begin{lemma}\label{lemma: rate_Omega}
Let Assumptions~\ref{ass: identification}, \ref{ass: loadings}, \ref{ass: noise}, and \ref{ass: dependence_moment} hold. For any $m\in[M]$, $k\in[K_m]$, $n\neq m$, we have
\begin{align*}
    &\;\quad
    \sum_{\bi \in [p_{n,1}] \times \cdots \times[p_{n,K_n}]} \Big\| \wh\bOmega_{k,n,\bi}^{(m)} - \bOmega_{k,n,\bi}^{(m)} \Big\|_F^2 
    = \cO_P\Big( \frac{p_m p_n}{T} \Big) .
\end{align*}
\end{lemma}

\begin{proof}[Proof of Lemma~\ref{lemma: rate_Omega}]
From \eqref{eqn: unfold_Xt} and \eqref{eqn: def_hat_Sigma_mk}, we may further decompose
\begin{align}
    &\;\quad
    \wh\bOmega_{k,n,\bi}^{(m)} = \frac{1}{T} \sum_{t=1}^T \mat_k(\cX_t^{(m)}) \cdot (\cX_t^{(n)})_{\bi} \notag \\
    &=
    \frac{1}{T} \sum_{t=1}^T \Big\{ \bA_{m,k} \mat_k(\cG_t^{(m)}) \bA_{m,\text{-}k}^\trans + \bB_{m,k} \mat_k(\cF_t^{(m)}) \bB_{m,\text{-}k}^\trans + \mat_k(\cE_t^{(m)}) \Big\} \cdot (\cX_t^{(n)})_{\bi} \notag \\
    &=
    \frac{1}{T} \sum_{t=1}^T \bA_{m,k} \mat_k(\cG_t^{(m)}) \Big( \cG_t^{(n)} \times_{j=1}^{K_n} \bA_{n,j} \Big)_{\bi} \bA_{m,\text{-}k}^\trans \notag \\
    &\;\quad
    + \frac{1}{T} \sum_{t=1}^T \bA_{m,k} \mat_k(\cG_t^{(m)}) \Big( \cE_t^{(n)} \Big)_{\bi} \bA_{m,\text{-}k}^\trans \notag \\
    &\;\quad
    + \frac{1}{T} \sum_{t=1}^T \bA_{m,k} \mat_k(\cG_t^{(m)}) \Big( \cF_t^{(n)} \times_{\ell=1}^{K_n} \bB_{n,\ell} \Big)_{\bi} \bA_{m,\text{-}k}^\trans \notag \\
    &\;\quad
    + \frac{1}{T} \sum_{t=1}^T \bB_{m,k} \mat_k(\cF_t^{(m)}) \Big( \cE_t^{(n)} \Big)_{\bi} \bB_{m,\text{-}k}^\trans \notag \\
    &\;\quad
    + \frac{1}{T} \sum_{t=1}^T \bB_{m,k} \mat_k(\cF_t^{(m)}) \Big( \cG_t^{(n)} \times_{j=1}^{K_n} \bA_{n,j} \Big)_{\bi} \bB_{m,\text{-}k}^\trans \notag \\
    &\;\quad
    + \frac{1}{T} \sum_{t=1}^T \mat_k(\cE_t^{(m)}) \Big( \cG_t^{(n)} \times_{j=1}^{K_n} \bA_{n,j} \Big)_{\bi} \notag \\
    &\;\quad
    + \frac{1}{T} \sum_{t=1}^T \bB_{m,k} \mat_k(\cF_t^{(m)}) \Big( \cF_t^{(n)} \times_{\ell=1}^{K_n} \bB_{n,\ell} \Big)_{\bi} \bB_{m,\text{-}k}^\trans \notag \\
    &\;\quad
    + \frac{1}{T} \sum_{t=1}^T \mat_k(\cE_t^{(m)}) \Big( \cF_t^{(n)} \times_{\ell=1}^{K_n} \bB_{n,\ell} \Big)_{\bi} \notag \\
    &\;\quad
    + \frac{1}{T} \sum_{t=1}^T \mat_k(\cE_t^{(m)}) \Big( \cE_t^{(n)} \Big)_{\bi} \notag \\
    &
    =: \cI_1 + \cI_2 + \cI_3 + \cI_4 + \cI_5 + \cI_6 + \cI_7 + \cI_8 + \cI_9 .
    \label{eqn: decomp_Omega}
\end{align}

Note that by Assumptions~\ref{ass: identification}(b), \ref{ass: noise}, and \ref{ass: dependence_moment}(a), $\cI_2$--$\cI_9$ all have zero mean. For $\cI_2$, consider first
\[
\sum_{\bi \in [p_{n,1}] \times \cdots \times[p_{n,K_n}]} \big\| T^{-1} \sum_{t=1}^T \mat_k(\cG_t^{(m)}) \big( \cE_t^{(n)} \big)_{\bi} \big\|_F^2 .
\]
Then we have
\begin{align*}
    \sum_{\bi \in \substack{[p_{n,1}] \times \cdots\\ \times[p_{n,K_n}] }} \Bigg\| \frac{1}{T} \sum_{t=1}^T \mat_k(\cG_t^{(m)}) \Big( \cE_t^{(n)} \Big)_{\bi} \Bigg\|_F^2 
    &= \Bigg\| \frac{1}{T} \sum_{t=1}^T \vec(\cG_t^{(m)}) \otimes \vec(\cE_t^{(n)}) \Bigg\|_F^2 \\
    &= \cO_P\Big( \frac{p_n}{T} \Big) ,
\end{align*}
where the last equality used Assumption~\ref{ass: dependence_moment}(b) and zero mean of $\cI_2$. Hence, by Assumption~\ref{ass: loadings}(b),
\begin{equation}
\label{eqn: cI2_rate}
\sum_{\bi \in [p_{n,1}] \times \cdots \times[p_{n,K_n}]} \|\cI_2\|_F^2 \leq \cO_P\Big( \frac{p_n}{T} \Big)\cdot \|\bA_{m,k}\|_F^2 \cdot \|\bA_{m,\text{-}k}\|_F^2 = \cO_P\Big( \frac{p_m p_n}{T} \Big) .
\end{equation}

For $\cI_3$--$\cI_9$, note that by Assumption~\ref{ass: loadings}(a), $\Big( \cG_t^{(n)} \times_{j=1}^{K_n} \bA_{n,j} \Big)_{\bi}$ and $\Big( \cF_t^{(n)} \times_{\ell=1}^{K_n} \bB_{n,\ell} \Big)_{\bi}$ are respectively finite linear combinations of elements in $\cG_t^{(n)}$ and $\cF_t^{(n)}$. Then by the similar arguments as the above steps for $\cI_2$, for $j=3, \dots, 9$, we have
\begin{equation}
\label{eqn: cI3_to_cI9_rate}
\sum_{\bi \in [p_{n,1}] \times \cdots \times[p_{n,K_n}]} \|\cI_j\|_F^2 = \cO_P\Big( \frac{p_m p_n}{T} \Big) .
\end{equation}

Finally, consider $\sum_{\bi \in [p_{n,1}] \times \cdots \times[p_{n,K_n}]} \|\cI_1 - \E(\cI_1)\|_F^2$. By Assumption~\ref{ass: loadings}(b), we have
\begin{align}
    &\;\quad
    \sum_{\bi \in [p_{n,1}] \times \cdots \times[p_{n,K_n}]} \|\cI_1 - \E(\cI_1)\|_F^2 \notag \\
    &\leq
    \|\bA_{m,k}\|_F^2 \cdot \|\bA_{m,\text{-}k}\|_F^2
    \sum_{\bi \in \substack{[p_{n,1}] \times \cdots\\ \times[p_{n,K_n}] }} \Bigg\| \frac{1}{T} \sum_{t=1}^T \Bigg[ \mat_k(\cG_t^{(m)}) \Big( \cG_t^{(n)} \times_{j=1}^{K_n} \bA_{n,j} \Big)_{\bi} \notag \\
    &\;\quad
    - \E\Big\{ \mat_k(\cG_t^{(m)}) \Big( \cG_t^{(n)} \times_{j=1}^{K_n} \bA_{n,j} \Big)_{\bi} \Big\} \Bigg] \Bigg\|_F^2 \notag \\
    &\lesssim
    p_m \cdot \sum_{\bi \in \substack{[p_{n,1}] \times \cdots\\ \times[p_{n,K_n}] }} \Bigg\| \frac{1}{T} \sum_{t=1}^T \Bigg[ \vec(\cG_t^{(m)}) \Big( \cG_t^{(n)} \times_{j=1}^{K_n} \bA_{n,j} \Big)_{\bi} \notag \\
    &\;\quad
    - \E\Big\{ \vec(\cG_t^{(m)}) \Big( \cG_t^{(n)} \times_{j=1}^{K_n} \bA_{n,j} \Big)_{\bi} \Big\} \Bigg] \Bigg\|_F^2 \notag \\
    &\lesssim
    p_m \cdot \sum_{\bi \in \substack{[p_{n,1}] \times \cdots\\ \times[p_{n,K_n}] }} \Bigg\| \frac{1}{T} \sum_{t=1}^T \Bigg[ \vec(\cG_t^{(m)}) \otimes \vec(\cG_t^{(n)}) - \E\Big\{ \vec(\cG_t^{(m)}) \otimes \vec(\cG_t^{(n)}) \Big\} \Bigg] \Bigg\|_F^2 \notag \\
    &=
    \cO_P\Big( \frac{p_m p_n}{T} \Big) ,
    \label{eqn: cI1_var}
\end{align}
where the last line used Assumptions~\ref{ass: loadings}(a) and \ref{ass: dependence_moment}(b).

Combining \eqref{eqn: decomp_Omega}, \eqref{eqn: cI2_rate}, \eqref{eqn: cI3_to_cI9_rate}, and \eqref{eqn: cI1_var}, we conclude that
\begin{align*}
    &\;\quad
    \sum_{\bi \in [p_{n,1}] \times \cdots \times[p_{n,K_n}]} \Big\| \wh\bOmega_{k,n,\bi}^{(m)} - \bOmega_{k,n,\bi}^{(m)} \Big\|_F^2 \\
    &\leq
    \sum_{\bi \in [p_{n,1}] \times \cdots \times[p_{n,K_n}]} \|\cI_1 - \E(\cI_1)\|_F^2 + \sum_{j=2}^9 \sum_{\bi \in [p_{n,1}] \times \cdots \times[p_{n,K_n}]} \|\cI_j\|_F^2 = \cO_P\Big( \frac{p_m p_n}{T} \Big) ,
\end{align*}
as desired. This completes the proof of the lemma.
\end{proof}

\begin{lemma}\label{lemma: rate_Sigma}
Let Assumptions~\ref{ass: identification}, \ref{ass: loadings}, \ref{ass: noise}, and \ref{ass: dependence_moment} hold. For any $m\in[M]$, $k\in[K_m]$, $n\neq m$, we have
\begin{align*}
    \Big\| \wh\bSigma_{m,k} -\bSigma_{m,k} \Big\|_F
    = \cO_P\Bigg\{ \frac{p_m}{\sqrt{T}} \Bigg( \sum_{\substack{n=1\\ n\neq m}}^M p_n \Bigg) \Bigg\} .
\end{align*}
\end{lemma}

\begin{proof}[Proof of Lemma~\ref{lemma: rate_Sigma}]
From \eqref{eqn: def_hat_Sigma_mk}, using the triangle inequality and Cauchy--Schwarz inequality,
\begin{align}
    &\;\quad
    \Big\| \wh\bSigma_{m,k} -\bSigma_{m,k} \Big\|_F = \Bigg\| \sum_{\substack{n=1\\ n\neq m}}^M \sum_{\bi \in [p_{n,1}] \times \cdots \times[p_{n,K_n}]} \Big( \wh\bOmega_{k,n,\bi}^{(m)} \wh\bOmega_{k,n,\bi}^{(m) \trans} - \bOmega_{k,n,\bi}^{(m)} \bOmega_{k,n,\bi}^{(m) \trans} \Big) \Bigg\|_F \notag \\
    &=
    \Bigg\| \sum_{\substack{n=1\\ n\neq m}}^M \sum_{\bi \in [p_{n,1}] \times \cdots \times[p_{n,K_n}]} \Bigg\{ \Big( \wh\bOmega_{k,n,\bi}^{(m)} - \bOmega_{k,n,\bi}^{(m)} \Big) \wh\bOmega_{k,n,\bi}^{(m) \trans} + \bOmega_{k,n,\bi}^{(m)} \Big( \wh\bOmega_{k,n,\bi}^{(m)} - \bOmega_{k,n,\bi}^{(m)} \Big)^\trans \Bigg\} \Bigg\|_F \notag \\
    &\leq
    \Bigg\| \sum_{\substack{n=1\\ n\neq m}}^M \sum_{\bi \in [p_{n,1}] \times \cdots \times[p_{n,K_n}]} \Big( \wh\bOmega_{k,n,\bi}^{(m)} - \bOmega_{k,n,\bi}^{(m)} \Big) \Big( \wh\bOmega_{k,n,\bi}^{(m)} - \bOmega_{k,n,\bi}^{(m)} \Big)^\trans \Bigg\|_F \notag \\
    &\;\quad
    + 2\Bigg\| \sum_{\substack{n=1\\ n\neq m}}^M \sum_{\bi \in [p_{n,1}] \times \cdots \times[p_{n,K_n}]} \bOmega_{k,n,\bi}^{(m)} \Big( \wh\bOmega_{k,n,\bi}^{(m)} - \bOmega_{k,n,\bi}^{(m)} \Big)^\trans \Bigg\|_F \notag \\
    &\lesssim
    \sum_{\substack{n=1\\ n\neq m}}^M \sum_{\bi \in \substack{[p_{n,1}] \times \cdots\\ \times[p_{n,K_n}] }} \Big\| \wh\bOmega_{k,n,\bi}^{(m)} - \bOmega_{k,n,\bi}^{(m)} \Big\|_F^2 \notag \\
    &\;\quad
    + \sum_{\substack{n=1\\ n\neq m}}^M \Bigg( \sum_{\bi \in \substack{[p_{n,1}] \times \cdots\\ \times[p_{n,K_n}] }} \Big\| \wh\bOmega_{k,n,\bi}^{(m)} - \bOmega_{k,n,\bi}^{(m)} \Big\|_F^2 \Bigg)^{1/2} \Bigg( \sum_{\bi \in \substack{[p_{n,1}] \times \cdots\\ \times[p_{n,K_n}] }} \Big\| \bOmega_{k,n,\bi}^{(m)} \Big\|_F^2 \Bigg)^{1/2} .
    \label{eqn: whSigma_minus_Sigma_rate}
\end{align}

By Assumption~\ref{ass: loadings}(b), we have
\begin{align}
    &\;\quad
    \sum_{\bi \in \substack{[p_{n,1}] \times \cdots\\ \times[p_{n,K_n}] }} \big\| \bOmega_{k,n,\bi}^{(m)} \big\|_F^2 \notag \\
    &=
    \sum_{\bi \in \substack{[p_{n,1}] \times \cdots\\ \times[p_{n,K_n}] }} \Big\| \bA_{m,k} \E\Big\{ \mat_k(\cG_t^{(m)}) \Big( \cG_t^{(n)} \times_{h=1}^{K_n} \bA_{n,h} \Big)_{\bi} \Big\} \bA_{m,\text{-}k}^\trans \Big\|_F^2 \notag \\
    &\leq
    \|\bA_{m,k}\|_F^2 \cdot \|\bA_{m,\text{-}k}\|_F^2 \cdot \Big\| \E\Big\{ \vec(\cG_t^{(m)}) \vec(\cG_t^{(n)})^\trans (\otimes_{h=1}^{K_n} \bA_{n,h})^\trans \Big\} \Big\|_F^2 \notag \\
    &\lesssim
    p_m \cdot \Big\| \E\Big\{ \vec(\cG_t^{(m)}) \otimes \vec(\cG_t^{(n)})^\trans \Big\} \Big\|^2 \cdot \Big\| \otimes_{h=1}^{K_n} \bA_{n,h} \Big\|_F^2
    = \cO(p_m p_n) ,
    \label{eqn: sum_i_Omega_rate}
\end{align}
where the last equality used Assumption~\ref{ass: core_factor}(c). Combining \eqref{eqn: whSigma_minus_Sigma_rate}, \eqref{eqn: sum_i_Omega_rate}, and Lemma~\ref{lemma: rate_Omega}, we concludes the proof of this lemma.
\end{proof}

\begin{lemma}\label{lemma: eigenvalue_Sigma}
Let Assumptions~\ref{ass: identification}, \ref{ass: loadings}, \ref{ass: core_factor}, \ref{ass: noise}, and \ref{ass: dependence_moment} hold. For any $m\in[M]$, $k\in[K_m]$,
\begin{align*}
    \lambda_{1}\big( \bSigma_{m,k} \big) \asymp \dots \asymp \lambda_{r_{m,k}}\big( \bSigma_{m,k} \big) \asymp p_m \Bigg( \sum_{\substack{n=1\\ n\neq m}}^M p_n \Bigg) .
\end{align*}
Also, all remaining eigenvalues of $\bSigma_{m,k}$ are zero and thus $\bSigma_{m,k}$ has rank $r_{m,k}$.
\end{lemma}

\begin{proof}[Proof of Lemma~\ref{lemma: eigenvalue_Sigma}]
Using Assumption~\ref{ass: loadings}(b), we have
\begin{align}
    &\;\quad
    \lambda_{r_{m,k}}\big( \bSigma_{m,k} \big) = \lambda_{r_{m,k}}\Bigg( \sum_{\substack{n=1\\ n\neq m}}^M \sum_{\bi \in [p_{n,1}] \times \cdots \times[p_{n,K_n}]} \bOmega_{k,n,\bi}^{(m)} \bOmega_{k,n,\bi}^{(m) \trans} \Bigg) \notag \\
    &=
    \lambda_{r_{m,k}}\Bigg( \bA_{m,k} \sum_{\substack{n=1\\ n\neq m}}^M \sum_{\bi \in \substack{[p_{n,1}] \times \cdots\\ \times[p_{n,K_n}] }}
    \E\Big\{ \mat_k(\cG_t^{(m)}) \Big( \cG_t^{(n)} \times_{h=1}^{K_n} \bA_{n,h} \Big)_{\bi} \Big\} \bA_{m,\text{-}k}^\trans \bA_{m,\text{-}k} \notag \\
    &\;\quad
    \cdot \E\Big\{ \mat_k(\cG_t^{(m)}) \Big( \cG_t^{(n)} \times_{h=1}^{K_n} \bA_{n,h} \Big)_{\bi} \Big\}^\trans
    \bA_{m,k}^\trans \Bigg) \notag \\
    &\geq
    \lambda_{r_{m,k}}\big( \bA_{m,k} \bA_{m,k}^\trans \big) \cdot \lambda_{r_{m,k}}\Bigg( \sum_{\substack{n=1\\ n\neq m}}^M \sum_{\bi \in \substack{[p_{n,1}] \times \cdots\\ \times[p_{n,K_n}] }} \E\Big\{ \mat_k(\cG_t^{(m)}) \Big( \cG_t^{(n)} \times_{h=1}^{K_n} \bA_{n,h} \Big)_{\bi} \Big\} \notag \\
    &\;\quad
    \cdot \bA_{m,\text{-}k}^\trans \bA_{m,\text{-}k} \E\Big\{ \mat_k(\cG_t^{(m)}) \Big( \cG_t^{(n)} \times_{h=1}^{K_n} \bA_{n,h} \Big)_{\bi} \Big\}^\trans
    \Bigg) \notag \\
    &\geq
    \Big\{ \lambda_{r_{m,k}}\big( p_{m,k} \bI \big) - \lambda_1\big( p_{m,k} \bI - \bA_{m,k} \bA_{m,k}^\trans \big) \Big\} \notag \\
    &\;\quad
    \cdot \Bigg\{ \lambda_{r_{m,k}}\Bigg( \sum_{\substack{n=1\\ n\neq m}}^M \sum_{\bi \in \substack{[p_{n,1}] \times \cdots\\ \times[p_{n,K_n}] }} \E\Big\{ \mat_k(\cG_t^{(m)}) \Big( \cG_t^{(n)} \times_{h=1}^{K_n} \bA_{n,h} \Big)_{\bi} \Big\} \notag \\
    &\;\quad
    \cdot p_{m,\text{-}k} \bI \cdot \E\Big\{ \mat_k(\cG_t^{(m)}) \Big( \cG_t^{(n)} \times_{h=1}^{K_n} \bA_{n,h} \Big)_{\bi} \Big\}^\trans \Bigg) \notag \\
    &\;\quad
    - \lambda_{1}\Bigg( \sum_{\substack{n=1\\ n\neq m}}^M \sum_{\bi \in \substack{[p_{n,1}] \times \cdots\\ \times[p_{n,K_n}] }} \E\Big\{ \mat_k(\cG_t^{(m)}) \Big( \cG_t^{(n)} \times_{h=1}^{K_n} \bA_{n,h} \Big)_{\bi} \Big\} \notag \\
    &\;\quad
    \cdot \big( p_{m,\text{-}k} \bI - \bA_{m,\text{-}k}^\trans \bA_{m,\text{-}k} \big) \cdot \E\Big\{ \mat_k(\cG_t^{(m)}) \Big( \cG_t^{(n)} \times_{h=1}^{K_n} \bA_{n,h} \Big)_{\bi} \Big\}^\trans \Bigg) \Bigg\} \notag \\
    &\asymp
    p_{m,k} \cdot p_{m,\text{-}k} \cdot \lambda_{r_{m,k}}\Bigg( \sum_{\substack{n=1\\ n\neq m}}^M \sum_{\bi \in \substack{[p_{n,1}] \times \cdots\\ \times[p_{n,K_n}] }} \E\Big\{ \mat_k(\cG_t^{(m)}) \Big( \cG_t^{(n)} \times_{h=1}^{K_n} \bA_{n,h} \Big)_{\bi} \Big\} \notag \\
    &\;\quad
    \cdot \E\Big\{ \mat_k(\cG_t^{(m)}) \Big( \cG_t^{(n)} \times_{h=1}^{K_n} \bA_{n,h} \Big)_{\bi} \Big\}^\trans \Bigg) \notag \\
    &\geq
    p_{m,k} \cdot p_{m,\text{-}k} \cdot \sum_{\substack{n=1\\ n\neq m}}^M \lambda_{r_{m,k}}\Bigg( \sum_{\bi \in \substack{[p_{n,1}] \times \cdots\\ \times[p_{n,K_n}] }} \E\Big\{ \mat_k(\cG_t^{(m)}) \Big( \cG_t^{(n)} \times_{h=1}^{K_n} \bA_{n,h} \Big)_{\bi} \Big\} \notag \\
    &\;\quad
    \cdot \E\Big\{ \mat_k(\cG_t^{(m)}) \Big( \cG_t^{(n)} \times_{h=1}^{K_n} \bA_{n,h} \Big)_{\bi} \Big\}^\trans \Bigg) . \label{eqn: eigenvalue_Sigma_step}
\end{align}
where the first and second inequalities used Weyl's inequality and the fact that for any compatible matrices $\bX,\bY$ with $\bX\bY$ being positive semi-definite ($\bSigma_{m,k}$ is $\bX\bY$ here), then 
any nonzero eigenvalues 
$\lambda_i(\bX\bY)=\lambda_i(\bY\bX)$; and the asymptotic equality (in probability) used Assumption~\ref{ass: loadings}(b) and Lemma~\ref{lemma: inequality_sandwich}(a); and the last inequality used Weyl's inequality again.

Next, consider the following for any $n\neq m$.
First, for any $\bi\in [p_{n,1}] \times \cdots \times[p_{n,K_n}]$, there exists some $i\in[p_n]$ such that we can write
\begin{align*}
    &\;\quad
    \E\Big\{ \mat_k(\cG_t^{(m)}) \Big( \cG_t^{(n)} \times_{h=1}^{K_n} \bA_{n,h} \Big)_{\bi} \Big\} \\
    &=
    \E\Big[ \mat_k(\cG_t^{(m)}) \otimes \Big\{ \be_{p_n,i}^\trans \big( \otimes_{h=1}^{K_n} \bA_{n,h} \big) \vec(\cG_t^{(n)}) \Big\} \Big] \\
    &=
    \E\Big[ \mat_k(\cG_t^{(m)}) \otimes \Big\{ \vec(\cG_t^{(n)})^\trans \big( \otimes_{h=1}^{K_n} \bA_{n,h} \big)^\trans \be_{p_n,i} \Big\} \Big] ,
\end{align*}
where $\be_{p_n,i} \in\R^{p_n}$ contains $1$ in the $i$th entry and zero otherwise. Then
\begin{align}
    &\;\quad
    \lambda_{r_{m,k}}\Bigg( \sum_{\bi \in \substack{[p_{n,1}] \times \cdots\\ \times[p_{n,K_n}] }} \E\Big\{ \mat_k(\cG_t^{(m)}) \Big( \cG_t^{(n)} \times_{h=1}^{K_n} \bA_{n,h} \Big)_{\bi} \Big\} \notag \\
    &\;\quad
    \cdot \E\Big\{ \mat_k(\cG_t^{(m)}) \Big( \cG_t^{(n)} \times_{h=1}^{K_n} \bA_{n,h} \Big)_{\bi} \Big\}^\trans \Bigg) \notag \\
    &=
    \lambda_{r_{m,k}}\Bigg( \sum_{i \in [p_n]} \E\Big[ \mat_k(\cG_t^{(m)}) \otimes \Big\{ \vec(\cG_t^{(n)})^\trans \big( \otimes_{h=1}^{K_n} \bA_{n,h} \big)^\trans \be_{p_n,i} \Big\} \Big] \notag \\
    &\;\quad
    \cdot \E\Big[ \mat_k(\cG_t^{(m)}) \otimes \Big\{ \vec(\cG_t^{(n)})^\trans \big( \otimes_{h=1}^{K_n} \bA_{n,h} \big)^\trans \be_{p_n,i} \Big\} \Big]^\trans \Bigg) \notag \\
    &=
    \lambda_{r_{m,k}}\Bigg( \E\Big[ \mat_k(\cG_t^{(m)}) \otimes \Big\{ \vec(\cG_t^{(n)})^\trans \big( \otimes_{h=1}^{K_n} \bA_{n,h} \big)^\trans \Big\} \Big] \notag \\
    &\;\quad
    \cdot \E\Big[ \mat_k(\cG_t^{(m)})^\trans \otimes \Big\{ \big( \otimes_{h=1}^{K_n} \bA_{n,h} \big) \vec(\cG_t^{(n)}) \Big\} \Big] \Bigg) 
    \asymp
    p_n \cdot c ,
    \label{eqn: eigenvalue_Sigma_step_2}
\end{align}
where $c$ is some positive constant and the last line used Assumptions~\ref{ass: loadings}(b),
\ref{ass: core_factor}(c), and Lemma~\ref{lemma: eigenvalue_Sigma_aux}. 
Plugging the above back in \eqref{eqn: eigenvalue_Sigma_step}, we conclude there exists some constant $c$ such that
\[
\lambda_{r_{m,k}}\big( \bSigma_{m,k} \big) \geq c \cdot p_m \sum_{n\in[M] \setminus\{m\}} p_n .
\]

It remains to show the same rate as an upper bound of $\lambda_{1}\big( \bSigma_{m,k} \big)$, but this is analogous to all the above steps with signs of inequalities reversed. 
As the last statement in the theorem is trivial (due to $\bA_{m,k}$ having $r_{m,k}$ columns),
the proof of this lemma is complete.
\end{proof}

\begin{lemma}[Useful inequalities to bound sandwiched terms]\label{lemma: inequality_sandwich}
We have the following.
\begin{enumerate}[itemsep=0pt, label = (\alph*), left = 0pt]
    \item For any square matrix $\bM$ and sequence of matrices $\bA_t$ with compatible dimensions for $t\in[T]$, we have
    \[
    \left\| \sum_{t=1}^T \bA_t \bM \bA_t^\trans \right\|_F^2 \leq \|\bM\|^2 \cdot \left\| \sum_{t=1}^T \bA_t \bA_t^\trans \right\|_F^2 .
    \]
    Note that this inequality is tight, which can be seen by taking $\bM$ as the identity matrix.
    \item (Lemma~B.16 in \cite{Barigozzietal2025_robust}.) For any sequence of matrices $\bA_t=[\bA_{t,ij}]$, $\bB_t$ for $t\in[T]$, and any matrix $\bM$ with compatible dimensions, we have
    \[
    \left\| \sum_{t=1}^T \bA_t \bM \bB_t^\trans \right\|_F^2 \leq \|\bM\|_F^2 \cdot \sum_{i,j} \left\| \sum_{t=1}^T \bA_{t,ij} \bB_t \right\|_F^2 .
    \]
\end{enumerate}
\end{lemma}

\begin{proof}[Proof of Lemma~\ref{lemma: inequality_sandwich}]
Part~(b) is a direct quotation of Lemma~B.16 in \cite{Barigozzietal2025_robust}, so it suffices to show part~(a). Expanding the squared Frobenius norm, we have
\begin{equation}
\label{eqn: sum_AMA_decomp}
\begin{split}
    \left\| \sum_{t=1}^T \bA_t \bM \bA_t^\trans \right\|_F^2 &= \tr\left\{ \left( \sum_{t=1}^T \bA_t \bM \bA_t^\trans \right) \left( \sum_{s=1}^T \bA_s \bM^\trans \bA_s^\trans \right) \right\} \\
    &=
    \sum_{t=1}^T \sum_{s=1}^T \tr(\bA_t \bM \bA_t^\trans \bA_s \bM^\trans \bA_s^\trans ) \\
    &=
    \sum_{t=1}^T \sum_{s=1}^T \tr(\bA_s^\trans \bA_t \bM \bA_t^\trans \bA_s \bM^\trans )
    =: \sum_{t=1}^T \sum_{s=1}^T \tr(\bA_{st} \bM \bA_{st}^\trans \bM^\trans ) ,
\end{split}
\end{equation}
where the second line used the cyclic property of trace, and the definition $\bA_{st}= \bA_s^\trans \bA_t$. For any $t,s\in[T]$, we may read $\tr(\bA_{st} \bM \bA_{st}^\trans \bM^\trans)$ as a Frobenius inner product and apply the Cauchy--Schwarz inequality such that
\begin{equation}
\label{eqn: trace_AMAM}
\tr(\bA_{st} \bM \bA_{st}^\trans \bM^\trans) = \langle \bA_{st} \bM, \bM \bA_{st} \rangle_F \leq \|\bA_{st} \bM\|_F \cdot \|\bM \bA_{st}\|_F \leq \|\bM\|^2 \|\bA_{st}\|_F^2 ,
\end{equation}
where the last inequality used the fact that for any matrix products $\bA\bB$, we have 
\[
\|\bA\bB\|_F=\sqrt{\sum_{i=1}^n \|\bA_{i\cdot}^\top \bB\|^2} \leq \sqrt{\sum_{i=1}^n \|\bA_{i\cdot}^\top\|^2 \|\bB\|^2} = \|\bA\|_F \cdot \|\bB\|.
\]

Combining \eqref{eqn: sum_AMA_decomp} and \eqref{eqn: trace_AMAM}, we have
\begin{align*}
    &\quad
    \left\| \sum_{t=1}^T \bA_t \bM \bA_t^\trans \right\|_F^2 = \sum_{t=1}^T \sum_{s=1}^T \tr(\bA_{st} \bM \bA_{st}^\trans \bM^\trans) \\
    &\leq
    \|\bM\|^2 \cdot \sum_{t=1}^T \sum_{s=1}^T \|\bA_{st}\|_F^2
    = \|\bM\|^2 \sum_{t=1}^T \sum_{s=1}^T \tr(\bA_s^\trans \bA_t \bA_t^\trans \bA_s) \\
    &=
    \|\bM\|^2 \cdot \tr\left\{ \left( \sum_{t=1}^T \bA_t \bA_t^\trans \right) \left( \sum_{s=1}^T \bA_s \bA_s^\trans \right) \right\}
    = \|\bM\|^2 \left\| \sum_{t=1}^T \bA_t \bA_t^\trans \right\|_F^2 ,
\end{align*}
as desired. This completes the proof of this lemma.
\end{proof}

\begin{lemma}
\label{lemma: A_perp_rate}
Let all assumptions in Theorem~\ref{thm: global_loading} hold. For any $m\in[M]$, $k\in[K_m]$, define the projection matrices
\begin{align*}
    \bP_{A,m,k} &:= \bA_{m,k} \big(\bA_{m,k}^\trans \bA_{m,k} \big)^{-1} \bA_{m,k}^\trans ,
    \quad 
    \bP_{A,m,k,\perp} := \bI_{p_{m,k}} - \bP_{A,m,k} , \\
    \wh\bP_{A,m,k} &:=
    \frac{1}{p_{m,k}} \wh\bA_{m,k} \wh\bA_{m,k}^\trans ,
    \quad
    \wh\bP_{A,m,k,\perp} := \bI_{p_{m,k}} - \wh\bP_{A,m,k} .
\end{align*}
Also denote the orthogonal complement of $\bA_{m,k}$ as $\bA_{m,k,\perp}$ which is defined as an orthogonal basis spanning the complement of $\col(\bA_{m,k})$.
Then it holds that
\begin{align*}
    \Big\| \wh\bP_{A,m,k} - \bP_{A,m,k} \Big\|_F 
    = \cO_P\Big( \frac{1}{\sqrt{T}} \Big) &, \\
    \Big\| \wh\bP_{A,m,k,\perp} - \bP_{A,m,k,\perp} \Big\|_F 
    = \cO_P\Big( \frac{1}{\sqrt{T}} \Big) &, \\
    \Big\|  \wh\bA_{m,k,\perp} \wh\bA_{m,k,\perp}^\trans - \bA_{m,k,\perp} \bA_{m,k,\perp}^\trans \Big\|_F 
    = \cO_P\Big( \frac{1}{\sqrt{T}} \Big) &.
\end{align*}
\end{lemma}

\begin{proof}[Proof of Lemma~\ref{lemma: A_perp_rate}]
Let $\bQ_{m,k}\in \R^{p_{m,k}\times r_{m,k}}$ be an orthonormal basis of $\col(\bA_{m,k})$, so that
\[
\bP_{A,m,k} = \bQ_{m,k} \bQ_{m,k}^\trans .
\]
Also, let $\wh{\bQ}_{m,k}$ be the matrix of the leading $r_{m,k}$ eigenvectors of $\wh\bSigma_{m,k}$. Recall that $\wh\bA_{m,k}$ is defined as $\sqrt{p_{m,k}}$ times those eigenvectors, we have
\[
\wh\bP_{A,m,k} =
\frac{1}{p_{m,k}} \wh\bA_{m,k} \wh\bA_{m,k}^\trans
= \wh\bQ_{m,k} \wh\bQ_{m,k}^\trans .
\]

From Lemma~\ref{lemma: eigenvalue_Sigma}, the leading eigenspace of $\bSigma_{m,k}$ is exactly $\col(\bA_{m,k})$, and particularly,
\[
\lambda_{r_{m,k}}(\bSigma_{m,k})
\asymp 
p_m \Bigg( \sum_{\substack{n=1\\ n\neq m}}^M p_n \Bigg) ,
\quad
\lambda_{r_{m,k}+1}(\bSigma_{m,k}) = 0 .
\]
Hence by the Davis--Kahan theorem \citep[e.g.][]{yu2015useful}, the $\sin\Theta$ distance between $\bQ_{m,k}$ and $\wh\bQ_{m,k}$ satisfies that
\[
\Big\| \sin\Theta \big( \wh\bQ_{m,k}, \bQ_{m,k} \big) \Big\|_F
\lesssim
\frac{ \big\| \wh\bSigma_{m,k}-\bSigma_{m,k} \big\|_F}{ \lambda_{r_{m,k}}(\bSigma_{m,k}) - \lambda_{r_{m,k}+1}(\bSigma_{m,k})} .
\]
Further with the identity (see e.g.\ Section~2 in \citeauthor{Fanetal2019}, \citeyear{Fanetal2019}):
\[
\Big\| \wh\bQ_{m,k}
\wh\bQ_{m,k}^\trans - \bQ_{m,k} \bQ_{m,k}^\trans \Big\|_F
=
\sqrt{2} \, \Big\| \sin\Theta \big(\wh\bQ_{m,k}, \bQ_{m,k}\big) \Big\|_F ,
\]
we may obtain
\[
\Big\| \wh\bP_{A,m,k} - \bP_{A,m,k} \Big\|_F
\lesssim
\frac{ \big\| \wh\bSigma_{m,k}-\bSigma_{m,k} \big\|_F}{ \lambda_{r_{m,k}}(\bSigma_{m,k}) } 
=
\cO_P\Big( \frac{1}{\sqrt{T}} \Big) ,
\]
where the last inequality used Lemmas~\ref{lemma: rate_Sigma} and \ref{lemma: eigenvalue_Sigma}.
This proves the first result.
The second statement follows immediately from the first and the definitions of $\wh\bP_{A,m,k,\perp}$ and $\bP_{A,m,k,\perp}$.

It remains to show the third result.
To this end, note that $\bA_{m,k,\perp}$ and $\wh\bA_{m,k,\perp}$ are orthonormal bases of the orthogonal complements of $\col(\bA_{m,k})$ and $\col(\wh\bA_{m,k})$, respectively.
Hence we can write
\[
\bP_{A,m,k,\perp} = \bA_{m,k,\perp} \bA_{m,k,\perp}^\trans , 
\quad
\wh\bP_{A,m,k,\perp} = \wh\bA_{m,k,\perp} \wh\bA_{m,k,\perp}^\trans .
\]
Thus the desired results follows directly from the second result of this lemma, which completes the proof.
\end{proof}

\begin{lemma}\label{lemma: eigenvalue_Sigma_B}
Let Assumptions~\ref{ass: identification}, \ref{ass: loadings}, \ref{ass: core_factor}, \ref{ass: noise}, \ref{ass: dependence_moment}, and \ref{ass: dependence_thread} hold.
Given any $m\in[M]$, $k\in[K_m]$ with $K_m> 1$, recall $\wh{\bD}_{m,k}$ from the statement of Theorem~\ref{thm: local_loading}.
Then the eigenvalues of $\wh{\bD}_{m,k}$ satisfy $\lambda_i(\wh{\bD}_{m,k}) = \lambda_i(\bSigma_{F,m,k}) + o_P(1)$ for any $i\leq u_{m,k}$.
\end{lemma}

\begin{proof}[Proof of Lemma~\ref{lemma: eigenvalue_Sigma_B}]
By the definition in \eqref{eqn: def_hat_Sigma_B_mk}, we can decompose
\begin{align}
    &\;\quad
    \wh\bSigma_{B,m,k} = \frac{1}{Tp_m} \sum_{t=1}^T \mat_k(\cX_t^{(m)}) \wh{\bA}_{m,\text{-}k,\perp} \wh{\bA}_{m,\text{-}k,\perp}^\trans \mat_k(\cX_t^{(m)})^\trans \notag \\
    &=
    \frac{1}{Tp_m} \sum_{t=1}^T \Big( \bB_{m,k} \mat_k(\cF_t^{(m)}) \bB_{m,\text{-}k}^\trans \wh{\bA}_{m,\text{-}k,\perp} + \mat_k(\cE_t^{(m)}) \wh{\bA}_{m,\text{-}k,\perp} \notag \\
    &\;\quad
    + \bA_{m,k} \mat_k(\cG_t^{(m)}) \bA_{m,\text{-}k}^\trans \wh{\bA}_{m,\text{-}k,\perp} \Big) 
    \cdot \Big( \bB_{m,k} \mat_k(\cF_t^{(m)}) \bB_{m,\text{-}k}^\trans \wh{\bA}_{m,\text{-}k,\perp} \notag \\
    &\;\quad
    + \mat_k(\cE_t^{(m)}) \wh{\bA}_{m,\text{-}k,\perp} + \bA_{m,k} \mat_k(\cG_t^{(m)}) \bA_{m,\text{-}k}^\trans \wh{\bA}_{m,\text{-}k,\perp} \Big)^\trans \notag \\
    &=
    \frac{1}{Tp_m} \sum_{t=1}^T \bB_{m,k} \mat_k(\cF_t^{(m)}) \bB_{m,\text{-}k}^\trans \wh{\bA}_{m,\text{-}k,\perp} \wh{\bA}_{m,\text{-}k,\perp}^\trans \bB_{m,\text{-}k} \mat_k(\cF_t^{(m)})^\trans \bB_{m,k}^\trans \notag \\
    &\;\quad
    + \frac{1}{Tp_m} \sum_{t=1}^T \bB_{m,k} \mat_k(\cF_t^{(m)}) \bB_{m,\text{-}k}^\trans \wh{\bA}_{m,\text{-}k,\perp} \wh{\bA}_{m,\text{-}k,\perp}^\trans \mat_k(\cE_t^{(m)})^\trans \notag \\
    &\;\quad
    + \frac{1}{Tp_m} \sum_{t=1}^T \bB_{m,k} \mat_k(\cF_t^{(m)}) \bB_{m,\text{-}k}^\trans \wh{\bA}_{m,\text{-}k,\perp} \wh{\bA}_{m,\text{-}k,\perp}^\trans \bA_{m,\text{-}k} \mat_k(\cG_t^{(m)})^\trans \bA_{m,k}^\trans \notag \\
    &\;\quad
    + \frac{1}{Tp_m} \sum_{t=1}^T \mat_k(\cE_t^{(m)}) \wh{\bA}_{m,\text{-}k,\perp} \wh{\bA}_{m,\text{-}k,\perp}^\trans \bB_{m,\text{-}k} \mat_k(\cF_t^{(m)})^\trans \bB_{m,k}^\trans \notag \\
    &\;\quad
    + \frac{1}{Tp_m} \sum_{t=1}^T \mat_k(\cE_t^{(m)}) \wh{\bA}_{m,\text{-}k,\perp} \wh{\bA}_{m,\text{-}k,\perp}^\trans \mat_k(\cE_t^{(m)})^\trans \notag \\
    &\;\quad
    + \frac{1}{Tp_m} \sum_{t=1}^T \mat_k(\cE_t^{(m)}) \wh{\bA}_{m,\text{-}k,\perp} \wh{\bA}_{m,\text{-}k,\perp}^\trans \bA_{m,\text{-}k} \mat_k(\cG_t^{(m)})^\trans \bA_{m,k}^\trans \notag \\
    &\;\quad
    + \frac{1}{Tp_m} \sum_{t=1}^T \bA_{m,k} \mat_k(\cG_t^{(m)}) \bA_{m,\text{-}k}^\trans \wh{\bA}_{m,\text{-}k,\perp} \wh{\bA}_{m,\text{-}k,\perp}^\trans \bB_{m,\text{-}k} \mat_k(\cF_t^{(m)})^\trans \bB_{m,k}^\trans \notag \\
    &\;\quad
    + \frac{1}{Tp_m} \sum_{t=1}^T \bA_{m,k} \mat_k(\cG_t^{(m)}) \bA_{m,\text{-}k}^\trans \wh{\bA}_{m,\text{-}k,\perp} \wh{\bA}_{m,\text{-}k,\perp}^\trans \mat_k(\cE_t^{(m)})^\trans \notag \\
    &\;\quad
    + \frac{1}{Tp_m} \sum_{t=1}^T \bA_{m,k} \mat_k(\cG_t^{(m)}) \bA_{m,\text{-}k}^\trans \wh{\bA}_{m,\text{-}k,\perp} \wh{\bA}_{m,\text{-}k,\perp}^\trans \bA_{m,\text{-}k} \mat_k(\cG_t^{(m)})^\trans \bA_{m,k}^\trans \notag \\
    &=:
    \cII_1 + \cII_2 + \cII_3 + \cII_4 + \cII_5 + \cII_6 + \cII_7 + \cII_8 + \cII_9 .
    \label{eqn: decomp_Sigma_B}
\end{align}

Consider $\cII_1$ first. 
By Lemmas~\ref{lemma: inequality_sandwich} and \ref{lemma: A_perp_rate}, we can write
\begin{align*}
    &\;\quad
    \cII_1 = \frac{1}{Tp_m} \sum_{t=1}^T \bB_{m,k} \mat_k(\cF_t^{(m)}) \bB_{m,\text{-}k}^\trans \wh{\bA}_{m,\text{-}k,\perp} \wh{\bA}_{m,\text{-}k,\perp}^\trans \bB_{m,\text{-}k} \mat_k(\cF_t^{(m)})^\trans \bB_{m,k}^\trans \\
    &=
    \frac{1}{Tp_m} \sum_{t=1}^T \bB_{m,k} \mat_k(\cF_t^{(m)}) \bB_{m,\text{-}k}^\trans \bA_{m,\text{-}k,\perp} \bA_{m,\text{-}k,\perp}^\trans \bB_{m,\text{-}k} \mat_k(\cF_t^{(m)})^\trans \bB_{m,k}^\trans \\
    &\;\quad
    + o_P\Big\{ \frac{1}{Tp_m} \sum_{t=1}^T \bB_{m,k} \mat_k(\cF_t^{(m)}) \bB_{m,\text{-}k}^\trans \bB_{m,\text{-}k} \mat_k(\cF_t^{(m)})^\trans \bB_{m,k}^\trans \Big\} \\
    &=: \cII_{\Sigma,1} + o_P(\cII_{\Sigma,2}) .
\end{align*}

{
By Assumption~\ref{ass: loadings}(b), for any $j\in[K_m]$, we have
\begin{equation}
\label{eqn: B_trans_B}
\frac{1}{p_{m,j}} \bB_{m,j}^\trans \bB_{m,j} = \bI_{u_{m,k}} + \cO(p_{m,j}^{-1/2}),
\quad
\frac{1}{p_{m,j}} \bB_{m,j}^\trans \bA_{m,j} = \cO(p_{m,j}^{-1/2}),
\end{equation}
with the $o(1)$ term in terms of the spectral norm.
Recall the notations in Lemma~\ref{lemma: A_perp_rate}, and denote the population projection matrix onto the orthogonal complement of $\col(\bA_{m,j})$ as
\[
\bP_{A,m,j,\perp} = 
\bI_{p_{m,j}} - \bA_{m,j} \big( \bA_{m,j}^\trans \bA_{m,j} \big)^{-1} \bA_{m,j}^\trans 
= \bA_{m,j,\perp} \bA_{m,j,\perp}^\trans ,
\]
where the last equality used the proof of Lemma~\ref{lemma: A_perp_rate}.
Then we have
\[
\frac{1}{ p_{m,j}} \bB_{m,j}^\trans \bA_{m,j,\perp} \bA_{m,j,\perp}^\trans \bB_{m,j}
=
\frac{1}{ p_{m,j}} \bB_{m,j}^\trans \bP_{A,m,j,\perp} \bB_{m,j}
= 
\bI_{u_{m,k}} + \cO(p_{m,j}^{-1/2}) .
\]
Hence for 
$\bP_{A,m,\text{-}k,\perp} := \otimes_{j\in [K_m] \setminus\{k\}} \bP_{A,m,j,\perp}
= \bA_{m,\text{-}k,\perp} \bA_{m,\text{-}k,\perp}^\trans$, 
we have (the dimension on the identity matrix is $\prod_{j\in [K_m] \setminus\{k\}} u_{m,j}$ and suppressed below for neatness):
\begin{align}
    \frac{1}{p_{m,\text{-}k}}
    \bB_{m,\text{-}k}^\trans \bA_{m,\text{-}k,\perp} \bA_{m,\text{-}k,\perp}^\trans\bB_{m,\text{-}k}
    &=
    \frac{1}{p_{m,\text{-}k}}
    \bB_{m,\text{-}k}^\trans \bP_{A,m,\text{-}k,\perp} \bB_{m,\text{-}k} \notag \\
    &=
    \bI + \cO\Big( \sum_{j\in [K_m]\setminus \{k\}} p_{m,j}^{-1/2} \Big) 
    = \bI + o(1) .
    \label{eqn: B_A_perp_A_perp_B}
\end{align}

}

Together with Lemma~\ref{lemma: inequality_sandwich}(a) and the implication of \eqref{eqn: B_trans_B} on $\cII_{\Sigma,2}$,
the term $\cII_{\Sigma}$ dominates the $o_P(\cII_{\Sigma,2})$ term, and hence it remains to consider $\cII_{\Sigma,1}$ for $\cII_1$. 
To this end, by Lemma~\ref{lemma: inequality_sandwich}(a) again, the $o(1)$ term in \eqref{eqn: B_A_perp_A_perp_B} is dominated when plugged in $\cII_{\Sigma,1}$. We thus only consider
\[
\cII_1 \asymp \frac{1}{Tp_{m,k}} \sum_{t=1}^T \bB_{m,k} \mat_k(\cF_t^{(m)}) \mat_k(\cF_t^{(m)})^\trans \bB_{m,k}^\trans \xrightarrow{P} \frac{1}{p_{m,k}} \bB_{m,k} \bSigma_{F,m,k} \bB_{m,k}^\trans ,
\]
according to Assumption~\ref{ass: core_factor}(b). Hence with Assumption~\ref{ass: loadings}(b) again, we conclude $\lambda_i(\cII_1) = \lambda_i(\bSigma_{F,m,k}) + o_P(1)$ for all $i\leq u_{m,k}$, while $\lambda_i(\cII_1) =0$ for $i> u_{m,k}$ by $\rank(\cII_1)\leq u_{m,k}$.

Next, consider $\cII_2$. By definition, we have
\begin{align}
    &\;\quad
    \|\cII_2\|_F^2 = \Bigg\| \frac{1}{Tp_m} \sum_{t=1}^T \bB_{m,k} \mat_k(\cF_t^{(m)}) \bB_{m,\text{-}k}^\trans \wh{\bA}_{m,\text{-}k,\perp} \wh{\bA}_{m,\text{-}k,\perp}^\trans \mat_k(\cE_t^{(m)})^\trans \Bigg\|_F^2 \notag \\
    &\asymp
    \sum_{i=1}^{p_{m,k}} \Bigg\| \frac{1}{Tp_m} \sum_{j=1}^{r_{m,\text{-}k}} \sum_{t=1}^T \bB_{m,k} \mat_k(\cF_t^{(m)})_{\cdot j} (\bB_{m,\text{-}k}^\trans \bA_{m,\text{-}k,\perp} \bA_{m,\text{-}k,\perp}^\trans)_{j\cdot}^\trans \mat_k(\cE_t^{(m)})_{i\cdot} \Bigg\|_F^2 \notag \\
    &\leq
    \frac{1}{T p_m^2} \cdot \big\| \bB_{m,k} \big\|_F^2 \cdot p_{m,\text{-}k} \cdot \max_{j\in[r_{m,\text{-}k}]} \sum_{i=1}^{p_{m,k}} \Bigg\| \frac{1}{\sqrt{T}} \sum_{t=1}^T \mat_k(\cF_t^{(m)})_{\cdot j} \notag \\
    &\;\quad
    \cdot \Bigg\{ \be_i^\trans \mat_k(\cE_t^{(m)}) \Bigg( \frac{\bB_{m,\text{-}k}^\trans \bA_{m,\text{-}k,\perp} \bA_{m,\text{-}k,\perp}^\trans}{\sqrt{p_{m,\text{-}k}}} \Bigg)_{j\cdot} \Bigg\} \Bigg\|_F^2 \notag \\
    &\leq
    \frac{1}{T p_m^2} \cdot \big\| \bB_{m,k} \big\|_F^2 \cdot p_{m,\text{-}k} \cdot \cO_P(p_{m,k})
    = \cO_P\Big( \frac{1}{T p_{m,\text{-}k}} \Big) ,
    \label{eqn: cII_2}
\end{align}
where $\be_i \in\R^{p_{m,k}}$ contains $1$ in its $i$th entry and $0$ elsewhere, and the last line used Assumptions~\ref{ass: loadings}(b) and \ref{ass: dependence_thread}(a). The same arguments hold for $\cII_4$, i.e.,
\begin{align}
    &\;\quad
    \|\cII_4\|_F^2 = \cO_P\Big( \frac{1}{T p_{m,\text{-}k}} \Big) .
    \label{eqn: cII_4}
\end{align}

For $\cII_3$ (and hence $\cII_7$), we have the following by similar steps in \eqref{eqn: cII_2}, Assumptions~\ref{ass: loadings}(b) and \ref{ass: dependence_thread}(b):
\begin{align}
    &\;\quad
    \|\cII_3\|_F^2, \; \|\cII_7\|_F^2 = \cO_P\Big( \frac{1}{T p_{m,\text{-}k}} \Big) .
    \label{eqn: cII_3_7}
\end{align}

Consider $\cII_6$ now. Using Lemma~\ref{lemma: A_perp_rate}, we have
\begin{equation}
\label{eqn: hat_A_perp_minus_A_kron}
    \Big\| \wh{\bA}_{m,\text{-}k,\perp} \wh{\bA}_{m,\text{-}k,\perp}^\trans - \bA_{m,\text{-}k,\perp} \bA_{m,\text{-}k,\perp}^\trans \Big\|_F = \cO_P\Big( \frac{1}{\sqrt{T}} \Big) ,
\end{equation}
by a simple induction argument (see \text{e.g.} the induction argument in the proof of Lemma~6 in \citeauthor{CenLam2025_KronProd}, \citeyear{CenLam2025_KronProd}). Then, using Lemma~\ref{lemma: inequality_sandwich}(b),
\begin{align}
    &\;\quad
    \|\cII_6\|_F^2 = \Bigg\| \frac{1}{Tp_m} \sum_{t=1}^T \mat_k(\cE_t^{(m)}) \wh{\bA}_{m,\text{-}k,\perp} \wh{\bA}_{m,\text{-}k,\perp}^\trans \bA_{m,\text{-}k} \mat_k(\cG_t^{(m)})^\trans \bA_{m,k}^\trans \Bigg\|_F^2 \notag \\
    &\lesssim
    \Bigg\| \frac{1}{Tp_m} \sum_{t=1}^T \mat_k(\cE_t^{(m)}) \notag \\
    &\;\quad
    \cdot \big( \wh{\bA}_{m,\text{-}k,\perp} \wh{\bA}_{m,\text{-}k,\perp}^\trans - \bA_{m,\text{-}k,\perp} \bA_{m,\text{-}k,\perp}^\trans \big) \bA_{m,\text{-}k} \mat_k(\cG_t^{(m)})^\trans \bA_{m,k}^\trans \Bigg\|_F^2 \notag \\
    &\leq
    \frac{1}{p_m^2} \cdot \|\bA_{m,k}\|_F^2 \cdot \Big\| \wh{\bA}_{m,\text{-}k,\perp} \wh{\bA}_{m,\text{-}k,\perp}^\trans - \bA_{m,\text{-}k,\perp} \bA_{m,\text{-}k,\perp}^\trans \Big\|_F^2 \cdot p_{m,\text{-}k} \notag \\
    &\;\quad
    \cdot \sum_{i=1}^{p_{m,k}} \sum_{j=1}^{p_{m,\text{-}k}} \Bigg\| \frac{1}{T} \sum_{t=1}^T \mat_k(\cE_t^{(m)})_{i,j} \frac{\bA_{m,\text{-}k}}{\sqrt{p_{m,\text{-}k}}} \mat_k(\cG_t^{(m)})^\trans \Bigg\|_F^2 \notag \\
    &=
    \frac{1}{p_m^2} \cdot \|\bA_{m,k}\|_F^2 \cdot \frac{p_{m,\text{-}k}}{T} \cdot \cO_P(p_{m})
    = \cO_P\Big( \frac{1}{T} \Big) ,
    \label{eqn: cII_6}
\end{align}
where the second last equality used \eqref{eqn: hat_A_perp_minus_A_kron}, Assumption~\ref{ass: dependence_thread}(c) and similar arguments in \eqref{eqn: cII_2}, and the last used Assumption~\ref{ass: loadings}(b). Similarly, we also have
\begin{align}
    &\;\quad
    \|\cII_8\|_F^2 = \cO_P\Big( \frac{1}{T p_{m,\text{-}k}} \Big) ,
    \label{eqn: cII_8}
\end{align}

For $\cII_9$, by Lemmas~\ref{lemma: inequality_sandwich}(a) and \ref{lemma: A_perp_rate}, we have
\begin{align}
    &\;\quad
    \|\cII_9\|_F^2 \notag \\
    &=
    \Bigg\| \frac{1}{Tp_m} \sum_{t=1}^T \bA_{m,k} \mat_k(\cG_t^{(m)}) \bA_{m,\text{-}k}^\trans \wh{\bA}_{m,\text{-}k,\perp} \wh{\bA}_{m,\text{-}k,\perp}^\trans \bA_{m,\text{-}k} \mat_k(\cG_t^{(m)})^\trans \bA_{m,k}^\trans \Bigg\|_F^2 \notag \\
    &\leq
    \Big\| \wh{\bA}_{m,\text{-}k,\perp} \wh{\bA}_{m,\text{-}k,\perp}^\trans - \bA_{m,\text{-}k,\perp} \bA_{m,\text{-}k,\perp}^\trans \Big\|_F^2 \notag \\
    &\;\quad
    \cdot \Bigg\| \frac{1}{T p_{m,k}} \sum_{t=1}^T \bA_{m,k} \mat_k(\cG_t^{(m)}) \bA_{m,\text{-}k}^\trans \bA_{m,\text{-}k} \mat_k(\cG_t^{(m)})^\trans \bA_{m,k}^\trans \Bigg\|_F^2
    = \cO_P\Big( \frac{1}{T} \Big) ,
    \label{eqn: cII_9}
\end{align}
where the last equality used Assumption~\ref{ass: core_factor} and similar arguments on $\cII_1$. Similarly for $\cII_5$, using also Lemma~\ref{lemma: inequality_sandwich}(a) and noting that $\big\| \wh{\bA}_{m,\text{-}k,\perp} \wh{\bA}_{m,\text{-}k,\perp}^\trans \big\| =\cO(1)$,
\begin{align}
    &\;\quad
    \|\cII_5\|_F^2 = \Bigg\| \frac{1}{Tp_m} \sum_{t=1}^T \mat_k(\cE_t^{(m)}) \wh{\bA}_{m,\text{-}k,\perp} \wh{\bA}_{m,\text{-}k,\perp}^\trans \mat_k(\cE_t^{(m)})^\trans \Bigg\|_F^2 \notag \\
    &\lesssim
    \Bigg\| \frac{1}{Tp_m} \sum_{t=1}^T \mat_k(\cE_t^{(m)}) \mat_k(\cE_t^{(m)})^\trans \Bigg\|_F^2
    = \cO_P\Big( \frac{1}{T p_{m,\text{-}k}} + \frac{1}{p_{m,k}} \Big) ,
    \label{eqn: cII_5}
\end{align}
where the last equality used Assumption~\ref{ass: noise} and the proof of Lemma~3 in \cite{barigozzi2022statistical}.

Finally, combining all \eqref{eqn: cII_2}--\eqref{eqn: cII_5}, then the statement in the lemma is concluded by \eqref{eqn: decomp_Sigma_B} and applying Weyl's inequality iteratively. This ends the proof of the lemma.
\end{proof}

\begin{lemma}\label{lemma: eigenvalue_tilde_Sigma_B}
Let Assumptions~\ref{ass: identification}, \ref{ass: loadings}, \ref{ass: core_factor}, \ref{ass: noise}, \ref{ass: dependence_moment}, and \ref{ass: dependence_thread} hold.
Given any $m\in[M]$ with $K_m= 1$, recall $\wt{\bD}_{m,1}$ from the statement of Theorem~\ref{thm: local_loading}.
Then the eigenvalues of $\wt{\bD}_{m,1}$ satisfy $\lambda_i(\wt{\bD}_{m,1}) = \lambda_i(\bSigma_{F,m,1}) + o_P(1)$ for any $i\leq u_{m,1}$.
\end{lemma}

\begin{proof}[Proof of Lemma~\ref{lemma: eigenvalue_tilde_Sigma_B}]
By the definition in \eqref{eqn: def_tilde_Sigma_B_m1}, we decompose
\begin{align}
    &\;\quad
    \wt\bSigma_{B,m,1} = \frac{1}{Tp_{m,1}} \sum_{t=1}^T \wh{\bA}_{m,1,\perp} \wh{\bA}_{m,1,\perp}^\trans \mat_1(\cX_t^{(m)}) \mat_1(\cX_t^{(m)})^\trans \wh{\bA}_{m,1,\perp} \wh{\bA}_{m,1,\perp}^\trans \notag \\
    &=
    \frac{1}{Tp_{m,1}} \sum_{t=1}^T \wh{\bA}_{m,1,\perp} \wh{\bA}_{m,1,\perp}^\trans \Big\{ \bA_{m,1} \mat_1(\cG_t^{(m)}) + \bB_{m,1} \mat_1(\cF_t^{(m)}) + \mat_k(\cE_t^{(m)}) \Big\} \notag \\
    &\;\quad
    \cdot \Big\{ \bA_{m,1} \mat_1(\cG_t^{(m)}) + \bB_{m,1} \mat_1(\cF_t^{(m)}) + \mat_k(\cE_t^{(m)}) \Big\}^\trans \wh{\bA}_{m,1,\perp} \wh{\bA}_{m,1,\perp}^\trans \notag \\
    &=
    \frac{1}{Tp_{m,1}} \sum_{t=1}^T \wh{\bA}_{m,1,\perp} \wh{\bA}_{m,1,\perp}^\trans \bB_{m,1} \mat_1(\cF_t^{(m)}) \mat_1(\cF_t^{(m)})^\trans \bB_{m,1}^\trans \wh{\bA}_{m,1,\perp} \wh{\bA}_{m,1,\perp}^\trans \notag \\
    &\;\quad
    + \frac{1}{Tp_{m,1}} \sum_{t=1}^T \wh{\bA}_{m,1,\perp} \wh{\bA}_{m,1,\perp}^\trans \bB_{m,1} \mat_1(\cF_t^{(m)}) \mat_1(\cG_t^{(m)})^\trans \bA_{m,1}^\trans \wh{\bA}_{m,1,\perp} \wh{\bA}_{m,1,\perp}^\trans \notag \\
    &\;\quad
    + \frac{1}{Tp_{m,1}} \sum_{t=1}^T \wh{\bA}_{m,1,\perp} \wh{\bA}_{m,1,\perp}^\trans \bB_{m,1} \mat_1(\cF_t^{(m)}) \mat_1(\cE_t^{(m)})^\trans \wh{\bA}_{m,1,\perp} \wh{\bA}_{m,1,\perp}^\trans \notag \\
    &\;\quad
    + \frac{1}{Tp_{m,1}} \sum_{t=1}^T \wh{\bA}_{m,1,\perp} \wh{\bA}_{m,1,\perp}^\trans \bA_{m,1} \mat_1(\cG_t^{(m)}) \mat_1(\cG_t^{(m)})^\trans \bA_{m,1}^\trans \wh{\bA}_{m,1,\perp} \wh{\bA}_{m,1,\perp}^\trans \notag \\
    &\;\quad
    + \frac{1}{Tp_{m,1}} \sum_{t=1}^T \wh{\bA}_{m,1,\perp} \wh{\bA}_{m,1,\perp}^\trans \bA_{m,1} \mat_1(\cG_t^{(m)}) \mat_1(\cF_t^{(m)})^\trans \bB_{m,1}^\trans \wh{\bA}_{m,1,\perp} \wh{\bA}_{m,1,\perp}^\trans \notag \\
    &\;\quad
    + \frac{1}{Tp_{m,1}} \sum_{t=1}^T \wh{\bA}_{m,1,\perp} \wh{\bA}_{m,1,\perp}^\trans \bA_{m,1} \mat_1(\cG_t^{(m)}) \mat_1(\cE_t^{(m)})^\trans \wh{\bA}_{m,1,\perp} \wh{\bA}_{m,1,\perp}^\trans \notag \\
    &\;\quad
    + \frac{1}{Tp_{m,1}} \sum_{t=1}^T \wh{\bA}_{m,1,\perp} \wh{\bA}_{m,1,\perp}^\trans \mat_1(\cE_t^{(m)}) \mat_1(\cG_t^{(m)})^\trans \bA_{m,1}^\trans \wh{\bA}_{m,1,\perp} \wh{\bA}_{m,1,\perp}^\trans \notag \\
    &\;\quad
    + \frac{1}{Tp_{m,1}} \sum_{t=1}^T \wh{\bA}_{m,1,\perp} \wh{\bA}_{m,1,\perp}^\trans \mat_1(\cE_t^{(m)}) \mat_1(\cF_t^{(m)})^\trans \bB_{m,1}^\trans \wh{\bA}_{m,1,\perp} \wh{\bA}_{m,1,\perp}^\trans \notag \\
    &\;\quad
    + \frac{1}{Tp_{m,1}} \sum_{t=1}^T \wh{\bA}_{m,1,\perp} \wh{\bA}_{m,1,\perp}^\trans \mat_1(\cE_t^{(m)}) \mat_1(\cE_t^{(m)})^\trans \wh{\bA}_{m,1,\perp} \wh{\bA}_{m,1,\perp}^\trans \notag \\
    &=: \wt\cII_1 + \wt\cII_2 + \wt\cII_3 + \wt\cII_4 + \wt\cII_5 + \wt\cII_6 + \wt\cII_7 + \wt\cII_8 + \wt\cII_9 .
    \label{eqn: decomp_tilde_Sigma_B}
\end{align}

By definition, the rank of $\wt\cII_1$ is $u_{m,1}$. As the nonzero eigenvalues of some matrix product $\bX\bY$ are the same as those of $\bY\bX$ (with multiplicity), from Assumption~\ref{ass: core_factor}(b) we have $\lambda_i(\wt\cII_1) \asymp \lambda_i(\bSigma_{F,m,1})$. Then similar to the proof of Lemma~\ref{lemma: eigenvalue_Sigma_B}, it remains to show the (squared) Frobenius norm of $\wt\cII_2$--$\wt\cII_9$ are $o_P(1)$.

To this end, consider $\wt\cII_2$. By the triangle inequality, we have
\begin{align}
    &\;\quad
    \|\wt\cII_2\|_F^2 \notag \\
    &=
    \Bigg\| \frac{1}{T p_{m,1}} \sum_{t=1}^T \wh{\bA}_{m,1,\perp} \wh{\bA}_{m,1,\perp}^\trans \bB_{m,1} \mat_1(\cF_t^{(m)}) \mat_1(\cG_t^{(m)})^\trans \bA_{m,1}^\trans \wh{\bA}_{m,1,\perp} \wh{\bA}_{m,1,\perp}^\trans \Bigg\|_F^2 \notag \\
    &\leq
    \frac{1}{T p_{m,1}^2} \big\|\bB_{m,1} \big\|_F^2 \cdot \sum_{j=1}^{p_{m,1}} \Bigg\| \frac{1}{\sqrt{T}} \sum_{t=1}^T \mat_1(\cF_t^{(m)}) \mat_1(\cG_t^{(m)})^\trans (\bA_{m,1})_{j\cdot} \Bigg\|_F^2 
    = \cO_P\Big( \frac{1}{T} \Big)
    \label{eqn: tilde_cII_2}
\end{align}
where the last equality used Assumptions~\ref{ass: loadings}(b) and \ref{ass: dependence_thread}(b). Similarly, we have
\begin{align}
    &\;\quad
    \|\wt\cII_3\|_F^2, \;
    \|\wt\cII_5\|_F^2, \;
    \|\wt\cII_8\|_F^2
    = \cO_P\Big( \frac{1}{T} \Big) .
    \label{eqn: tilde_cII_3_5_8}
\end{align}

With Lemma~\ref{lemma: A_perp_rate}, Assumptions~\ref{ass: core_factor}(a) and \ref{ass: dependence_thread}(c), we also have 
\begin{align}
    &\;\quad
    \|\wt\cII_4\|_F^2, \;
    \|\wt\cII_6\|_F^2, \;
    \|\wt\cII_7\|_F^2
    = \cO_P\Big( \frac{1}{T} \Big) .
    \label{eqn: tilde_cII_4_6_7}
\end{align}

Lastly, using the same argument in \eqref{eqn: cII_5} (except that $p_{m,\text{-}k}=1$ here) \begin{align}
    &\;\quad
    \|\wt\cII_9\|_F^2 = \cO_P\Big( \frac{1}{T} + \frac{1}{p_{m,1}} \Big) ,
    \label{eqn: tilde_cII_9}
\end{align}
hence concluding the proof of this lemma.
\end{proof}

\begin{lemma}\label{lemma: eigenvalue_Sigma_aux}
Under Assumption~\ref{ass: loadings}(b), it holds that Assumption~\ref{ass: core_factor}(c) is equivalent to Assumption~\ref{ass: core_factor_part_c_general} in Section~\ref{subsec: discussion_core_fac_part_c}.
\end{lemma}

\begin{proof}[Proof of Lemma~\ref{lemma: eigenvalue_Sigma_aux}]
Fix any $m\in[M]$, $k\in[K_m]$, and $n\neq m$. Under Assumption~\ref{ass: loadings}(b), recall the notations in Section~\ref{subsec: discussion_core_fac_part_c} that
\[
\wt\bA_n =
\otimes_{h=1}^{K_n} \frac{\bA_{n,h}}{ \sqrt{p_{n,h}} } ,
\]
so that $\wt\bA_n^\trans \wt\bA_n = \bI_{r_n}$.
Then define $\bD_{m,n,k} := \E\big\{ \mat_k(\cG_t^{(m)}) \otimes \vec(\cG_t^{(n)})^\trans \big\}$
and
\begin{align*}
    \wt\bD_{m,n,k} &:=
    \E\Big[ \mat_k(\cG_t^{(m)}) \otimes \Big\{ \vec(\cG_t^{(n)})^\trans \big( \otimes_{h=1}^{K_n} \bA_{n,h} \big)^\trans \Big\} \Big] \\
    &=
    \sqrt{p_n} \, \bD_{m,n,k} \left( \bI_{r_{m,\text{-}k}} \otimes \wt\bA_n^\trans \right) ,
\end{align*}
where the last equality is by simple algebra in Kronecker products.
Since it holds under Assumption~\ref{ass: loadings}(b) that 
$(\bI_{r_{m,\text{-}k}} \otimes \wt\bA_n^\trans) (\bI_{r_{m,\text{-}k}} \otimes \wt\bA_n) = \bI_{r_{m,\text{-}k} \cdot r_n}$,
we have $\wt\bD_{m,n,k} \wt\bD_{m,n,k}^\trans = p_n \bD_{m,n,k} \bD_{m,n,k}^\trans$.
Hence for every $j\in [r_{m,k}]$, $\sigma_j(\bD_{m,n,k}) \asymp 1$ is equivalent to $\sigma_j(\wt\bD_{m,n,k}/ \sqrt{p_n}) \asymp 1$.
The former is Assumption~\ref{ass: core_factor}(c), and the latter is Assumption~\ref{ass: core_factor_part_c_general} after normalization of the loading matrices. Thus the two conditions are equivalent, which completes this proof.
\end{proof}

\subsection{Proof of theorems}

\begin{proof}[Proof of Theorem~\ref{prop: identification}]
Fix $m\in[M]$. 
Under Assumption~\ref{ass: identification}(b), local factors in different threads are uncorrelated.
Hence for any $n\neq m$ and any multi-index $\bi$, we have
\begin{align*}
    \bOmega_{k,n,\bi}^{(m)} &= 
    \frac{1}{T} \sum_{t=1}^T 
    \E\Big\{ \mat_k(\cX_t^{(m)}) \cdot (\cX_t^{(n)})_{\bi} \Big\} \\
    &= \bA_{m,k} \E\Bigg\{ 
    \frac{1}{T} \sum_{t=1}^T 
    \mat_k(\cG_t^{(m)}) \bA_{m,\text{-}k}^\trans \Big( \cG_t^{(n)} \times_{h=1}^{K_n} \bA_{n,h} \Big)_{\bi} \Bigg\} ,
\end{align*}
where the first equality is by definition in \eqref{eqn: def_Omega_pop}, and the second used Assumption~\ref{ass: identification}(b) (and Assumption~\ref{ass: dependence_moment}(a), if $\cE_t^{(m)}$ is nonzero for noisy observations; see the discussion below Theorem~\ref{prop: identification}).
Denote the expectation term above as $\bM_{m,n,\bi}$. Then by definition in \eqref{eqn: def_Sigma_pop}, we have
\begin{align*}
    \bSigma_{m,k} = \bA_{m,k} \left(
    \sum_{\substack{n=1\\ n\neq m}}^M \sum_{\bi \in [p_{n,1}] \times \cdots \times[p_{n,K_n}]} \bM_{m,n,\bi} \bM_{m,n,\bi}^\trans 
    \right)
    \bA_{m,k}^\trans
    =: 
    \bA_{m,k} \bC_{m,k} \bA_{m,k}^\trans .
\end{align*}
By condition~(a) in the statement of Theorem~\ref{prop: identification}, $\bA_{m,k}$ has full column rank and hence $\col(\bSigma_{m,k}) = \col(\bA_{m,k})$.
Similarly, the two admissible representations share the same global loading space, i.e.,
\begin{align*}
    \col( \otimes_{k=1}^{K_m} \wt\bA_{m,k} ) = \col( \otimes_{k=1}^{K_m} \bA_{m,k} ) ,
\end{align*}
which we denote by $\cL_m$;
also, let $\bP_m$ be its orthogonal projection matrix.

From condition~(b) in the statement of Theorem~\ref{prop: identification}, we have
\[
\vec(\cX_{G,t}^{(m)}) \in \cL_m,
\;
\text{\ and \ }
\vec(\cX_{F,t}^{(m)}) \in \cL_{m,\perp} ,
\]
with $\cL_{m,\perp}$ denoting the orthogonal complement of $\cL_{m}$;
and the above hold similarly for $\vec(\wt\cX_{G,t}^{(m)})$ and $\vec(\wt\cX_{F,t}^{(m)})$. 
Therefore, we have
\[
\bP_m \vec(\cX_{G,t}^{(m)}) = \vec(\cX_{G,t}^{(m)}),
\;
\text{\ and \ }
\bP_m \vec(\cX_{F,t}^{(m)}) = 0 .
\]
Since $\vec(\cX_{G,t}^{(m)}) + \vec(\cX_{F,t}^{(m)}) = \vec(\wt{\cX}_{G,t}^{(m)}) +  \vec(\wt{\cX}_{F,t}^{(m)})$,
left-multiplying $\bP_m$ to both sides gives $\vec(\cX_{G,t}^{(m)}) = \vec(\wt{\cX}_{G,t}^{(m)})$.
This indicates that the global and thus the local components are uniquely determined,
which completes the proof.
\end{proof}

\begin{proof}[Proof of Theorem~\ref{thm: global_loading}]
By applying Lemma~3 in \cite{Lametal2011} and notice the matrix $P$ within is of fixed rank, we have
\begin{align*}
    \frac{1}{\sqrt{p_{m,k}}} \Big\| \wh{\bQ}_{m,k} -\bA_{m,k} \Big\|_F \leq \frac{8}{\lambda_{r_{m,k}}\big( \bSigma_{m,k} \big)} \Big\| \wh\bSigma_{m,k} -\bSigma_{m,k} \Big\|_F
    = \cO_P\Big( \frac{1}{\sqrt{T}} \Big) ,
\end{align*}
where $\wh{\bQ}_{m,k} \in \R^{p_{m,k} \times r_{m,k}}$ is $\sqrt{p_{m,k}}$ times some orthonormal basis for the invariant subspace of the $r_{m,k}$ leading eigenvectors of $\wh\bSigma_{m,k}$,
and the last equality used Lemmas~\ref{lemma: rate_Sigma} and \ref{lemma: eigenvalue_Sigma}. 
By the construction of $\wh{\bA}_{m,k}$, there exists an exactly orthogonal matrix $\bH_{A,m,k}$ such that $\wh{\bA}_{m,k} =\wh{\bQ}_{m,k} \bH_{A,m,k}$,
together with $\big\| \wh{\bQ}_{m,k} -\bA_{m,k} \big\|_F = \big\| \wh{\bQ}_{m,k} \bH_{A,m,k} -\bA_{m,k} \bH_{A,m,k} \big\|_F$,
we conclude the proof of the theorem.
\end{proof}

\begin{proof}[Proof of Theorem~\ref{thm: local_loading}]
By the definition of $\wh{\bB}_{m,k}$, we have
\[
\wh{\bB}_{m,k} = \wh\bSigma_{B,m,k} \wh{\bB}_{m,k} \wh{\bD}_{B,m,k}^{-1} .
\]
First consider $K_m> 1$, and note that from the decomposition in \eqref{eqn: decomp_Sigma_B},
recall from the statement of Theorem~\ref{thm: local_loading} that
{\small
\[
\wh{\bH}_{m,k} = \frac{1}{Tp_m} \sum_{t=1}^T \mat_k(\cF_t^{(m)}) \bB_{m,\text{-}k}^\trans \wh{\bA}_{m,\text{-}k,\perp} \wh{\bA}_{m,\text{-}k,\perp}^\trans \bB_{m,\text{-}k} \mat_k(\cF_t^{(m)})^\trans \bB_{m,k}^\trans \wh{\bB}_{m,k} \wh{\bD}_{B,m,k}^{-1} ,
\]
}
so that with the notation in \eqref{eqn: decomp_Sigma_B}, by Cauchy--Schwarz inequality, we have
\begin{align}
    &\;\quad
    \frac{1}{p_{m,k}} \Big\| \wh{\bB}_{m,k} - \bB_{m,k} \wh{\bH}_{m,k} \Big\|_F^2 = \frac{1}{p_{m,k}} \Big\| \Big( \sum_{i=2}^9 \cII_i \Big) \wh{\bB}_{m,k} \wh{\bD}_{B,m,k}^{-1} \Big\|_F^2 \notag \\
    &\lesssim
    \frac{1}{p_{m,k}} \Big\| \cII_5 \wh{\bB}_{m,k} \wh{\bD}_{B,m,k}^{-1} \Big\|_F^2 + \frac{1}{p_{m,k}} \Big\| \cII_9 \wh{\bB}_{m,k} \wh{\bD}_{B,m,k}^{-1} \Big\|_F^2 \notag \\
    &\;\quad
    + \frac{1}{p_{m,k}}  \cdot \big\| \wh{\bD}_{B,m,k}^{-1} \big\|_F^2 \cdot \sum_{\substack{i=2\\ i\neq 5}}^8 \Big\| \cII_i \wh{\bB}_{m,k} \Big\|_F^2 \notag \\
    &=
    \frac{1}{p_{m,k}} \Big\| \cII_5 \wh{\bB}_{m,k} \Big\|_F^2 + \frac{1}{p_{m,k}} \Big\| \cII_9 \wh{\bB}_{m,k} \Big\|_F^2 + \cO_P\Big( \frac{1}{T p_{m,\text{-}k}} \Big) 
    \notag \\
    &\;\quad
    + o_P\Big( \frac{1}{p_{m,k}} \Big\| \wh{\bB}_{m,k} - \bB_{m,k} \wh{\bH}_{m,k} \Big\|_F^2 \Big) \notag \\
    &\lesssim
    \frac{1}{p_{m,k}} \Big\| \cII_5 \bB_{m,k} \Big\|_F^2 + \frac{1}{p_{m,k}} \Big\| \cII_9 \bB_{m,k} \Big\|_F^2 + \cO_P\Big( \frac{1}{T p_{m,\text{-}k}} \Big) 
    \notag \\
    &\;\quad
    + o_P\Big( \frac{1}{p_{m,k}} \Big\| \wh{\bB}_{m,k} - \bB_{m,k} \wh{\bH}_{m,k} \Big\|_F^2 \Big) ,
    \label{eqn: hat_B_minus_BH}
\end{align}
where: (i) the last equality used Lemma~\ref{lemma: eigenvalue_Sigma_B}, the rates in its proof, and the exact half-orthogonality\footnote{Under asymptotically half-orthogonality in Assumption~\ref{ass: loadings_part_b_asymp}, an extra term $\cO_P\{ 1/(T p_{m,k}) \}$ would be incurred since the last equality above, due to $\big\|\bA_{m,k}^\trans \bB_{m,k} \big\|_F =\cO_P(\sqrt{p_{m,k}})$.}
that $\bA_{m,k}^\trans \bB_{m,k} = \zero$
by Assumption~\ref{ass: loadings}(b) and hence, for example, for $\cII_6 \wh{\bB}_{m,k}$, similar to \eqref{eqn: cII_6}:
\begin{align*}
    \Big\| \cII_6 \wh{\bB}_{m,k} \Big\|_F^2 &\lesssim
    \frac{1}{p_m^2} \cdot \big\|\bA_{m,k}^\trans \bB_{m,k} \big\|_F^2 \cdot \frac{p_{m,\text{-}k}}{T} \cdot \cO_P(p_{m})
    + o_P\Big( \Big\| \wh{\bB}_{m,k} - \bB_{m,k} \wh{\bH}_{m,k} \Big\|_F^2 \Big) \\
    &=
    o_P\Big( \Big\| \wh{\bB}_{m,k} - \bB_{m,k} \wh{\bH}_{m,k} \Big\|_F^2 \Big) ;
\end{align*}
and (ii) the last line used \eqref{eqn: cII_9}, \eqref{eqn: cII_5}, and the fact that $\big\| \wh{\bH}_{m,k} \big\|_F^2 =\cO_P(1)$ by the convergence of $\cII_1$ in the proof of Lemma~\ref{lemma: eigenvalue_Sigma_B}.

Next, from Lemma~\ref{lemma: A_perp_rate} and its proof, we can write
\[
{\bA}_{m,\text{-}k,\perp} {\bA}_{m,\text{-}k,\perp}^\trans = \bI - \otimes_{\kappa \in [K_m]\setminus \{k\}} (\bQ_{m,\kappa} \bQ_{m,\kappa}^\trans) ,
\]
where each matrix $\bQ_{m,\kappa} \in \R^{p_{m,\kappa} \times r_{m,\kappa}}$ is half-orthogonal such that $\bQ_{m,\kappa}^\trans \bQ_{m,\kappa} = \bI_{r_{m,\kappa}}$.
Then by Lemma~\ref{lemma: inequality_sandwich}, 
Lemma~\ref{lemma: A_perp_rate},
Assumption~\ref{ass: loadings},
and \eqref{eqn: cII_5}, we have
\begin{align}
    &\;\quad
    \frac{1}{p_{m,k}} \Big\| \cII_5 \bB_{m,k} \Big\|_F^2 \notag \\
    &=
    \frac{1}{p_{m,k}} \Bigg\| \frac{1}{Tp_m} \sum_{t=1}^T \mat_k(\cE_t^{(m)}) \wh{\bA}_{m,\text{-}k,\perp} \wh{\bA}_{m,\text{-}k,\perp}^\trans \mat_k(\cE_t^{(m)})^\trans \bB_{m,k} \Bigg\|_F^2 \notag \\
    &\lesssim
    \frac{1}{p_{m,k}} \Bigg\| \frac{1}{Tp_m} \sum_{t=1}^T \mat_k(\cE_t^{(m)}) 
    \big( \wh{\bA}_{m,\text{-}k,\perp} \wh{\bA}_{m,\text{-}k,\perp}^\trans - {\bA}_{m,\text{-}k,\perp} {\bA}_{m,\text{-}k,\perp}^\trans \big)
    \mat_k(\cE_t^{(m)})^\trans \bB_{m,k} \Bigg\|_F^2 \notag \\
    &\;\quad
    + \frac{1}{p_{m,k}} \Bigg\| \frac{1}{Tp_m} \sum_{t=1}^T \mat_k(\cE_t^{(m)}) 
    {\bA}_{m,\text{-}k,\perp} {\bA}_{m,\text{-}k,\perp}^\trans 
    \mat_k(\cE_t^{(m)})^\trans \bB_{m,k} \Bigg\|_F^2 \notag \\
    &=
    \frac{1}{p_{m,k}} \Bigg\| \frac{1}{Tp_m} \sum_{t=1}^T \mat_k(\cE_t^{(m)}) 
    \big( \wh{\bA}_{m,\text{-}k,\perp} \wh{\bA}_{m,\text{-}k,\perp}^\trans - {\bA}_{m,\text{-}k,\perp} {\bA}_{m,\text{-}k,\perp}^\trans \big)
    \mat_k(\cE_t^{(m)})^\trans \bB_{m,k} \Bigg\|_F^2 \notag \\
    &\;\quad
    + \frac{1}{p_{m,k}} \Bigg\| \frac{1}{Tp_m} \sum_{t=1}^T \mat_k(\cE_t^{(m)}) 
    \left\{ \bI - \otimes_{\kappa \in [K_m]\setminus \{k\}} (\bQ_{m,\kappa} \bQ_{m,\kappa}^\trans) \right\} 
    \mat_k(\cE_t^{(m)})^\trans \bB_{m,k} \Bigg\|_F^2 \notag \\
    &\lesssim
    \frac{\|\bB_{m,k} \|^2}{p_{m,k}} \Bigg\| \frac{1}{Tp_m} \sum_{t=1}^T \mat_k(\cE_t^{(m)}) 
    \big( \wh{\bA}_{m,\text{-}k,\perp} \wh{\bA}_{m,\text{-}k,\perp}^\trans - {\bA}_{m,\text{-}k,\perp} {\bA}_{m,\text{-}k,\perp}^\trans \big)
    \mat_k(\cE_t^{(m)})^\trans \Bigg\|_F^2 \notag \\
    &\;\quad
    + \frac{1}{p_{m,k}} \Bigg\| \frac{1}{Tp_m} \sum_{t=1}^T \mat_k(\cE_t^{(m)})  
    \mat_k(\cE_t^{(m)})^\trans \bB_{m,k} \Bigg\|_F^2 \notag \\
    &\;\quad
    + \frac{1}{p_{m,k}} \Bigg\| \frac{1}{Tp_m} \sum_{t=1}^T \mat_k(\cE_t^{(m)}) 
    \left\{ \otimes_{\kappa \in [K_m]\setminus \{k\}} (\bQ_{m,\kappa} \bQ_{m,\kappa}^\trans) \right\} 
    \mat_k(\cE_t^{(m)})^\trans \bB_{m,k} \Bigg\|_F^2 \notag \\
    &\lesssim
    \big\| \wh{\bA}_{m,\text{-}k,\perp} \wh{\bA}_{m,\text{-}k,\perp}^\trans - {\bA}_{m,\text{-}k,\perp} {\bA}_{m,\text{-}k,\perp}^\trans \big\|^2 \cdot
    \Bigg\| \frac{1}{Tp_m} \sum_{t=1}^T \mat_k(\cE_t^{(m)}) 
    \mat_k(\cE_t^{(m)})^\trans \Bigg\|_F^2 \notag \\
    &\;\quad
    + \frac{1}{p_{m,k}} \Bigg\| \frac{1}{Tp_m} \sum_{t=1}^T \mat_k(\cE_t^{(m)})  
    \mat_k(\cE_t^{(m)})^\trans \bB_{m,k} \Bigg\|_F^2 \notag \\
    &\;\quad
    + \frac{\left\| \otimes_{\kappa \in [K_m]\setminus \{k\}} (\bQ_{m,\kappa} \bQ_{m,\kappa}^\trans) \right\|_F^2 }{T^2 p_m^2 p_{m,k}} 
    \sum_{i=1}^{p_{m,k}} 
    \sum_{j=1}^{p_{m,\text{-}k}} 
    \Bigg\| \sum_{t=1}^T \mat_k(\cE_t^{(m)})_{ij} 
    \mat_k(\cE_t^{(m)})^\trans \bB_{m,k} \Bigg\|_F^2 \notag \\
    &=
    \cO_P\Bigg( \frac{1}{T^2 p_{m,\text{-}k}} + \frac{1}{T p_{m,k}} \Bigg)
    + \cO_P\Bigg( \frac{1}{p_{m,k}^2} + \frac{1}{T p_m}
    \Bigg) \notag \\
    &\;\quad
    + \frac{1}{T^2 p_m^2 p_{m,k}} 
    \sum_{i=1}^{p_{m,k}} 
    \sum_{j=1}^{p_{m,\text{-}k}} 
    \Bigg\| \sum_{t=1}^T \mat_k(\cE_t^{(m)})_{ij} 
    \mat_k(\cE_t^{(m)})^\trans \bB_{m,k} \Bigg\|_F^2 \notag \\
    &=
    \cO_P\Bigg( \frac{1}{T^2 p_{m,\text{-}k}} + \frac{1}{T p_{m,k}} 
    + \frac{1}{p_{m,k}^2} \Bigg) \notag \\
    &\;\quad
    + \frac{1}{T^2 p_m^2 p_{m,k}} \sum_{i=1}^{p_{m,k}}
    \sum_{h, j=1}^{p_{m,\text{-}k}}
    \Bigg\| \sum_{t=1}^T
    \sum_{z=1}^{p_{m,k}}
    \mat_k(\cE_t^{(m)})_{ij} 
    (\bB_{m,k})_{z \cdot} \mat_k(\cE_t^{(m)})_{zh}  
    \Bigg\|_F^2 \notag \\
    &=
    \cO_P\Bigg( \frac{1}{T^2 p_{m,\text{-}k}} + \frac{1}{T p_{m,k}} 
    + \frac{1}{p_{m,k}^2} \Bigg) .
    \label{eqn: cII_5_Bmk}
\end{align}

On the other hand, by triangle inequality, we have
\begin{align*}
    &\;\quad
    \frac{1}{p_{m,k}} \Big\| \cII_9 \bB_{m,k} \Big\|_F^2 \\
    &=
    \frac{1}{p_{m,k}} \Bigg\| \frac{1}{Tp_m} \sum_{t=1}^T \bA_{m,k} \mat_k(\cG_t^{(m)}) \\
    &\;\quad
    \cdot \bA_{m,\text{-}k}^\trans \wh{\bA}_{m,\text{-}k,\perp} \wh{\bA}_{m,\text{-}k,\perp}^\trans \bA_{m,\text{-}k} \mat_k(\cG_t^{(m)})^\trans \bA_{m,k}^\trans \bB_{m,k} \Bigg\|_F^2 \\
    &\leq
    \frac{1}{p_m^2 p_{m,k}} \big\|\bA_{m,k}\big\|_F^2 \cdot \Big\| \wh{\bA}_{m,\text{-}k,\perp} \wh{\bA}_{m,\text{-}k,\perp}^\trans - \bA_{m,\text{-}k,\perp} \bA_{m,\text{-}k,\perp}^\trans \Big\|_F^2 \\
    &\;\quad
    \cdot \Bigg\| \frac{1}{T} \sum_{t=1}^T \mat_k(\cG_t^{(m)}) \bA_{m,\text{-}k}^\trans \bA_{m,\text{-}k} \mat_k(\cG_t^{(m)})^\trans \Bigg\|_F^2 \cdot \big\| \bA_{m,k}^\trans \bB_{m,k} \big\|_F^2 \\
    &\lesssim
    \frac{1}{T} \cdot \frac{1}{p_{m,k}^2} \big\| \bA_{m,k}^\trans \bB_{m,k} \big\|_F^2
    = \frac{1}{T} \cdot \Bigg\| \frac{1}{p_{m,k}} \bA_{m,k}^\trans \bB_{m,k} \Bigg\|_F^2
    = 0,
\end{align*}
where the last line used Assumption~\ref{ass: loadings}(b)\footnote{If Assumption~\ref{ass: loadings}(b) is replaced by Assumption~\ref{ass: loadings_part_b_asymp} and no reformulation is performed, the last term above would become $\cO_P\{ 1/(T p_{m,k}) \}$.}, 
Lemma~\ref{lemma: A_perp_rate}, and similar arguments in \eqref{eqn: cII_9}.

Finally, noting that the last term in \eqref{eqn: hat_B_minus_BH} can be omitted, we conclude that
\[
\frac{1}{p_{m,k}} \Big\| \wh{\bB}_{m,k} - \bB_{m,k} \wh{\bH}_{m,k} \Big\|_F^2 = \cO_P\Big( \frac{1}{p_{m,k}^2} + \frac{1}{T p_{m,k}} + \frac{1}{T p_{m,\text{-}k}} \Big) .
\]

It remains to show $\wh{\bH}_{m,k} \wh{\bH}_{m,k}^\trans \xrightarrow{P} \bI$, which is obvious by
\begin{align*}
    \bI &= \frac{1}{p_{m,k}} \wh{\bB}_{m,k}^\trans \wh{\bB}_{m,k} = \frac{1}{p_{m,k}} \wh{\bB}_{m,k}^\trans (\wh{\bB}_{m,k} - \bB_{m,k} \wh{\bH}_{m,k}) + \frac{1}{p_{m,k}} \wh{\bB}_{m,k}^\trans \bB_{m,k} \wh{\bH}_{m,k} \\
    &=
    \frac{1}{p_{m,k}} \wh{\bB}_{m,k}^\trans (\wh{\bB}_{m,k} - \bB_{m,k} \wh{\bH}_{m,k}) + \frac{1}{p_{m,k}} (\wh{\bB}_{m,k} - \bB_{m,k} \wh{\bH}_{m,k})^\trans \bB_{m,k} \wh{\bH}_{m,k} \\
    &\;\quad
    + \frac{1}{p_{m,k}} \wh{\bH}_{m,k}^\trans \bB_{m,k}^\trans \bB_{m,k} \wh{\bH}_{m,k}
    =
    \wh{\bH}_{m,k} \wh{\bH}_{m,k}^\trans + o_P(1) ,
\end{align*}
which completes part~(a) of the theorem.

{
Consider part~(b). Similar to part~(a), we also have $\wh{\bB}_{m,1} = \wt\bSigma_{B,m,1} \wh{\bB}_{m,1} \wt{\bD}_{B,m,1}^{-1}$.
Using Lemma~\ref{lemma: A_perp_rate} and the notations that $\bP_{A,m,1,\perp} = \bA_{m,1,\perp} \bA_{m,1,\perp}^\trans$ and $\wh\bP_{A,m,1,\perp} =
\wh\bA_{m,1,\perp}
\wh\bA_{m,1,\perp}^\trans$,
the leading term $\wt\cII_1$ in \eqref{eqn: decomp_tilde_Sigma_B} can be written as
\[
\wt\cII_1 =
\frac{1}{T p_{m,1}} \sum_{t=1}^{T} \wh\bP_{A,m,1,\perp} \bB_{m,1} \mat_1(\cF_t^{(m)}) \mat_1(\cF_t^{(m)})^\trans \bB_{m,1}^\trans \wh\bP_{A,m,1,\perp} ,
\]
whose nonzero eigenvalues are the same as those of
\[
\frac{1}{T p_{m,1}} \sum_{t=1}^{T} \bB_{m,1}^\trans \wh\bP_{A,m,1,\perp} \bB_{m,1} \mat_1(\cF_t^{(m)}) \mat_1(\cF_t^{(m)})^\trans .
\]
By similar arguments in \eqref{eqn: B_A_perp_A_perp_B}, 
Lemma~\ref{lemma: inequality_sandwich}(a), 
Lemma~\ref{lemma: A_perp_rate},
and Assumption~\ref{ass: loadings}(b), we have
\[
\frac{1}{p_{m,1}} \bB_{m,1}^\trans \wh\bP_{A,m,1,\perp} \bB_{m,1}
=
\frac{1}{p_{m,1}} \bB_{m,1}^\trans \bP_{A,m,1,\perp} \bB_{m,1}
+ o_P(1)
=
\bI_{u_{m,1}} + o_P(1) .
\]
Together with Assumption~\ref{ass: core_factor}(b), this implies that the nonzero eigenvalues of $\wt\cII_1$ converge to those of $\bSigma_{F,m,1}$, as shown in Lemma~\ref{lemma: eigenvalue_tilde_Sigma_B}.
The remaining steps of part~(b) can be shown by repeating the steps for part~(a), except that Lemma~\ref{lemma: eigenvalue_tilde_Sigma_B} is used in place of Lemma~\ref{lemma: eigenvalue_Sigma_B}, and for \eqref{eqn: cII_5_Bmk}, Lemma~4 in \cite{barigozzi2022statistical} can be directly applied. The proof of this theorem is now complete.
}

\end{proof}

\begin{proof}[Proof of Theorem~\ref{thm: core_factor}]
For the local factor estimator, consider $K_m> 1$ first. From \eqref{eqn: def_hat_Ft}, we have
\begin{align}
    &\;\quad
    \vec(\wh{\cF}_t^{(m)}) = (\wh{\bC}_m^\trans \wh{\bC}_m)^{-1} \wh{\bC}_m^\trans (\otimes_{k=1}^K \wh{\bA}_{m,k,\perp})^\trans \vec(\cX_t^{m}) \notag \\
    &=
    (\wh{\bC}_m^\trans \wh{\bC}_m)^{-1} \wh{\bC}_m^\trans \Big\{ (\otimes_{k=1}^K \wh{\bA}_{m,k,\perp})^\trans (\otimes_{k=1}^K \bB_{m,k}) \vec(\cF_t^{(m)}) \notag \\
    &\;\quad
    + (\otimes_{k=1}^K \wh{\bA}_{m,k,\perp})^\trans (\otimes_{k=1}^K \bA_{m,k}) \vec(\cG_t^{(m)}) + (\otimes_{k=1}^K \wh{\bA}_{m,k,\perp})^\trans \vec(\cE_t^{m}) \Big\} \notag \\
    &=
    (\otimes_{k=1}^K \wh{\bH}_{m,k}^\trans) \vec(\cF_t^{(m)}) \notag \\
    &\;\quad
    + (\wh{\bC}_m^\trans \wh{\bC}_m)^{-1} \wh{\bC}_m^\trans (\otimes_{k=1}^K \wh{\bA}_{m,k,\perp})^\trans \big( \otimes_{k=1}^K \bB_{m,k} - \otimes_{k=1}^K \wh{\bB}_{m,k} \wh{\bH}_{m,k}^\trans \big) \vec(\cF_t^{(m)}) \notag \\
    &\;\quad
    + (\wh{\bC}_m^\trans \wh{\bC}_m)^{-1} \wh{\bC}_m^\trans (\otimes_{k=1}^K \wh{\bA}_{m,k,\perp})^\trans (\otimes_{k=1}^K \bA_{m,k}) \vec(\cG_t^{(m)}) \notag \\
    &\;\quad
    + (\wh{\bC}_m^\trans \wh{\bC}_m)^{-1} \wh{\bC}_m^\trans (\otimes_{k=1}^K \wh{\bA}_{m,k,\perp})^\trans \vec(\cE_t^{m}) \notag \\
    &=: (\otimes_{k=1}^K \wh{\bH}_{m,k}^\trans) \vec(\cF_t^{(m)}) + \cIII_1 + \cIII_2 + \cIII_3 .
    \label{eqn: vec_F_decomp}
\end{align}

Before considering the last line in \eqref{eqn: vec_F_decomp}, first note that
\begin{align}
    &\;\quad
    \Big\| (\wh{\bC}_m^\trans \wh{\bC}_m)^{-1} \wh{\bC}_m^\trans \Big\|_F^2 \notag \\
    &\leq
    \Big\| \big( \otimes_{k=1}^K \wh{\bB}_{m,k}^\trans \wh{\bA}_{m,k,\perp} \wh{\bA}_{m,k,\perp}^\trans \wh{\bB}_{m,k} \big)^{-1} \Big\|_F^2 \cdot \big\| \otimes_{k=1}^K \wh{\bA}_{m,k,\perp} \big\|_F^2 \cdot \big\| \otimes_{k=1}^K \wh{\bB}_{m,k} \big\|_F^2 \notag \\
    &\asymp
    \big\| (p_m \bI)^{-1} \big\|_F^2 \cdot \big\| \otimes_{k=1}^K \wh{\bA}_{m,k,\perp} \big\|_F^2 \cdot \big\| \otimes_{k=1}^K \wh{\bB}_{m,k} \big\|_F^2
    = \cO_P\big( p_m^{-1} \big) ,
    \label{eqn: CCC_rate}
\end{align}
where the last line used
\eqref{eqn: B_A_perp_A_perp_B} and Assumption~\ref{ass: loadings}(b).

With \eqref{eqn: CCC_rate}, by Assumption~\ref{ass: core_factor}(a) and Theorem~\ref{thm: local_loading}, we immediately have (by a simple induction):
\begin{equation}
\label{eqn: cIII_1}
\begin{split}
    \| \cIII_1\|_F^2 &\lesssim \sum_{k\in[K_m]} \frac{1}{p_{m,k}} \Big\| \bB_{m,k} - \wh{\bB}_{m,k} \wh{\bH}_{m,k}^\trans \Big\|_F^2 \\
    &=
    \cO_P\Bigg\{ \sum_{k=1}^{K_m} \Big( \frac{1}{p_{m,k}^2} + \frac{1}{T p_{m,k}} + \frac{1}{T p_{m,\text{-}k}} \Big) \Bigg\} .
\end{split}
\end{equation}

{
For $\cIII_2$, using \eqref{eqn: CCC_rate} 
and Assumption~\ref{ass: core_factor}(a), we have
\begin{align}
    \|\cIII_2\|_F^2 &\lesssim
    \frac{1}{p_m} \Big\|
    \big( \otimes_{j=1}^{K_m} \wh\bA_{m,j,\perp} \big)^\trans
    \big( \otimes_{j=1}^{K_m} \bA_{m,j} \big) \Big\|_F^2 \notag \\
    &\lesssim
    \frac{1}{p_m}
    \Bigg( \sum_{j=1}^{K_m} \Big\| \wh\bP_{A,m,j,\perp} - \bP_{A,m,j,\perp} \Big\|_F \big\| \bA_{m,j} \big\| \Bigg)^2 
    =
    \cO_P\Big( \frac{1}{T} \Big) ,
    \label{eqn: cIII_2}
\end{align}
where the second inequality used the fact that
$\wh\bA_{m,j,\perp}^\trans \bA_{m,j} = \wh\bA_{m,j,\perp}^\trans ( \wh\bP_{A,m,j,\perp} - \bP_{A,m,j,\perp} ) \bA_{m,j}$,
and the final equality used Lemma~\ref{lemma: A_perp_rate} and Assumption~\ref{ass: loadings}(b).
}

For $\cIII_3$, first consider $\wh{\bC}_m^\trans (\otimes_{k=1}^K \wh{\bA}_{m,k,\perp})^\trans \vec(\cE_t^{m})$. By Lemma~\ref{lemma: inequality_sandwich}(b) and arguments similar to \eqref{eqn: B_A_perp_A_perp_B}, we have
\begin{align*}
    &\;\quad
    \Big\| \wh{\bC}_m^\trans (\otimes_{k=1}^K \wh{\bA}_{m,k,\perp})^\trans \vec(\cE_t^{m}) \Big\|_F^2 \\
    &\lesssim \max_i \big\| (\otimes_{k=1}^K \bB_{m,k})_{i\cdot} \big\|_F^2 \sum_{j,l=1}^{p_m} \big| \E\big\{ \vec(\cE_t^{m})_j \vec(\cE_t^{m})_l \big\} \Big|
    = \cO_P(p_m) ,
\end{align*}
where the last equality used Assumptions~\ref{ass: loadings}(a) and \ref{ass: noise}(b). Noting that $\big\| (\wh{\bC}_m^\trans \wh{\bC}_m)^{-1} \big\|_F^2 =\cO_P\big( p_m^{-2} \big)$ from \eqref{eqn: CCC_rate}, it hence holds that
\begin{equation}
\label{eqn: cIII_3}
\| \cIII_3\|_F^2 = \cO_P\Big( \frac{1}{p_m} \Big).
\end{equation}

Therefore, combining \eqref{eqn: vec_F_decomp}, \eqref{eqn: cIII_1}, \eqref{eqn: cIII_2}, and \eqref{eqn: cIII_3}, we conclude
\[
\Big\| \vec(\wh{\cF}_t^{(m)}) - (\otimes_{k=1}^K \wh{\bH}_{m,k}^\trans) \vec(\cF_t^{(m)}) \Big\|_F^2 = \cO_P\Big( \sum_{k\in[K_m]} \frac{1}{p_{m,k}^2} + \frac{1}{T} + \frac{1}{p_m} \Big) ,
\]
which completes the proof for the scenario $K_m> 1$. When $K_m=1$, all the previous arguments follow, except that for \eqref{eqn: cIII_1}, we have
\[
\| \cIII_1\|_F^2 = \cO_P\Big( \frac{1}{p_{m,1}^2} + \frac{1}{T} \Big) ,
\]
which leads to essentially the same expression (note that $p_m=p_{m,1}$ for $K_m=1$, and $\wh{\bH}_{m,k}$ is replaced by $\wt{\bH}_{m,1}$), i.e.,
\[
\Big\| \vec(\wh{\cF}_t^{(m)}) - \wt{\bH}_{m,1}^\trans \vec(\cF_t^{(m)}) \Big\|_F^2 = \cO_P\Big( \frac{1}{T} + \frac{1}{p_{m,1}} \Big) .
\]
This completes the proof for part~(a) of the theorem.

Next, we first show the consistency of the vectorized local component as follows. For $K_m> 1$,
\begin{align}
    &\;\quad
    \frac{1}{p_m} \Big\| (\otimes_{k=1}^K \bB_{m,k}) \vec(\cF_t^{(m)}) - (\otimes_{k=1}^K \wh{\bB}_{m,k}) \vec(\wh\cF_t^{(m)}) \Big\|_F^2 \notag \\
    &=
    \frac{1}{p_m} \Big\| (\otimes_{k=1}^K \bB_{m,k} \wh{\bH}_{m,k}) (\otimes \wh{\bH}_{m,k}^\trans) \vec(\cF_t^{(m)}) - (\otimes_{k=1}^K \wh{\bB}_{m,k}) \vec(\wh\cF_t^{(m)}) \Big\|_F^2 \notag \\
    &\lesssim
    \frac{1}{p_m} \Big\| (\otimes_{k=1}^K \bB_{m,k} \wh{\bH}_{m,k}) - (\otimes_{k=1}^K \wh{\bB}_{m,k}) \Big\|_F^2
    + \Big\| (\otimes \wh{\bH}_{m,k}^\trans) \vec(\cF_t^{(m)}) - \vec(\wh\cF_t^{(m)}) \Big\|_F^2 \notag \\
    &=
    \cO_P\Big( \sum_{k\in[K_m]} \frac{1}{p_{m,k}^2} + \frac{1}{T} + \frac{1}{p_m} \Big) ,
    \label{eqn: rate_local_component_Km_not1}
\end{align}
by Theorems~\ref{thm: local_loading} and \ref{thm: core_factor}(a). Similarly, for $K_m=1$,
\begin{align}
    &\;\quad
    \frac{1}{p_m} \Big\| (\otimes_{k=1}^K \bB_{m,k}) \vec(\cF_t^{(m)}) - (\otimes_{k=1}^K \wh{\bB}_{m,k}) \vec(\wh\cF_t^{(m)}) \Big\|_F^2 \notag \\
    &=
    \cO_P\Big( \frac{1}{T} + \frac{1}{p_{m,1}} \Big) .
    \label{eqn: rate_local_component_Km_1}
\end{align}

Define $\bH_{A,m} := \otimes_{k=1}^{K_m} \bH_{A,m,k}$
and $\wt\bA_{m,k} := \bA_{m,k} \bH_{A,m,k}$.
Since each $\bH_{A,m,k}$ is exactly orthogonal, so is $\bH_{A,m}$.
Then for part~(b), by the definition in \eqref{eqn: def_hat_Gt}, we may write
\begin{align*}
    &\;\quad
    \vec(\wh{\cG}_t^{(m)}) \\
    &=
    \Big\{ (\otimes_{k=1}^K \wh{\bA}_{m,k})^\trans (\otimes_{k=1}^K \wh{\bA}_{m,k}) \Big\}^{-1} (\otimes_{k=1}^K \wh{\bA}_{m,k})^\trans \Big\{ \vec(\cX_t^{m}) - (\otimes_{k=1}^K \wh{\bB}_{m,k}) \vec(\wh\cF_t^{(m)}) \Big\} \\
    &=
    \Big\{ (\otimes_{k=1}^K \wh{\bA}_{m,k})^\trans (\otimes_{k=1}^K \wh{\bA}_{m,k}) \Big\}^{-1} (\otimes_{k=1}^K \wh{\bA}_{m,k})^\trans \Big\{ \vec(\cX_t^{m}) - (\otimes_{k=1}^K \bB_{m,k}) \vec(\cF_t^{(m)}) \Big\} \\
    &\;\quad
    + \Big\{ (\otimes_{k=1}^K \wh{\bA}_{m,k})^\trans (\otimes_{k=1}^K \wh{\bA}_{m,k}) \Big\}^{-1} (\otimes_{k=1}^K \wh{\bA}_{m,k})^\trans \\
    &\;\quad
    \cdot \Big\{ (\otimes_{k=1}^K \bB_{m,k}) \vec(\cF_t^{(m)}) - (\otimes_{k=1}^K \wh{\bB}_{m,k}) \vec(\wh\cF_t^{(m)}) \Big\} \\
    &=
    \Big\{ (\otimes_{k=1}^K \wh{\bA}_{m,k})^\trans (\otimes_{k=1}^K \wh{\bA}_{m,k}) \Big\}^{-1} (\otimes_{k=1}^K \wh{\bA}_{m,k})^\trans \Big\{ (\otimes_{k=1}^K \bA_{m,k}) \vec(\cG_t^{(m)}) + \vec(\cE_t^{m}) \Big\} \\
    &\;\quad
    + \Big\{ (\otimes_{k=1}^K \wh{\bA}_{m,k})^\trans (\otimes_{k=1}^K \wh{\bA}_{m,k}) \Big\}^{-1} (\otimes_{k=1}^K \wh{\bA}_{m,k})^\trans \\
    &\;\quad
    \cdot \Big\{ (\otimes_{k=1}^K \bB_{m,k}) \vec(\cF_t^{(m)}) - (\otimes_{k=1}^K \wh{\bB}_{m,k}) \vec(\wh\cF_t^{(m)}) \Big\} \\
    &=
    \bH_{A,m}^\trans \vec(\cG_t^{(m)}) 
    + \Big\{ (\otimes_{k=1}^K \wh{\bA}_{m,k})^\trans (\otimes_{k=1}^K \wh{\bA}_{m,k}) \Big\}^{-1} \\
    &\;\quad
    \cdot (\otimes_{k=1}^K \wh{\bA}_{m,k})^\trans \Big\{ (\otimes_{k=1}^K \bA_{m,k}) -(\otimes_{k=1}^K \wh{\bA}_{m,k}) \bH_{A,m}^\trans \Big\} \vec(\cG_t^{(m)}) \\
    &\;\quad
    + \Big\{ (\otimes_{k=1}^K \wh{\bA}_{m,k})^\trans (\otimes_{k=1}^K \wh{\bA}_{m,k}) \Big\}^{-1} (\otimes_{k=1}^K \wh{\bA}_{m,k})^\trans \\
    &\;\quad
    \cdot \Big\{ (\otimes_{k=1}^K \bB_{m,k}) \vec(\cF_t^{(m)}) - (\otimes_{k=1}^K \wh{\bB}_{m,k}) \vec(\wh\cF_t^{(m)}) \Big\} \\
    &\;\quad
    + \Big\{ (\otimes_{k=1}^K \wh{\bA}_{m,k})^\trans (\otimes_{k=1}^K \wh{\bA}_{m,k}) \Big\}^{-1} (\otimes_{k=1}^K \wh{\bA}_{m,k})^\trans \vec(\cE_t^{m}) \\
    &=: 
    \bH_{A,m}^\trans \vec(\cG_t^{(m)}) + \cIV_1 + \cIV_2 + \cIV_3 .
\end{align*}

The terms $\cIV_1$--$\cIV_3$ are comparable to $\cIII_1$--$\cIII_3$ in \eqref{eqn: vec_F_decomp}. Indeed, combining Theorem~\ref{thm: global_loading}, \eqref{eqn: rate_local_component_Km_not1}, \eqref{eqn: rate_local_component_Km_1}, and analogous argument as for $\cIII_3$ (see details above \eqref{eqn: cIII_3}), it holds similarly as the previous proof for part~(a) that
\begin{align*}
    \|\cIV_1\|_F^2 &\lesssim
    \sum_{k=1}^{K_m} \frac{1}{p_{m,k}} \Big\| \wh{\bA}_{m,k} -\bA_{m,k} \bH_{A,m,k}  \Big\|_F^2 = \cO_P\Big( \frac{1}{T} \Big) ; \\
    \|\cIV_2\|_F^2 &\lesssim \frac{1}{p_m} \Big\| (\otimes_{k=1}^K \bB_{m,k}) \vec(\cF_t^{(m)}) - (\otimes_{k=1}^K \wh{\bB}_{m,k}) \vec(\wh\cF_t^{(m)}) \Big\|_F^2 \\
    &= \begin{cases}
        \cO_P\Big( \sum_{k\in[K_m]} \frac{1}{p_{m,k}^2} + \frac{1}{T} + \frac{1}{p_m} \Big) & \text{if $K_m> 1$,} \\
        \cO_P\Big( \frac{1}{T} + \frac{1}{p_{m,1}} \Big) & \text{if $K_m= 1$;}
    \end{cases} \\
    \|\cIV_3\|_F^2 &= \cO_P\Big( \frac{1}{p_m} \Big) .
\end{align*}

Putting all together, we hence conclude (in squared Frobenius norm, which is of the same order as in the squared spectral norm result presented in the theorem):
\[
\Big\| \vec(\wh{\cG}_t^{(m)}) - \bH_{A,m}^\trans \vec(\cG_t^{(m)}) \Big\|_F^2 =
\begin{cases}
    \cO_P\Big( \sum_{k\in[K_m]} \frac{1}{p_{m,k}^2} + \frac{1}{T} + \frac{1}{p_m} \Big) & \text{if $K_m> 1$,} \\
    \cO_P\Big( \frac{1}{T} + \frac{1}{p_{m,1}} \Big) & \text{if $K_m= 1$.}
\end{cases}
\]

The rate of convergence for the local component estimator has been spelled out in \eqref{eqn: rate_local_component_Km_not1} and \eqref{eqn: rate_local_component_Km_1}. The arguments for the global component estimator follow similarly and hence omitted here. This ends the proof of the theorem.
\end{proof}

\begin{proof}[Proof of Theorem~\ref{thm: eigenratio}]
First, let us only consider the case where the global component is present.
Fix any $m\in[M]$ and $k\in[K_m]$. By Weyl's inequality and Lemma~\ref{lemma: rate_Sigma}, we have, for any $j\in[p_{m,k}]$,
\begin{align*}
    \Big| \lambda_j(\wh\bSigma_{m,k}) -\lambda_j(\bSigma_{m,k}) \Big| \leq \Big\| \wh\bSigma_{m,k} -\bSigma_{m,k} \Big\|
    = \cO_P\Bigg\{ \frac{p_m}{\sqrt{T}} \Bigg( \sum_{\substack{n=1\\ n\neq m}}^M p_n \Bigg) \Bigg\}
    = \cO_P(\xi_{m,k}) ,
\end{align*}
where the last equality used the definition of $\xi_{m,k}$.

Consider the scenario of $r_{m,k}>1$. For any $j\in[r_{m,k} -1]$, we hence have
\begin{equation}
\label{eqn: eigenratio_proof_1}
\frac{\lambda_{j+1}(\wh\bSigma_{m,k}) +\xi_{m,k}}{\lambda_j(\wh\bSigma_{m,k}) +\xi_{m,k}} 
\geq \frac{\lambda_{j+1}(\bSigma_{m,k}) +\cO_P(\xi_{m,k})}{\lambda_j(\bSigma_{m,k}) +\cO_P(\xi_{m,k})}
\asymp c ,
\end{equation}
for some positive constant $c$, where the asymptotic (in probability) equality used Lemma~\ref{lemma: eigenvalue_Sigma} and the definition of $\xi_{m,k}$. The above is unnecessary to prove this theorem when $r_{m,k}=1$.

Next, note that the rank of $\bSigma_{m,k}$ is upper bounded by $r_{m,k}$ due to the definition of $\bSigma_{m,k}$ where $\bA_{m,k}\in \R^{p_{m,k}\times r_{m,k}}$ within. We then have
\begin{equation}
\label{eqn: eigenratio_proof_2}
\begin{split}
    &\;\quad
    \frac{\lambda_{r_{m,k}+1}(\wh\bSigma_{m,k}) +\xi_{m,k}}{\lambda_{r_{m,k}}(\wh\bSigma_{m,k}) +\xi_{m,k}} \leq \frac{\lambda_{r_{m,k}+1}(\bSigma_{m,k}) +\cO_P(\xi_{m,k})}{\lambda_{r_{m,k}}(\bSigma_{m,k}) +\cO_P(\xi_{m,k})} \\
    &=
    \frac{\cO_P(\xi_{m,k})}{\lambda_{r_{m,k}}(\bSigma_{m,k}) +\cO_P(\xi_{m,k})}
    = \cO_P\Big( \frac{1}{\sqrt{T}} \Big) ,
\end{split}
\end{equation}
where the last equality used again Lemma~\ref{lemma: eigenvalue_Sigma} and the definition of $\xi_{m,k}$.

It remains to consider the eigenvalue ratio for $j\in \{r_{m,k}+1, \dots, r_{\max,m,k}\}$. In detail, we have
\begin{equation}
\label{eqn: eigenratio_proof_3}
\frac{\lambda_{j+1}(\wh\bSigma_{m,k}) +\xi_{m,k}}{\lambda_j(\wh\bSigma_{m,k}) +\xi_{m,k}}
\geq \frac{\xi_{m,k}}{\lambda_j(\wh\bSigma_{m,k}) +\xi_{m,k}}
\geq \frac{\xi_{m,k}}{\cO_P(\xi_{m,k}) +\xi_{m,k}}
\geq 1/c ,
\end{equation}
in probability for some positive constant $c$, where we also used the definition of $\xi_{m,k}$. Lastly, combining \eqref{eqn: eigenratio_proof_1}, \eqref{eqn: eigenratio_proof_2}, and \eqref{eqn: eigenratio_proof_3}, we conclude the proof of the part when there are global components.

Next, consider the scenario when there can be no global components, namely, $r_{m,k}\geq 0$, and the ratio $\big\{\lambda_{1}(\wh\bSigma_{m,k}) + \xi_{m,k} \big\}/ \big\{ \lambda_0(\wh\bSigma_{m,k}) +\xi_{m,k} \big\}$ should be taken into account, where $\lambda_0 (\wh\bSigma_{m,k})$ is defined in \eqref{eqn: mock_evalue}. We first show that when the true $r_{m,k}>0$, the rank estimator is consistent. First, note that all the above arguments remain to hold, regardless of the absence of global components in thread $n\neq m$ (using the dimensionality condition in the statement of this theorem). By definition,
\[
\lambda_0(\wh\bSigma_{m,k}) \asymp \Big[ \lambda_1(\wh\bSigma_{m,k}) + \lambda_{p_{m,k}}(\wh\bSigma_{m,k}) \cdot \min\{p_{m,k}, T\} \Big] / \log(\min\{p_{m,k}, T\}) .
\]
Thus, for some positive constant $c$, we have
\begin{align*}
    &\;\quad
    \frac{\lambda_{1}(\wh\bSigma_{m,k}) + \xi_{m,k}}{\lambda_0(\wh\bSigma_{m,k}) +\xi_{m,k}} \geq  \frac{\log(\min\{p_{m,k}, T\}) \big\{ \lambda_{1}(\bSigma_{m,k}) + \cO_P(\xi_{m,k})\big\}}{\lambda_1(\bSigma_{m,k}) + \big[ \min\{p_{m,k}, T\} + \log(\min\{p_{m,k}, T\}) \big] \cO_P(\xi_{m,k})} \\
    &\geq
    \frac{c \cdot \log(\min\{p_{m,k}, T\}) \xi_{m,k} \sqrt{T}}{\xi_{m,k} \sqrt{T} + \xi_{m,k} T}
    \asymp
    \frac{\log(\min\{p_{m,k}, T\})}{\sqrt{T}} ,
\end{align*}
which implies $\wh{r}_{m,k}$ is consistent, when compared to \eqref{eqn: eigenratio_proof_2}.

It then remains to show if $r_{m,k}=0$, the rank estimator is consistent such that $\wh{r}_{m,k}=0$. In this case, since there is no global component, only $\cI_4$, $\cI_7$, $\cI_8$, and $\cI_9$ are present in the decomposition \eqref{eqn: decomp_Omega}, and hence we now have from the proof of Lemma~\ref{lemma: rate_Omega} that
\[
\sum_{\bi \in [p_{n,1}] \times \cdots \times[p_{n,K_n}]} \Big\| \wh\bOmega_{k,n,\bi}^{(m)} \Big\|_F^2 = \cO_P\Big( \frac{p_m p_n}{T} \Big) .
\]
Together with the proof of Lemma~\ref{lemma: rate_Sigma}, for any $j\in[p_{m,k}]$,
\[
\lambda_{j}(\wh\bSigma_{m,k}) \leq \big\| \wh\bSigma_{m,k} \big\|_F = \cO_P\Bigg\{ \frac{p_m}{T} \Bigg( \sum_{\substack{n=1\\ n\neq m}}^M p_n \Bigg) \Bigg\} .
\]

Then for $j\in[r_{\max,m,k}]$, the ratio $\big\{\lambda_{j+1}(\wh\bSigma_{m,k}) + \xi_{m,k} \big\}/ \big\{ \lambda_j(\wh\bSigma_{m,k}) +\xi_{m,k} \big\}$ behaves similarly as in \eqref{eqn: eigenratio_proof_3}. Lastly, consider the ratio
\begin{align*}
    \frac{\lambda_{1}(\wh\bSigma_{m,k}) + \xi_{m,k}}{\lambda_0(\wh\bSigma_{m,k}) +\xi_{m,k}} &\lesssim
    \frac{\xi_{m,k} \log(\min\{p_{m,k}, T\})}{\xi_{m,k} \log(\min\{p_{m,k}, T\}) + \xi_{m,k} \min\{p_{m,k}, T\} /\sqrt{\min\{p_{m,k}, T\}}} \\
    &\leq
    \frac{\log(\min\{p_{m,k}, T\})}{\sqrt{\min\{p_{m,k}, T\}}} ,
\end{align*}
which thereby completes the proof of the theorem.
\end{proof}

\begin{proof}[Proof of Theorem~\ref{prop: global_loading_bootstrap}]
The proof is direct by following Lemmas~\ref{lemma: rate_Omega}--\ref{lemma: eigenvalue_Sigma} similarly (with the rate $p_n$ being $|\cS_n^\dagger|$ now) and hence Theorem~\ref{thm: global_loading}, except that in \eqref{eqn: eigenvalue_Sigma_step_2} in the proof of Lemma~\ref{lemma: eigenvalue_Sigma}:
\begin{align*}
    &\;\quad
    \lambda_{r_{m,k}}\Bigg( \sum_{\bi \in \cS_n} \E\Big\{ \mat_k(\cG_t^{(m)}) \Big( \cG_t^{(n)} \times_{h=1}^{K_n} \bA_{n,h} \Big)_{\bi} \Big\} \\
    &\;\quad
    \cdot \E\Big\{ \mat_k(\cG_t^{(m)}) \Big( \cG_t^{(n)} \times_{h=1}^{K_n} \bA_{n,h} \Big)_{\bi} \Big\}^\trans \Bigg) \notag \\
    &=
    \lambda_{r_{m,k}}\Bigg( \E\Big[ \mat_k(\cG_t^{(m)}) \otimes \Big\{ \vec(\cG_t^{(n)})^\trans \big( \otimes_{h=1}^{K_n} \bA_{n,h} \big)_{\cS_n^\dagger}^\trans \Big\} \Big] \notag \\
    &\;\quad
    \cdot \E\Big[ \mat_k(\cG_t^{(m)})^\trans \otimes \Big\{ \big( \otimes_{h=1}^{K_n} \bA_{n,h} \big)_{\cS_n^\dagger} \vec(\cG_t^{(n)}) \Big\} \Big] \Bigg) 
    \asymp
    \cS_n^\dagger \cdot c ,
\end{align*}
where the last line used Assumption~\ref{ass: core_factor}(c) and the additional condition in the statement of Theorem~\ref{prop: global_loading_bootstrap}. This concludes the proof of the theorem.
\end{proof}

\subsection{Proof of claims}\label{subsec: proof_claim}

\begin{proof}[Proof of Claim~\ref{claim: assump_fac_simpler_asymp_ortho}]
Note that elements of the global factors have zero mean according to Assumption~\ref{ass: identification}(a). To prove the claim when all tensor orders are the same, it suffices to consider the case when global factors have elements with unit variance, since otherwise the variance can be absorbed into the global loading matrices by the Kronecker product structure in the covariance matrix of the vectorized global factor.

For simplicity, we use the following notations throughout this proof.
\begin{align*}
    & \cG_t := \cG_t^{(1)} =\dots =\cG_t^{(M)}, \\
    & K:= K_1=\dots =K_M, \\
    & r_{h} := r_{1,h}=\dots =r_{M,h} \quad \text{for each $h\in[K]$} , \\
    & r_{\text{-}k} := r_1 \dots r_{k-1} r_{k+1} \dots r_K .
\end{align*}
where we fix $k\in[K]$ and $n\in[K] \setminus\{m\}$.
For $h\in[K]$, also define the normalized loading matrices $\wt\bA_h := \bA_h/ \sqrt{p_h}$.
Note that $\wt\bA_h^\trans \wt\bA_h \to \bI_{r_h}$ by Assumption~\ref{ass: loadings_part_b_asymp}.

Further define
\begin{equation}
\label{eqn: def_Mt}
\begin{split}
    \bM_t &:= \mat_k(\cG_t^{(m)}) \otimes \Big\{ \vec(\cG_t^{(n)})^\trans \big( \otimes_{h=1}^{K_n} \wt\bA_{h} \big)^\trans \Big\} \\
    &=
    \mat_k(\cG_t) \otimes \Big\{ \vec(\cG_t)^\trans \big( \otimes_{h=1}^{K} \wt\bA_{h} \big)^\trans \Big\} , 
\end{split}
\end{equation}
which is a matrix of dimension $r_k \times (r_{\text{-}k} p_n)$. We next write the expression in index form.
Let the entries of $\cG_t$ be denoted by $g_{t,i_1, \dots, i_K}$ with $i_h \in [r_h]$.
Denote the entries of the mode-$k$ unfolding of $\cG_t$ as
\[
\mat_k(\cG_t)_{i_k, \balpha_k} = g_{t,i_1, \dots, i_K},
\]
where the notion $\balpha_k$ corresponds to the multi-index $(i_1, \dots, i_{k-1}, i_{k+1}, \dots, i_K)$. 
For the row vector $\vec(\cG_t^{(n)})^\trans \big( \otimes_{h=1}^{K_n} \wt\bA_{n,h} \big)^\trans$ which has length $p_n=p_{n,1}\dots p_{n,K}$, let $\bbeta = (j_1, \dots, j_K)$ denote its multi-index (corresponding to each $\wt\bA_{n,h}$). Then we can write
\[
\Big\{ \vec(\cG_t^{(n)})^\trans \big( \otimes_{h=1}^{K_n} \wt\bA_{n,h} \big)^\trans \Big\}_{\bbeta} = \sum_{i'_1, \dots, i'_K} \left( \prod_{h=1}^K (\wt\bA_h)_{j_h, i'_h} \right) g_{t,i'_1, \dots, i'_K}.
\]

Hence, the $(i_k, (\balpha_k, \bbeta))$ entry of $\bM_t$ is
\[
\bM_t(i_k, (\balpha_k, \bbeta)) = g_{t,i_1, \dots, i_K} \cdot \left\{ \sum_{i'_1, \dots, i'_K} \left( \prod_{h=1}^K (\wt\bA_h)_{j_h, i'_h} \right) g_{t,i'_1, \dots, i'_K} \right\},
\]
where $i_k$ determines the row index and $\balpha_k$ determines the other indices.

Since the entries of $\cG_t$ are uncorrelated with mean 0 and variance 1, we have
\[
\E(g_{t,i_1, \dots, i_K} \cdot g_{t,i'_1, \dots, i'_K}) = 
\begin{cases}
1 & \text{if } (i_1, \dots, i_K) = (i'_1, \dots, i'_K), \\
0 & \text{otherwise}.
\end{cases}
\]
Therefore,
\begin{equation}
\label{eqn: Mt_multi_index_simplify}
\E[\bM_t(i_k, (\balpha_k, \bbeta))] = \prod_{h=1}^K (\wt\bA_h)_{j_h, i_h}
= (\wt\bA_1)_{j_1, i_1} (\wt\bA_2)_{j_2, i_2} \dots (\wt\bA_K)_{j_K, i_K} .
\end{equation}

In what follows, we recognize a Kronecker product structure in $\bM_t$. To this end, rewrite the product in \eqref{eqn: Mt_multi_index_simplify} as
\begin{equation}
\label{eqn: Mt_multi_index_simplify_2}
\E[\bM_t(i_k, (\balpha_k, \bbeta))] = (\wt\bA_k)_{j_k, i_k} \cdot \left\{ \prod_{h \neq k} (\wt\bA_h)_{j_h, i_h} \right\}.
\end{equation}

For each $h \neq k$, denote the $i_h$th column of $\wt\bA_h$ as $\wt\bA_{h,\cdot i_h}$. Then the product over $h \neq k$ in \eqref{eqn: Mt_multi_index_simplify_2} is the $(j_1, \dots, j_{k-1}, j_{k+1}, \dots, j_K)$ entry of
\[
\wt\bA_{1,\cdot i_1} \otimes \wt\bA_{2,\cdot i_2} \otimes \dots \otimes \wt\bA_{k-1,\cdot i_{k-1}} \otimes \wt\bA_{k+1,\cdot i_{k+1}} \otimes \dots \otimes \wt\bA_{K,\cdot i_K} .
\]

Define $\wt\bA_{\text{-}k} := \otimes_{h\in[K] \setminus \{k\}} \wt\bA_h$. Then the element of $\wt\bA_{\text{-}k}$ corresponding to row (indexed by multi-index) $(j_1, \dots, j_{k-1}, j_{k+1}, \dots, j_K)$ and column $\balpha_k$ is 
\[
\wt\bA_{\text{-}k}(\bbeta, \balpha_k) = \prod_{h \neq k} (\wt\bA_h)_{j_h, i_h} .
\]
Thus by \eqref{eqn: Mt_multi_index_simplify_2}, we have
\begin{align*}
    \E[\bM_t(i_k, (\balpha_k, \bbeta))] &= (\wt\bA_k)_{j_k, i_k} \cdot \wt\bA_{\text{-}k}(\bbeta, \balpha_k)
    = (\wt\bA_k^\trans)_{i_k, j_k} \cdot \wt\bA_{\text{-}k}(\bbeta, \balpha_k) ,
\end{align*}
which implies $\E[\bM_t(:, (\balpha_k, \bbeta))] = (\wt\bA_k^\trans)_{\cdot j_k} \cdot \wt\bA_{\text{-}k}(\bbeta, \balpha_k)$ and hence
\begin{align*}
    \E[\bM_t(:, (\balpha_k, :))] &= \wt\bA_k^\trans \otimes \{ \wt\bA_{\text{-}k}(:, \balpha_k) \}^\trans 
    = \wt\bA_k^\trans \otimes (\wt\bA_{\text{-}k})_{\cdot \balpha_k}^\trans .
\end{align*}

Finally, $\E[\bM_t]$ is the row block vector by concatenating all $\E[\bM_t(:, (\balpha_k, :))]$ in order. Thus, we conclude
\begin{equation}
\label{eqn: Mt_expectation}
\E[\bM_t] = \wt\bA_k^\trans \otimes \vec(\wt\bA_{\text{-}k})^\trans
= \wt\bA_k^\trans \otimes \vec\big( \otimes_{h\in[K] \setminus \{k\}} \wt\bA_h \big)^\trans .
\end{equation}

Then it is easy to see that for any $j\in[r_k]$,
\begin{align*}
    &\;\quad
    \sigma_{j}^2\big( \E[\bM_t] \big) = \lambda_{j}\big( \E[\bM_t] \E[\bM_t]^\trans \big) \\
    &=
    \lambda_{j}\Big\{ \wt\bA_k^\trans \wt\bA_k \otimes \vec\big( \otimes_{h\in[K] \setminus \{k\}} \wt\bA_h \big)^\trans \vec\big( \otimes_{h\in[K] \setminus \{k\}} \wt\bA_h \big) \Big\}
    \asymp c ,
\end{align*}
for some positive constant $c$, where the last line used the fact that $\wt\bA_k^\trans \wt\bA_k$ has all first $r_k$ eigenvalues bounded away from zero and infinity asymptotically, and
\begin{align*}
    &\;\quad
    \vec\big( \otimes_{h\in[K] \setminus \{k\}} \wt\bA_h \big)^\trans \vec\big( \otimes_{h\in[K] \setminus \{k\}} \wt\bA_h \big) \Big\} \\
    &=
    \Big\| \vec\big( \otimes_{h\in[K] \setminus \{k\}} \wt\bA_h \big) \Big\|^2
    = \Big\| \otimes_{h\in[K] \setminus \{k\}} \wt\bA_h \Big\|_F^2 \asymp c .
\end{align*}
This concludes the proof of the case $K_1=\dots=K_M$.

Next, when $\cG_t^{(m)} =\map(\cG_t, V_m)$ using the notations in \eqref{eqn: model_map_Gt}, we may still assume without loss of generality that every elements in the global factors have unit variance, since each $\bSigma_{G,k}$ in the statement of this claim can be absorbed into $\otimes_{h=1}^{K_n} \wt\bA_{h}$ analogously to the previous case, according to Theorem~1 in \cite{CenLam2025_KronProd}. In fact, the remaining proof is identical to all the previous arguments, by noting that $\vec(\cG_t^{(n)})^\trans \big( \otimes_{h=1}^{K_n} \wt\bA_{h} \big)^\trans$ in the first line of \eqref{eqn: def_Mt} can be re-organized such that 
\[
\vec(\map(\cG_t, V_n))^\trans \big( \otimes_{h=1}^{K_n} \wt\bA_{h} \big)^\trans = \vec(\map(\cG_t, V_m))^\trans \big( \otimes_{h=1}^{K_m} \wt\bA_{h}^\ast \big)^\trans
\]
for some matrices $\wt\bA_{h}^\ast$, $h\in[K_m]$. This can be done by iteratively applying Theorem~1 in \cite{CenLam2025_KronProd}, and the Kronecker product structure is preserved by the tensor reshape operator. More importantly, in their Theorem~1, the explicit form of the reshaped loading matrices implies that $\wt\bA_{k}^\ast$ is either the exact or a sub-matrix (by taking several columns only) of a Kronecker product of some matrices in $\{\wt\bA_{h}\}_{h\in[K_n]}$. Either case, the full column rank in $\{\wt\bA_{h}\}_{h\in[K_n]}$ can be translated to $\wt\bA_{k}^\ast$, and the result follows.
The proof of the claim is thus complete.
\end{proof}

\begin{proof}[Proof of Claim~\ref{claim: assump_fac_simpler}]
Note that by Claim~\ref{claim: assump_fac_simpler_asymp_ortho} in Section~\ref{subsec: discussion_core_fac_part_c}, the stated conditions imply Assumption~\ref{ass: core_factor_part_c_general}. Under Assumption~\ref{ass: loadings}(b), Lemma~\ref{lemma: eigenvalue_Sigma_aux} shows that Assumption~\ref{ass: core_factor_part_c_general} is equivalent to Assumption~\ref{ass: core_factor}(c). 
Therefore Claim~\ref{claim: assump_fac_simpler} follows directly.
\end{proof}

\end{document}